\documentclass[12pt]{article}
\usepackage{amsmath,amssymb,amsthm, amstext} 
\usepackage{bbm}
\usepackage{stmaryrd}
\usepackage{commath}
\usepackage[utf8]{inputenc}
\usepackage{enumitem}
\usepackage{dsfont}
\usepackage{natbib}
\usepackage{mathrsfs} 
\usepackage{mathtools}
\mathtoolsset{showonlyrefs}
\usepackage[ruled]{algorithm2e}
\usepackage{appendix}
\usepackage{float}
\usepackage{booktabs}

\newenvironment{policy}[1][htb]
{%
\begin{algorithm}[#1]%
}{\end{algorithm}}

\newenvironment{simulation}[1][htb]
{%
\begin{algorithm}[#1]%
}{\end{algorithm}}

\usepackage{mathrsfs} 
\usepackage[colorlinks, pdfstartview=FitH, citecolor = black, linkcolor  = black]{hyperref}

\usepackage{graphicx}
\usepackage{setspace}
\setdisplayskipstretch{}

\renewcommand{\P}{\mathbb{P}}

\renewcommand{\rho}{\varrho}

\newcommand\eps{\varepsilon}

\newcommand\p{\mathbb{P}}
\newcommand\N{\mathbb{N}}
\newcommand\E{\mathbb{E}}
\newcommand\R{\mathbb{R}}

\DeclareMathOperator*{\argmax}{arg\,max}

\newtheorem{theorem}{Theorem}[section]  
\newtheorem{assumption}[theorem]{Assumption}

\newtheorem{definition}[theorem]{Definition}
\newtheorem{lemma}[theorem]{Lemma}

\theoremstyle{definition}
\newtheorem{example}[theorem]{Example}
\newtheorem{remark}[theorem]{Remark}

\begin{document}
\title{Functional Sequential Treatment Allocation\footnote{
		We are grateful to the Editor, the Associate Editor, and two referees for their comments and suggestions that helped to substantially improve the manuscript.		
	}
}
\author{
\begin{tabular}{c}
Anders Bredahl Kock \\ 
\small	University of Oxford \\
\small	CREATES, Aarhus University\\
%\small	10 Manor Rd, Oxford OX1 3UQ
%\\
\small	{\small	\href{mailto:anders.kock@economics.ox.ac.uk}{anders.kock@economics.ox.ac.uk}} 
\end{tabular}
\and
\begin{tabular}{c}
David Preinerstorfer \\ 
{\small	ECARES, SBS-EM} \\ 
{\small	 Universit\'e libre de Bruxelles} \\
%{\small	50 Ave.~F.D.~Roosevelt CP 114/04, 1050 Brussels} \\
{\small	 \href{mailto:david.preinerstorfer@ulb.ac.be}{david.preinerstorfer@ulb.ac.be}}
\end{tabular}
\and
\begin{tabular}{c}
Bezirgen Veliyev \\ 
{\small	CREATES} \\ 
\small Aarhus University \\ 
%{\small	Fuglesangs Alle 10, 8210 Aarhus V.} \\ 
{\small	\href{mailto:bveliyev@econ.au.dk}{bveliyev@econ.au.dk}}
\end{tabular}
}
\date{}
\date{First version: December 2018 \\
This version\footnote{Early versions of the paper have also contained results incorporating covariates, which are now available in the companion paper \cite{kpv2}.} : July 2020}
\maketitle
\singlespacing
\begin{abstract}
Consider a setting in which a policy maker assigns subjects to treatments, observing each outcome before the next subject arrives. Initially, it is unknown which treatment is best, but the sequential nature of the problem permits learning about the effectiveness of the treatments. While the multi-armed-bandit literature has shed much light on the situation when the policy maker compares the effectiveness of the treatments through their mean, much less is known about other targets. This is restrictive, because a cautious decision maker may prefer to target a robust location measure such as a quantile or a trimmed mean. Furthermore, socio-economic decision making often requires targeting purpose specific characteristics of the outcome distribution, such as its inherent degree of inequality, welfare or poverty. In the present paper we introduce and study  sequential learning algorithms when the distributional characteristic of interest is a general functional of the outcome distribution. Minimax expected regret optimality results are obtained within the subclass of explore-then-commit policies, and for the unrestricted class of all policies.
\end{abstract}

%\noindent \textbf{JEL Classification}: C21, C44.


\newpage
\doublespacing
\section{Introduction}\label{sec:Intro}

A fundamental question in statistical decision theory is how to optimally assign subjects to treatments. Important recent contributions include \cite{chamberlain2000econometrics}, \cite{manski2004statistical},
\cite{dehejia2005program},
\cite{hirano2009asymptotics}, \cite{stoye2009minimax},
\cite{bhattacharya2012inferring}, \cite{stoye2012minimax}, \cite{tetenov2012statistical}  \cite{manski2016sufficient},
\cite{athey2017efficient}, 
\cite{kitagawa2018should}, and \cite{manski2019treatment}; cf.~also the overview in~\cite{hirano2018statistical}. In the present paper we focus on assignment problems where the subjects to be treated arrive sequentially. Thus, in contrast to the above mentioned articles, the dataset is gradually constructed during the learning process. In this setting, a policy maker who seeks to assign subjects to the treatment with the highest \emph{expected} outcome (but who initially does not know which treatment is best), can draw on a rich and rapidly expanding literature on ``multi-armed bandits.'' Important contributions include ~\cite{thomp},~\cite{robbins1952some}, \cite{gittins1979bandit}, \cite{lai1985asymptotically}, \cite{agrawal1995sample},  \cite{auer1995gambling}, \cite{MOSS}; cf.~\cite{bubeck2012regret} and \cite{lattimore2019bandit} for introductions to the subject and for further references. In many applications, however, the quality of treatments cannot successfully be compared according to the expectation of the outcome distribution: A cautious policy maker may prefer to use another (more robust) measure of location, e.g., a quantile or a trimmed mean; or may actually want to make assignments targeting a different distributional characteristic than its location. Examples falling into the latter category are encountered in many socio-economic decision problems, where one wants to target, e.g., a welfare measure that incorporates inequality or poverty implications of a treatment. Inference for such ``distributional policy effects'' has received a great deal of attention in non-sequential settings, e.g.,~\cite{Gastwirth1974}, \cite{manski1988ordinal}, \cite{Thistle1990}, \cite{mills1997statistical}, \cite{davidson2000statistical}, \cite{abadie2002instrumental}, \cite{abadie2002bootstrap},
\cite{chernozhukov2005iv}, \cite{davidsonflachaire2007}, \cite{BarrettDonald2009}, \cite{hirano2009asymptotics}, \cite{SCHLUTER2009}, \cite{rostek2010quantile}, \cite{rothe2010nonparametric, rothe2012partial}, \cite{chernozhukov2013inference}, \cite{kit2017} and \cite{manski2019remarks}.\footnote{In contrast to much of the existing theoretical results concerning inference on inequality, welfare, or poverty measures, we do not investigate (first or higher-order) asymptotic approximations, but we establish exact finite sample results with explicit constants. To this end we cannot rely on classical asymptotic techniques, e.g., distributional approximations based on linearization arguments.} 

Motivated by robustness considerations and the general interest in distributional policy effects, we consider a decision maker who seeks to minimize regret compared to always assigning the unknown best treatment according to a \emph{functional} of interest. In order to achieve a low regret, the policy maker must sequentially learn the \emph{distributional characteristic} of interest for all available treatments, yet treat as many subjects as well as possible.

While most of the multi-armed bandit literature focuses on targeting the treatment  with the highest expectation, there are articles going beyond the first moment. This previous work has focused on risk functionals: \cite{maillard} considers problems where one targets a coherent risk measure. \cite{NIPS2012_4753}, \cite{7515237} and \cite{vakili2018decision} study a problem targeting the mean-variance functional, i.e., the variance minus a multiple of the expectation. \cite{zimin2014generalized}, motivated by earlier results on problems targeting specific risk measures, and \cite{kock2017optimal} consider problems where one targets a functional that can be written as a function of the mean and the variance. \cite{tran2014functional} and \cite{cassel2018general} do not restrict themselves to functionals of latter type, and consider bandit problems, where the target can be a general risk functional. These papers use various types of regret frameworks. \cite{tran2014functional} consider a ``pure-exploration'' regret function into which the errors made during the assignment period do not enter. \cite{maillard}, \cite{zimin2014generalized}, \cite{kock2017optimal} and \cite{vakili2018decision} consider a ``cumulative'' regret function that is closely related to the regret used in classical  multi-armed bandit problems (i.e., where the expectation is targeted). \cite{NIPS2012_4753}, \cite{7515237} and \cite{cassel2018general} consider a ``path-dependent'' regret function. The just-mentioned articles have in common that pointwise regret upper bounds are derived for certain policies (and the regret considered). Except for \cite{7515237} and \cite{vakili2018decision}, who exclusively consider the mean-variance functional, matching lower bounds are not established. Therefore, apart from the mean-variance functional, it remains unclear if the policies developed are optimal. The main goal of the present paper is to develop a minimax optimality theory for general functional targets. The regret function we work with is cumulative, and thus has the following important features which are relevant for many socio-economic assignment problems:
\begin{itemize}
\item Every subject not assigned to the best treatment contributes to the regret.
\item A loss incurred for one subject cannot be compensated by future assignments.
\end{itemize}
The first bullet point is not satisfied by a ``pure-exploration'' regret; the second is violated by ``path-dependent'' regrets. 

Our \emph{first contribution} is to establish minimax expected regret optimality properties within the subclass of ``explore-then-commit'' policies (cf.~Theorems~\ref{thm:LBETC} and~\ref{thm:ETCES}). These are policies that strictly separate the exploration and exploitation phases: one first attempts to learn the best treatment, e.g., by conducting a randomized controlled trial (RCT), on an initial segment of subjects. Based on the outcome, one then assigns all remaining subjects to the \emph{inferred best} treatment (which is not guaranteed to be the optimal one). Such policies are close to current practice in many socio-economic decision problems. \cite{garivier2016explore} recently studied optimality properties of explore-then-commit policies in a 2-arm Gaussian setting targeting exclusively the expectation. 

Our \emph{second contribution} is to obtain lower bounds on maximal expected regret over the class of \emph{all} policies (cf.~Theorem~\ref{thm:LB_NoCov}), and to show that they are matched by uniform upper bounds for the following two policies: Firstly, the ``F-UCB" policy (an extension of the UCB1 policy of~\cite{auer2002finite}), and secondly the ``F-aMOSS" policy (an extension of the anytime MOSS policy of \cite{degenne2016anytime}), cf.~Theorems~\ref{RegretBound} and~\ref{RegretBoundfaMOSS}.

Our lower bounds hold under very weak assumptions. Therefore, they settle firmly what can and cannot be achieved in a functional sequential treatment assignment problem. 

As a corollary to our results, comparing the regret upper bounds derived for the F-UCB and the F-aMOSS policy to the lower bound obtained for explore-then-commit policies, we reveal that in terms of maximal expected regret \emph{all} explore-then-commit policies are inferior to the F-UCB and the F-aMOSS policy, and therefore should not be used if it can be avoided. If an explore-then-commit policy has to be used, our results provide guidance on the optimal length of the exploration period.

In Sections~\ref{sec:num} and~\ref{sec:app} we provide numerical results (based on simulated and empirical data) comparing the regret-behavior of explore-then-commit policies with that of the F-UCB and the F-aMOSS policy. In this context we develop test-based and empirical-success-based explore-then-commit policies that might be of independent interest, because they provably possess desirable performance guarantees.

Concerning the functionals we permit our theory is very general. We verify in detail that it covers many inequality, welfare, and poverty measures, such as the Schutz coefficient, the Atkinson-, Gini- and Kolm-indices. This discussion can be found in~Appendix~\ref{app:FUNC}. We also show that our theory covers quantiles, U-functionals, generalized L-functionals, and trimmed means. These results can be found in~Appendix~\ref{sec:LCgeneral}. The results in these appendices are of high practical relevance, because they allow the policy maker to choose the functional-dependent constants appearing in the optimal policies in such a way that the performance guarantees apply.

In the companion paper \cite{kpv2} we address the important but nontrivial question how to construct policies that optimally incorporate covariate information. The results in the present paper are crucial for obtaining those results.

\section{Setting and assumptions}\label{sec:setup}
We consider a setting, where at each point in time~$t=1,\hdots,n$ a policy maker must assign a subject to one out of~$K$ treatments. Each subject is only treated once.\footnote{We emphasize that the sequential setting is different from the ``longitudinal'' or ``dynamic'' one in, e.g., \cite{robins1997causal}, \cite{lavori2000flexible}, \cite{murphy2001marginal}, \cite{murphy2003optimal} and \cite{murphy2005experimental}, where the \emph{same} subjects are treated repeatedly.} Thus, the index~$t$ can equivalently be thought of as indexing subjects instead of time. The observational structure is the one of a multi-armed bandit problem: After assigning a treatment, its outcome is observed, but the policy maker does not observe the counterfactuals. Having observed the outcomes of treatments~$1,\hdots,t-1$, subject~$t$ arrives, and must be assigned to a treatment. The assignment can be based on the information gathered from all \textit{previous} assignments and their outcomes, and, potentially, randomization. Thus, the data set is gradually constructed in the course of the treatment program. Without knowing a priori the identity of the ``best'' treatment, the policy maker seeks to assign subjects to treatments so as to minimize maximal expected regret (which we introduce in Equation~\eqref{eq:regret} further below).

This setting is a sequential version of the potential outcomes framework with multiple treatments. Note also that restricting attention to problems where only one out of the~$K$ treatments can be assigned does not exclude that a treatment consists of a combination of several other treatments (for example a combination of several drugs) --- one simply defines this combined treatment as a separate treatment at the expense of increasing the set of treatments. 

The precise setup is as follows: let the random variable~$Y_{i,t}$ denote the potential outcome of assigning subject~$t\in\cbr[0]{1,\hdots,n}$ to treatment~$i\in\mathcal{I}:=\cbr[0]{1,\hdots,K}$.\footnote{We do not explicitly consider the case of individuals arriving in batches. However, in our setup, one may also interpret $Y_{i,t}$ as a summary statistic of the outcomes of batch $t$, when all of its subjects were assigned to treatment $i$. For a more sophisticated way of handling batched data in case of targeting the mean treatment outcome, we refer to \cite{perchet2016batched}.}  That is, the potential outcomes of subject~$t$ are~$Y_t := (Y_{1,t}, \hdots, Y_{K, t})$. We assume that~$a \leq Y_{i,t} \leq b$, where~$a<b$ are real numbers. Furthermore, for every~$t$, let~$G_t$ be a random variable, which can be used for randomization in assigning the $t$-th subject. \emph{Throughout, we assume that~$Y_t$ for~$t \in \N$ are independent and identically distributed (i.i.d.); and we assume that the sequence~$G_t$ is i.i.d., and is independent of the sequence~$Y_t$.} Note that no assumptions are imposed concerning the dependence between the components of each random vector~$Y_t$.
%We shall denote the joint distribution of~$(Y_t, G_t)$ on the Borel sets of~$\R^{K+1}$ by~$\mathbb{P}_{F, G}$ where~$F = (F^1, \hdots, F^K)$ and~$G$ denotes the cdf of~$G_t$. 
We think of the \emph{randomization measure}, i.e., the distribution of~$G_t$, as being fixed, e.g., the uniform distribution on~$[0, 1]$. We denote the cumulative distribution function (cdf) of~$Y_{i,t}$ by~$F^i \in D_{cdf}([a,b])$, where~$D_{cdf}([a,b])$ denotes the set of all cdfs~$F$ such that~$F(a-) = 0$ and~$F(b) = 1$. 
%Furthermore, we denote the joint distribution of~$(Y_t, G_t)$ for~$t = 1, \hdots, n$ (on the Borel sets of~$\R^{n(K+1)}$) by~$\mathbb{P}^n_{F, G} = \bigotimes_{i = 1}^n \mathbb{P}_{F, G}$. Expectation with respect to (w.r.t.)~$\mathbb{P}_{F, G}$ and~$\mathbb{P}^n_{F, G}$ is denoted as~$\mathbb{E}_{F, G}$ and~$\mathbb{E}^n_{F, G}$, respectively. Note that this notation does not stress the joint distribution of the potential outcomes~$Y_t$ (up to its marginal cdfs).
The cdfs~$F^i$ for~$i = 1, \hdots, K$ are unknown to the policy maker.

A \emph{policy} is a triangular array of (measurable) functions~$\pi=\cbr[0]{\pi_{n,t}: n \in \N, 1\leq t\leq n}$. Here~$\pi_{n,t}$ denotes the assignment of the~$t$-th subject out of~$n$ subjects. In each row of the array, i.e., for each~$n\in \N$, the assignment~$\pi_{n,t}$ can depend only on previously observed treatment outcomes and  randomizations (previous and current). Formally, 
\begin{equation}\label{eqn:policyfunctiondef}
\pi_{n,t}: ([a,b] \times \R)^{t-1}\times \R\to \mathcal{I}.
\end{equation}
Given a policy~$\pi$ and~$n\in \N$, the input to~$\pi_{n,t}$ is denoted as~$(Z_{t-1}, G_t)$. Here~$Z_{t-1}$ is defined recursively: The first treatment~$\pi_{n,1}$ is a function of~$G_1$ alone, as no treatment outcomes have been observed yet (we may interpret~$(Z_{0}, G_1) = G_1$). The second treatment is a function of~$Z_1:=(Y_{\pi_{n,1}(G_1),1}, G_1)$, the outcome of the first treatment and the first randomization, and of~$G_2$. For~$t \geq 3$ we have~$$Z_{t-1}:=(Y_{\pi_{n,t-1}(Z_{t-2},G_{t-1}), t-1}, G_{t-1}, Z_{t-2}) = (Y_{\pi_{n,t-1}(Z_{t-2},G_{t-1}), t-1}, G_{t-1}, \hdots, Y_{\pi_{n,1}(G_{1}),1}, G_{1}).$$ 
The~$2(t-1)$-dimensional random vector~$Z_{t-1}$ can be interpreted as the information available after the~$(t-1)$-th treatment outcome was observed. We emphasize that~$Z_{t-1}$ depends on the policy~$\pi$ via~$\pi_{n,1}, \hdots, \pi_{n, t-1}$. In particular,~$Z_{t-1}$ also depends on~$n$, which we do not show in our notation. For convenience, the dependence of~$\pi_{n,t}(Z_{t-1}, G_t)$ on~$Z_{t-1}$ and~$G_t$ is often suppressed, i.e., we often abbreviate~$\pi_{n,t}(Z_{t-1}, G_t)$ by~$\pi_{n,t}$ if it is clear from the context that the actual assignment~$\pi_{n,t}(Z_{t-1}, G_t)$ is meant, instead of the function defined in Equation~\eqref{eqn:policyfunctiondef}.

\begin{remark}[Concerning the dependence of~$\pi_{n,t}$ on the horizon~$n$]\label{rem:horizonpolicy}
We have chosen to allow the assignments~$\pi_{n,1}, \hdots, \pi_{n,n}$ to depend on $n$, the total number of assignments to be made. Consequently, for~$n_1<n_2$ it may be that~$\cbr[0]{\pi_{n_1,t}:1\leq t\leq n_1}$ does not coincide with the first~$n_1$ elements of~$\cbr[0]{\pi_{n_2,t}:1\leq t\leq n_2}$. This is crucial, as a policy maker who knows~$n$ may choose different sequences of allocations for different~$n$. For example, one may wish to explore the efficacies of the available treatments in more detail if one knows that the total sample size is large, such that there is much opportunity to benefit from this knowledge later on. We emphasize that while our setup allows us to study policies that make use of~$n$, we devote much attention to policies that do not. The latter subclass of policies is important. For example, a policy maker may want to run a treatment program for a year, say, but it is unknown in advance how many subjects will arrive to be treated. In such a situation, one needs a policy that works well irrespective of the unknown horizon. Such policies are called ``anytime policies,'' as~$\pi_t:=\pi_{n,t}$ does not depend on~$n$.
\end{remark}
%
%This reflects real life situations where one, for example, may expect the outcome of similar drugs to have similar outcome distributions. 
%Let~$\P^i$ denote the outcome distribution on~$\mathcal{B}([a,b])$ of treatment~$i$ with corresponding cdf~$F^i$.
The ideal solution of the policy maker would be to assign every subject to the ``best'' treatment. In the present paper, this is understood in the sense that the outcome distribution for the best treatment maximizes a given functional 
\begin{equation}\label{eqn:functionaldef}
\mathsf{T}: D_{cdf}([a,b]) \to \mathbb{R}.
\end{equation}

We do not assume that the maximizer is unique, i.e.,~$\argmax_{i \in \mathcal{I}} \mathsf{T}(F^i)$ need not be a singleton. The specific functional chosen by the policy maker will depend on the application, and encodes the particular distributional characteristics the policy maker is interested in. For a streamlined presentation of our results it is helpful to keep the functional~$\mathsf{T}$ abstract at this point (see Section~\ref{subs:exshort} below for an example, and a brief overview of examples we study in detail in appendices).

The ideal solution of the policy maker of assigning each subject to the best treatment is infeasible, simply because it is not known in advance which treatment is best. Therefore, every policy will make mistakes. To compare different policies,
we define the (cumulative) \emph{regret} of a policy~$\pi$ at horizon~$n$ as
\begin{equation}
\begin{aligned}
R_n(\pi) &= R_n(\pi; F^1, \hdots, F^K, Z_{n-1}, G_n) = 
\sum_{t=1}^n 
\left[
\max_{i\in\mathcal{I}}  \mathsf{T}(F^i)- \mathsf{T}(F^{\pi_{n,t}(Z_{t-1}, G_t)})\right];
\label{eq:regret}
\end{aligned}
\end{equation}
i.e., for every individual subject that is not assigned to the best treatment one incurs a loss. One important feature of~$R_n(\pi)$ is that the losses incurred at time~$t$ cannot be nullified by later assignments. As discussed in the introduction, cumulative regret functions have previously been used by \cite{maillard}, \cite{zimin2014generalized}, \cite{kock2017optimal} and \cite{vakili2018decision}, the latter explicitly emphasizing the practical relevance of this regret notion in the context of clinical trials where the loss in each individual assignment needs to be controlled.

The unknown outcome distributions~$F^1, \hdots, F^K$ are assumed to vary in a pre-specified class of cdfs. Following the minimax-paradigm, we evaluate policies according to their worst-case behavior over such classes. We refer to~\cite{manski2016sufficient} for further details concerning the minimax point-of-view in the context of treatment assignment problems, and for a comparison with other approaches such as the Bayesian. Formally, we seek a policy~$\pi$ that minimizes maximal expected regret, that is, a policy that minimizes 
\begin{equation}\label{eq:maxreg}
\sup_{\substack{F^i \in \mathscr{D} \\ i = 1, \hdots, K }}\E [R_n(\pi)],
\end{equation}
where~$\mathscr{D}$ is a subset of~$D_{cdf}([a,b])$. The supremum is taken over all potential outcome vectors~$Y_t$ such that the marginals~$Y_{i,t}$ for~$i = 1, \hdots, K$ have a cdf in~$\mathscr{D}$. The set~$\mathscr{D}$ will typically be nonparametric, and corresponds to the assumptions one is willing to impose on the cdfs of each treatment outcome, i.e., on~$F^1, \hdots, F^K$. Note that the maximal expected regret of a policy~$\pi$ as defined in the previous display depends on the horizon~$n$. We will study this dependence on~$n$. In particular, we will study the \emph{rate} at which the maximal expected regret increases in~$n$ for a given policy~$\pi$; furthermore, we will study the question of which kind of policy is optimal in the sense that the rate is optimal.

The following assumption is the main requirement we impose on the functional~$\mathsf{T}$ and the set~$\mathscr{D}$. We denote the supremum metric on~$D_{cdf}([a,b])$ by~$\|\cdot\|_{\infty}$, i.e., for cdfs~$F$ and~$G$ we let~$\|F-G\|_{\infty} = \sup_{x \in \R} |F(x) - G(x)|$.

\begin{assumption}\label{as:MAIN}
The functional~$\mathsf{T}: D_{cdf}([a,b]) \to \R$ and the non-empty set~$\mathscr{D} \subseteq D_{cdf}([a,b])$ satisfy
\begin{equation}\label{eqn:lipcondAS}
|\mathsf{T}(F) - \mathsf{T}(G)|\leq C\|F- G\|_{\infty} \quad \text{ for every } \quad F \in \mathscr{D} \text{ and every } G \in D_{cdf}([a,b])
\end{equation}
for some~$C > 0$.
\end{assumption}
\begin{remark}[Restricted-Lipschitz continuity]\label{rem:asym}
Assumption~\ref{as:MAIN} implies that the functional~$\mathsf{T}$ is Lipschitz continuous when restricted to~$\mathscr{D}$ (the domain being equipped with~$\|\cdot\|_{\infty}$). We emphasize, however, that if~$\mathscr{D} \neq D_{cdf}([a,b])$, the functional~$\mathsf{T}$ is not necessarily required to be Lipschitz-continuous on all of~$D_{cdf}([a,b])$. This is due to the asymmetry inherent in the condition imposed in Equation~\eqref{eqn:lipcondAS}, where $F$ varies only in~$\mathscr{D}$, but~$G$ varies in all of~$D_{cdf}([a,b])$. 
\end{remark}

\begin{remark}\label{rem:closure}
A simple approximation
argument\footnote{
Let~$\bar{F} \in D_{cdf}([a,b])$ be such that~$\|F_m - \bar{F}\|_{\infty} \to 0$ as $m \to \infty$ for a sequence~$F_m \in \mathscr{D}$, and let~$G \in D_{cdf}([a,b])$. Then,~$|\mathsf{T}(\bar{F}) - \mathsf{T}(G)| \leq |\mathsf{T}(\bar{F}) - \mathsf{T}(F_m)| + |\mathsf{T}(F_m) - \mathsf{T}(G)|$, which, by Assumption~\ref{as:MAIN}, is not greater than~$2C\|\bar{F} - F_m\|_{\infty} + C\|\bar{F} - G\|_{\infty} \to C\|\bar{F} - G\|_{\infty}$ as~$m \to \infty$.} shows that if Assumption~\ref{as:MAIN} is satisfied with~$\mathscr{D}$ and~$C$, then Assumption~\ref{as:MAIN} is also satisfied with~$\mathscr{D}$ replaced by the closure of~$\mathscr{D} \subseteq D_{cdf}([a,b])$ (the ambient space~$D_{cdf}([a,b])$ being equipped with the metric~$\|\cdot\|_{\infty}$) and the same constant~$C$.
\end{remark}

\begin{remark}\label{rem:Lips}
The set~$\mathscr{D}$ encodes assumptions imposed on the cdfs of each treatment outcome. In particular, the larger~$\mathscr{D}$, the less restrictive is~$F^i\in\mathscr{D}$ for~$i\in\mathcal{I}$. Ideally, one would thus like~$\mathscr{D} = D_{cdf}([a,b])$, which, however, is too much to ask for some functionals. Furthermore, there is a trade-off between the sizes of~$C$ and~$\mathscr{D}$, in the sense that a larger class~$\mathscr{D}$ typically requires a larger constant~$C$. The reader who wants to get an impression of some of the classes of cdfs we consider may want to consult Appendix~\ref{sec:cdfs}, where important classes of cdfs are defined.
\end{remark}

\subsection{Functionals that satisfy Assumption~\ref{as:MAIN}: A summary of results in Appendix~\ref{app:FUNC} and Appendix~\ref{sec:LCgeneral}}\label{subs:exshort}

In the present paper, we do not contribute to the construction of functionals for specific questions. Rather, we take the functional as given. To choose an appropriate functional, the policy maker can already draw on a very rich and still expanding body of literature; cf.~\cite{lambert}, \cite{chakravarty2009} or \cite{cowell} for textbook-treatments. To equip the reader with a specific and important example of a functional~$\mathsf{T}$, one may think of the Gini-welfare measure (cf.~\cite{SEN1974387})
\begin{equation}\label{eqn:GINIWELFAREM}
\mathsf{T}(F) = \int x dF(x) - \frac{1}{2} \int \int |x_1 - x_2|dF(x_1) dF(x_2).
\end{equation}

Because all of our results impose Assumption~\ref{as:MAIN}, a natural question concerns its generality. To convince the reader that Assumption~\ref{as:MAIN} is often satisfied, and to make the policies studied implementable (as they require knowledge of~$C$), we show in Appendix~\ref{app:FUNC} that Assumption~\ref{as:MAIN} is satisfied for many important inequality, welfare, and poverty measures (together with formal results concerning the sets~$\mathscr{D}$ along with corresponding constants~$C$). For example, it is shown that for the above Gini-welfare measure, Assumption~\ref{as:MAIN} is satisfied with~$\mathscr{D} = D_{cdf}([a,b])$, i.e., without any restriction on the treatment cdfs~$F^1, \hdots, F^K$ (apart from having support~$[a,b]$), and with constant~$C = 2(b-a)$. At this point we highlight some further functionals that satisfy Assumption~\ref{as:MAIN}:
\begin{enumerate}
\item The \emph{inequality measures} we discuss in Appendix~\ref{sec:inequalitymeasures} include the Schutz-coefficient (\cite{schutz}, \cite{schutzcomment}), the Gini-index, the class of linear inequality measures of \cite{mehran}, the generalized entropy family of inequality indices including Theil's index, the Atkinson family of inequality indices (\cite{atkinson1970}), and the family of Kolm-indices (\cite{kolm1}). In many cases, we discuss both relative and absolute versions of these measures. 
\item In Appendix~\ref{sec:welfaremeasures} we provide results for \emph{welfare measures} based on inequality measures. 
\item The \emph{poverty measures} we discuss in Section~\ref{sec:povertymeasures} are the headcount ratio, the family of poverty measures of \cite{sen1976} in the generalized form of \cite{kakwani1980}, and the family of poverty measures suggested by \cite{foster84}.
\end{enumerate}
The results in Appendices~\ref{sec:inequalitymeasures},~\ref{sec:welfaremeasures}, and~\ref{sec:povertymeasures} mentioned above are obtained from and supplemented by a series of general results that we develop in Appendix~\ref{sec:LCgeneral}. These results verify Assumption~\ref{as:MAIN} for \emph{U-functionals} defined in Equation~\eqref{eqn:SINT} (i.e., population versions of U-statistics, e.g., the mean or the variance), \emph{quantiles}, generalized \emph{L-functionals} due to \cite{serflinggen} defined in Equation~\eqref{eqn:defLstat}, and \emph{trimmed U-functionals} defined in Equation~\eqref{eqn:functinttrimmed}. These techniques are of particular interest in case one wants to apply our results to functionals~$\mathsf{T}$ that we do not explicitly discuss in Appendix~\ref{app:FUNC}. 

The results in Appendix~\ref{app:FUNC} and Appendix~\ref{sec:LCgeneral} could also be of independent interest, because they immediately allow the construction of uniformly valid (over~$\mathscr{D}$) confidence intervals and tests in finite samples. To see this, observe that Assumption~\ref{as:MAIN} together with the measurability Assumption~\ref{as:MB} given further below and the Dvoretzky-Kiefer-Wolfowitz-Massart inequality in \cite{massart1990} implies that, uniformly over~$F \in \mathscr{D}$, the confidence interval~$\mathsf{T}(\hat{F}_n) \pm C\sqrt{\log(2/\alpha)/(2n)}$ covers~$\mathsf{T}(F)$ with probability not smaller than~$1-\alpha$; here~$\hat{F}_n$ denotes the empirical cdf based on an i.i.d.~sample of size~$n$ from~$F$.

\subsection{Further notation and an additional assumption}\label{sec:furthernot} 
Before we consider maximal expected regret properties of certain classes of policies, we need to introduce some more notation: Given a policy~$\pi$ and~$n \in \N$, we denote the number of times treatment~$i$ has been assigned up to time~$t$ by
\begin{equation}\label{eqn:Sintdef}
S_{i,n}(t):= \sum_{s = 1}^t \mathds{1}\{\pi_{n,s}(Z_{s-1}, G_s) = i\},
\end{equation}
and we  abbreviate~$S_{i,n}(n) = S_{i}(n)$. Defining the loss incurred due to assigning treatment~$i$ instead of an optimal one by~$\Delta_i:=\max_{k \in \mathcal{I}}\mathsf{T}(F^{k})-\mathsf{T}(F^i)$, the regret~$R_n(\pi)$, which was defined in Equation~\eqref{eq:regret}, can equivalently be written as
\begin{align}\label{eq:regret2new}
R_n(\pi)=\sum_{i:\Delta_i>0}\Delta_i\sum_{t=1}^n\mathds{1}\{\pi_{n,t}(Z_{t-1}, G_t)=i\} = \sum_{i:\Delta_i>0}\Delta_iS_i(n).
\end{align}
On the event~$\{S_{i,n}(t) > 0\}$ we define the empirical cdf based on the outcomes of all subjects in~$\{1, \hdots, t\}$ that have been assigned to treatment~$i$
\begin{equation}\label{eq:Fhat}
\hat{F}_{i, t, n}(z) := S^{-1}_{i,n}(t) \sum_{\substack{1 \leq s \leq t \\
\pi_{n, s}(Z_{s-1}, G_s) = i
}} \mathds{1}\{Y_{i, s} \leq z\}, \quad \text{ for every } z \in \R.
\end{equation}
Note that the random sampling times~$s$ such that~$\pi_{n, s}(Z_{s-1}, G_s) = i$ depend on previously observed treatment outcomes. 

We shall frequently need an assumption that guarantees that the functional~$\mathsf{T}$ evaluated at empirical cdfs, such as~$\hat{F}_{i, t, n}$ just defined in~Equation~\eqref{eq:Fhat}, is measurable.
\begin{assumption}\label{as:MB}
For every~$m \in \N$, the function on~$[a,b]^m$ that is defined via~$x \mapsto \mathsf{T}(m^{-1} \sum_{j = 1}^m \mathds{1}\cbr[0]{x_j \leq \cdot}),$ i.e.,~$\mathsf{T}$ evaluated at the empirical cdf corresponding to~$x_1, \hdots, x_m$, is Borel measurable.
\end{assumption}
\noindent
Assumption~\ref{as:MB} is typically satisfied and imposes no practical restrictions.

Finally, and following up on the discussion in Remark~\ref{rem:horizonpolicy}, we shall introduce some notational simplifications in case a policy~$\pi$ is such that~$\pi_{n,t}$ is independent of~$n$, i.e., is an anytime policy. It is then easily seen that the random quantities~$S_{i,n}(t)$ and~$\hat{F}_{i,t,n}$ do not depend on~$n$ (as long as~$t$ and~$n$ are such that~$n \geq t$). Therefore, for such policies, we shall drop the index~$n$ in these quantities.

\section{Explore-then-commit policies}\label{sec:ETC}
A natural approach to assigning subjects to treatments in our sequential setup would be to first conduct a randomized controlled trial (RCT) to study which treatment is best, and then to use the acquired knowledge to assign the inferred best treatment to all remaining subjects. Such policies are special cases of \emph{explore-then-commit} policies, which we study in this section. Informally, an explore-then-commit policy deserves its name as it (i) uses the first~$n_1$ subjects to \emph{explore}, in the sense that every treatment is assigned, in expectation, at least proportionally to~$n_1$; and (ii) then \emph{commits} to a single (inferred best) treatment after the first~$n_1$ treatments have been used for exploration. Here,~$n_1$ may depend on the horizon~$n$.

Formally, we define an explore-then-commit policy as follows. %(recall the definition of~$S_{i,n}(t)$ in Equation~\eqref{eqn:Sintdef} above).
\begin{definition}[Explore-then-commit policy]\label{def:etc}
A policy~$\pi$ is an \emph{explore-then-commit policy}, if there exists a function~$n_1: \N \to \N$ and an~$\eta \in (0, 1)$, such that for every~$n \in \N$ we have that~$n_1(n) \leq n$, and such that the following conditions hold for every~$n \geq K$:
\begin{enumerate}
\item \textbf{Exploration Condition}: We have that~$$\inf_{\substack{F^i \in \mathscr{D} \\ i = 1, \hdots, K}} \inf_{j \in \mathcal{I}}~\E[ S_{j,n}(n_1(n))] \geq \eta n_1(n).$$
Here, the first infimum is taken over all potential outcome vectors~$Y_t$ such that the marginals~$Y_{i,t}$ for~$i = 1, \hdots, K$ have a cdf in~$\mathscr{D}$. 

[That is, regardless of the (unknown) underlying marginal distributions of the potential outcomes, each treatment is assigned, in expectation, at least~$\eta n_1(n)$ times among the first~$n_1(n)$ subjects.]
\item \textbf{Commitment Condition}: There exists a function $\pi^c_n: ([a,b] \times \R)^{n_1(n)}\to \mathcal{I}$ such that, for every~$t = n_1(n)+1, \hdots, n$, we have
\begin{equation}
\pi_{n,t}(z_{t-1}, g) = \pi^c_n(z_{n_1(n)}) \quad \text{ for every } z_{t-1} \in ([a,b] \times \R)^{t-1} \text{ and every } g \in \R,
\end{equation}
where~$z_{n_1(n)}$ is the vector of the last $2n_1(n)$ coordinates of~$z_{t-1}$. 

[That is, the subjects~$t = n_1(n)+1, \hdots, n$ are all assigned to the same treatment, which is selected based on the~$n_1(n)$ outcomes and randomizations observed during the exploration period.]
\end{enumerate}
\end{definition}
It would easily be possible to let the commitment rule~$\pi_n^c$ depend on further external randomization. For simplicity, we omit formalizing such a generalization. We shall now discuss some important examples of explore-then-commit policies.
\begin{example}\label{rem:RCTex}
A policy that first conducts an RCT based on a sample of~$n_1(n) \leq n$ subjects, followed by \emph{any} assignment rule for subjects~$n_1(n) + 1, \hdots, n$ that satisfies the commitment condition in Definition~\ref{def:etc}, is an explore-then-commit policy, provided the concrete randomization scheme used in the RCT encompasses sufficient exploration. In particular,~$\pi_{n,t}(Z_{t-1},G_t)=G_t$ with~$\P(G_t=i):=\frac{1}{K}$ for every~$1 \leq t \leq n_1(n)$ and every~$i\in\mathcal{I}$ satisfies the exploration condition in Definition~\ref{def:etc} with~$\eta=\frac{1}{K}$; more generally, Definition~\ref{def:etc} holds if~$\eta := \inf_{i \in \mathcal{I}} \P(G_t = i) >0$. Alternatively, a policy that enforces balancedness in the exploration phase through assigning subjects~$t = 1, \hdots, n_1(n)$ to treatments ``cyclically,'' i.e.,~$\pi_{n,t}(Z_{t-1}, G_t) = (t \mod K) + 1$, satisfies the exploration condition in Definition~\ref{def:etc} with~$\eta = 1/(2K)$ if~$n_1(n) \geq K$ for every~$n \geq K$. Concrete choices for commitment rules for subjects~$n_1(n) + 1, \hdots, n$ include:
\begin{enumerate}
\item In case~$K=2$, a typical approach is to assign the fall-back treatment if, according to some test, the alternative treatment is not significantly better, and to assign the alternative treatment if it is significantly better. The sample size~$n_1(n)$ used in the RCT is typically chosen to ensure that the specific test used achieves a desired power against a certain effect size. We refer to the description of the ETC-T policy in Section~\ref{sec:numimplement} for a specific example of a test and a corresponding rule for choosing~$n_1$, of which we establish that it achieves the desired power requirement (while holding the size).
\item As an alternative to test-based commitment rules, one can use an \emph{empirical success rule} as in \cite{manski2004statistical}, which in our general context amounts to assigning an element of~$\argmax_{i \in \mathcal{I}} \mathsf{T}(\hat{F}_{i, n_1(n),n})$ to subjects~$n_1(n)+1, \hdots, n$. Specific examples of such a policy, together with concrete ways of choosing~$n_1$ that come with certain performance guarantees, are discussed in Policy~\ref{poly:ES} below and in the description of the ETC-ES policy in Section~\ref{sec:numimplement}.
\end{enumerate}
\end{example}

We now establish regret lower bounds for the class of explore-then-commit policies. To exclude trivial cases, we assume that~$\mathscr{D}$ (which is typically convex) contains a line segment on which the functional~$\mathsf{T}$ is not everywhere constant. 
\begin{assumption}\label{as:NCONSTLINE}
The functional~$\mathsf{T}: D_{cdf}([a,b]) \to \R$ satisfies Assumption~\ref{as:MAIN}, and~$\mathscr{D}$ contains two elements~$H_1$ and~$H_2$, such that 
\begin{equation}\label{eqn:NCONSTLINE}
J_{\tau} := \tau H_1 + (1-\tau)H_2 \in \mathscr{D} \quad \text{ for every }\tau \in [0, 1],
\end{equation}
and such that~$\mathsf{T}(H_1) \neq \mathsf{T}(H_2)$.
\end{assumption}

Since there only have to exist \emph{two} cdfs $H_1$ and $H_2$ as in Assumption~\ref{as:NCONSTLINE}, this is a condition that is practically always satisfied. 

The next theorem considers general explore-then-commit policies, as well as the subclass of policies where~$n_1(n)\leq n^*$ holds for every~$n \in \N$ for some~$n^* \in \N$. This subclass models situations, where the horizon~$n$ is unknown or ignored in planning the experiment, and the envisioned number of subjects used for exploration~$n^*$ is fixed in advance (here $n_1(n) = n^*$ for every~$n \geq n^*$, and~$n_1(n) = n$, else); the subclass also models situations where the sample size that can be used for experimentation is limited due to budget constraints. 
\begin{theorem}\label{thm:LBETC}
Suppose~$K = 2$ and that Assumption~\ref{as:NCONSTLINE} holds. Then the following statements hold:\footnote{The constants~$c_l$ depend on properties of the function~$\tau \mapsto \mathsf{T}(J_{\tau})$ for~$\tau \in [0, 1]$. More specifically, the constants depend on the quantities~$\eps$ and~$c_-$ from Lemma~\ref{lem:LBdist}. The precise dependence is made explicit in the proof.}
\begin{enumerate}
\item There exists a constant~$c_l > 0$, such that, for every explore-then-commit policy~$\pi$ that satisfies the exploration condition with~$\eta \in (0, 1)$, and for any randomization measure, it holds that 
\begin{align*}
\sup_{\substack{F^i \in \{J_{\tau} : \tau \in [0,1]\} \\ i = 1, 2}}\E [R_n(\pi)]  \geq \eta c_l n^{2/3} \quad \text{ for every } n \geq 2.
\end{align*} 
\item For every~$n^* \in \N$ there exists a constant~$c_l = c_l(n^*)$, such that, for every explore-then-commit policy~$\pi$ that satisfies (i) the exploration condition with~$\eta \in (0,1)$ and (ii)~$n_1(\cdot)\leq n^*$, and for any randomization measure, it holds that
\begin{align*}
\sup_{\substack{F^i \in \{J_{\tau} : \tau \in [0,1]\} \\ i = 1, 2}}\E [R_n(\pi)]  \geq \eta c_l n  \quad \text{ for every } n \geq 2.
\end{align*} 
\end{enumerate}
\end{theorem}
The first part of Theorem~\ref{thm:LBETC} shows that, under the minimal assumption of~$\mathscr{D}$ containing a line segment on which~$\mathsf{T}$ is not constant, \emph{any} explore-then-commit policy must incur maximal expected regret that increases at least of order~$n^{2/3}$ in the horizon~$n$. 

The second part implies in particular that when~$n$ is unknown, such that the  exploration period~$n_1$ cannot depend on it, any explore-then-commit policy must incur linear maximal expected regret. We note that this is the worst possible rate of regret, since by Assumption~\ref{as:MAIN} \emph{no} policy can have larger than linear maximal expected regret. 

The lower bounds on maximal expected regret are obtained by taking the maximum only over all potential outcome vectors with marginal distributions in
the line segment in Equation~\eqref{eqn:NCONSTLINE}. This is a one-parametric subset of~$\mathscr{D}$ over which~$\mathsf{T}$ nevertheless varies sufficiently to obtain a good lower bound.

We now prove that a maximal expected regret of rate~$n^{2/3}$ is attainable in the class of explore-then-commit policies, i.e., we show that the lower bound in the first part of Theorem~\ref{thm:LBETC} cannot be improved upon. In particular, we show that employing an empirical success type commitment rule after an RCT in the exploration phase as discussed in Example~\ref{rem:RCTex} yields a maximal expected regret of this order. To be precise, we consider the following policy, which in contrast to test-based commitment rules (which require the choice of a suitable test and taking into account multiple-comparison issues in case~$K>2$) can be implemented seamlessly for any number of treatments:

\bigskip

\onehalfspacing
\begin{policy}[H]
\caption{Explore-then-commit empirical-success policy~$\tilde{\pi}$ \label{alg}}\label{poly:ES}
\For{$t = 1, \hdots, n_1(n) :=  \min(K \lceil n^{2/3} \rceil, n)$}{assign~$\tilde{\pi}_{n,t}(Z_{t-1}, G_t) = G_t$, with~$G_t$ uniformly distributed on~$\mathcal{I}$} 
\For{$t = n_1(n)+1, \hdots, n$}{assign~$\tilde{\pi}_{n,t}(Z_{t-1}, G_t) = \min \argmax\limits_{i: S_{i, n}(n_1(n)) > 0} \mathsf{T}(\hat{F}_{i,n_1(n), n})$}
\end{policy}
\bigskip

\doublespacing

Note that the policy~$\tilde{\pi}$ is an explore-then-commit policy that requires knowledge of the horizon~$n$, which by Theorem~\ref{thm:LBETC} is necessary for obtaining a rate slower than~$n$. The outer minimum in the second for loop in the policy is just taken to break ties (if necessary). Our result concerning~$\tilde{\pi}$ is as follows (an identical statement can be established for a version of~$\tilde{\pi}$ with cyclical assignment during the exploration phase as discussed in Remark~\ref{rem:RCTex}; the proof follows along the same lines, and we skip the details).
\begin{theorem}\label{thm:ETCES}
Under Assumptions~\ref{as:MAIN} and~\ref{as:MB}, the explore-then-commit empirical-success policy~$\tilde{\pi}$ satisfies
\begin{equation}\label{eqn:regretc23}
\sup_{\substack{F^i \in \mathscr{D} \\ i = 1, \hdots, K }} \E [R_n(\tilde{\pi})] \leq 6CKn^{2/3} \quad \text{ for every } n \in \N.
\end{equation}
\end{theorem}

Theorems~\ref{thm:LBETC} and~\ref{thm:ETCES} together prove that within the class of explore-then-commit policies, the policy~$\tilde{\pi}$ is rate optimal in~$n$. An upper bound as in Theorem~\ref{thm:ETCES} for the special case of the mean functional can be found in Chapter~6 of~\cite{lattimore2019bandit}. We shall next show that policies which do not separate the exploration and commitment phase can obtain lower maximal expected regret. In this sense, the natural idea of separating exploration and commitment phases turns out to be suboptimal from a decision-theoretic point-of-view in functional sequential treatment assignment problems.

The finding that for large classes of functional targets explore-then-commit policies are suboptimal in terms of maximal expected regret does, of course, by no means discredit RCTs and subsequent testing for other purposes. For example, RCTs are often used to test for a causal effect of a treatment, cf.~\cite{imbens2009recent} for an overview and further references. The goal of the present article is not to test for a causal effect, but to assist the policy maker in minimizing regret, i.e., to keep to a minimum the sum of all losses due to assigning subjects wrongly. This goal, as pointed out in, e.g., \cite{manski2004statistical}, \cite{manski2016sufficient} and \cite{manski2019treatment}, is only weakly related to testing. For example, the policy maker may care about more than just controlling the probabilities of Type~1 and Type~2 errors. In particular the magnitude of the losses when errors occur are important components of regret.

\section{Functional UCB-type policies and regret bounds}\label{sec:F-UCBnoCov}
In this section we define and study two policies based on upper-confidence-bounds. We start with the \emph{Functional Upper Confidence Bound} (F-UCB) policy. It is inspired by the UCB1 policy of~\cite{auer2002finite} for multi-armed bandit problems targeting the mean, which is derived from a policy in \cite{agrawal1995sample}, building on \cite{lai1985asymptotically}. Extensions of the UCB1 policy to targeting risk functionals have been considered by \cite{NIPS2012_4753}, \cite{maillard}, \cite{zimin2014generalized}, \cite{7515237}, and \cite{vakili2018decision}. The F-UCB policy can target any functional (and reduces to the UCB1 policy of~\cite{auer2002finite} in case one targets the mean). It has the practical advantage of not needing to know the horizon~$n$, cf.~Remark~\ref{rem:horizonpolicy} (recall also the notation introduced in Section~\ref{sec:furthernot}). Furthermore, no external randomization is required, which will therefore be notationally suppressed as an argument to the policy. The policy is defined as follows, where~$C$ is the constant from Assumption~\ref{as:MAIN}. 

\bigskip
\onehalfspacing
\begin{policy}[H]
	\caption{F-UCB policy~$\hat{\pi}$ \label{FUCB} }\label{poly:FUCB}
	\textbf{Input:}~$\beta > 2$ \\
	\For{$t = 1, \hdots, K$}{assign~$\hat{\pi}_{t}(Z_{t-1}) = t$}
	\For{$t \geq K+1$}{assign~$\hat{\pi}_{t}(Z_{t-1}) = \min \argmax_{i\in \mathcal{I}}\cbr[2]{\mathsf{T}(\hat{F}_{i,t-1})+C\sqrt{ \beta \log(t)/(2 S_i(t-1))}}$}
\end{policy}

\bigskip
\doublespacing

After the~$K$ initialization rounds, the F-UCB policy assigns a treatment that i) is promising, in the sense that~$\mathsf{T}(\hat{F}_{i,t-1})$ is large, or ii) has not been well explored, in the sense that~$S_i(t-1)$ is small. The parameter~$\beta$ is chosen by the researcher and indicates the weight put on assigning scarcely explored treatments, i.e., treatments with low~$S_i(t-1)$. An optimal choice of~$\beta$, minimizing the upper bound on maximal expected regret, is given after Theorem~\ref{RegretBound} below. We use the notation~$\overline{\log}(x):=\max(\log(x), 1)$ for~$x > 0$. 
\begin{theorem}\label{RegretBound}
	Under Assumptions~\ref{as:MAIN} and~\ref{as:MB}, the F-UCB policy~$\hat{\pi}$ satisfies
	\begin{align}\label{eqn:RegretBoundthmstatic}
		\sup_{\substack{F^i \in \mathscr{D} \\ i = 1, \hdots, K }}\E [R_n(\hat{\pi})]\leq c\sqrt{K n \overline{\log}(n)} \quad \text{ for every } n \in \N,
	\end{align}
	where~$c=c(\beta,C) = C \sqrt{ 2 \beta + (\beta+2)/(\beta-2)}$.
\end{theorem}

The upper bound on maximal expected regret just obtained is increasing in the number of available treatments~$K$. This is due to the fact that it becomes harder to find the best treatment as the number of available treatments increases. Note also that the choice~$\beta=2+\sqrt{2}$ minimizes~$c(\beta,C)$ and implies~$c\leq\sqrt{11}C$. 

In case of the mean functional, an upper bound as in Theorem~\ref{RegretBound} can be obtained from Theorem~1 in~\cite{auer2002finite} as explained after Theorem~2 in~\cite{MOSS}, cf.~also the discussion in Section 2.4.3 of~\cite{bubeck2012regret}.\footnote{High-probability bounds as in Theorem~8 in~\cite{audibert2009exploration} can also be obtained for the F-UCB policy, cf.~Theorem~\ref{thm:HPB} in Appendix \ref{sec:HPB}.} The proof of Theorem~\ref{RegretBound} is inspired by their arguments. However, we cannot exploit the specific structure of the mean functional and related concentration inequalities. Instead we rely on the high-level condition of~Assumption~\ref{as:MAIN} and the Dvoretzky-Kiefer-Wolfowitz-Massart inequality as established by~\cite{massart1990} to obtain suitable concentration inequalities, cf.~Equation~\eqref{eqn:appconcineq} in~Appendix~\ref{sec:LCgeneral}. Since adaptive sampling introduces dependence, we also need to take care of the fact that the empirical cdfs defined in~\eqref{eq:Fhat} are not directly based on a fixed number of i.i.d.~random variables. This is done via the optional skipping theorem of \cite{doob}, cf.~Appendix~\ref{app:OSkipD}. For functionals that can be written as a Lipschitz-continuous function of the first and second moment (a situation where Assumption~\ref{as:MAIN} holds), an upper bound of the same order as in Theorem~\ref{RegretBound} has been obtained in~\cite{kock2017optimal} for a successive-elimination type policy.

The lower bound in Theorem~\ref{thm:LBETC} combined with the upper bound in Theorem~\ref{RegretBound} shows that the maximal expected regret incurred by \emph{any} explore-then-commit policy grows much faster in~$n$ than that of the F-UCB policy. What is more, the F-UCB policy achieves this without making use of the horizon~$n$. Thus, in particular when~$n$ is unknown, a large improvement is obtained over \emph{any} explore-then-commit policy,  as the order of the regret decreases from~$n$ to~$\sqrt{n\log(n)}$. Hence, in terms of maximal expected regret, the policy maker is not recommended to separate the exploration and commitment phases.

Theorem~\ref{RegretBound} leaves open the possibility that one can construct policies with even slower growth rates of maximal expected regret. We now turn to establishing a lower bound on maximal expected regret within the class of all policies. In particular, the theorem also applies to policies that incorporate the horizon~$n$. 
\begin{theorem}\label{thm:LB_NoCov}
	Suppose~$K = 2$ and that Assumption~\ref{as:NCONSTLINE} holds. Then there exists a constant $c_l > 0$, such that for any policy~$\pi$ and any randomization measure, it holds that
	\begin{align}\label{eqn:suplownoco}
		\sup_{\substack{F^i \in \{J_{\tau} : \tau \in [0,1]\} \\ i = 1, 2}} \E [R_n(\pi)]  \geq c_ln^{1/2} \quad \text{ for every } n \in \N.
	\end{align}
\end{theorem}
Under the same assumptions used to establish the lower bound on maximal expected regret in the class of explore-then-commit policies, Theorem~\ref{thm:LB_NoCov} shows that \emph{any} policy must incur maximal expected regret of order at least~$n^{1/2}$. In combination with Theorem~\ref{RegretBound} this shows that, up to a multiplicative factor of~$\sqrt{\log(n)}$, no policy exists that has a better dependence of maximal expected regret on~$n$ than the F-UCB policy. In this sense the F-UCB policy is near minimax (rate-) optimal. 

For the special case of the mean functional a lower bound as in Theorem~\ref{thm:LB_NoCov} was given in Theorem~7.1 in~\cite{auer1995gambling}. Their proof is based on suitably chosen Bernoulli cdfs with parameters about~$1/2$, and thus provides a lower bound over all sets~$\mathscr{D}$ containing these cdfs, in particular over~$D_{cdf}([0,1])$. Depending on the functional considered, however, Bernoulli cdfs may not create sufficient variation in the functional to get good lower bounds. Furthermore, Bernoulli cdfs may not be contained in~$\mathscr{D}$, if, e.g., the latter does not contain discrete cdfs, in which case a lower bound derived for Bernoulli cdfs is not informative. For these two reasons, we have tailored the lower bound towards the functional and parameter space~$\mathscr{D}$ under consideration. As in the proof of Theorem~\ref{thm:LBETC} this is achieved by working with a suitably chosen one-parametric family of binary mixture cdfs of elements of the functional-specific line segment~$\{J_{\tau} : \tau \in [0,1]\}$; cf.~Lemma~\ref{lem:LBdist} in Appendix~\ref{sec:Aux}.

It is natural to ask whether a policy exists, which avoids the factor of~$\sqrt{\log(n)}$ appearing in Theorem~\ref{RegretBound}. In the special case of the mean functional, \cite{MOSS} and \cite{degenne2016anytime} answered this question affirmatively for the MOSS policy and an anytime MOSS policy, respectively. As the second policy in this section, following the construction in~\cite{degenne2016anytime}, we now consider a \emph{Functional anytime MOSS} (F-aMOSS) policy, and establish an upper bound on its maximal expected regret that matches the lower bound in Theorem~\ref{thm:LB_NoCov}. The policy is of UCB-type in the sense that it proceeds similarly as Policy~\ref{poly:FUCB}, but uses a slightly different confidence bound; cf.~Policy~\ref{poly:faMOSS} where for~$x > 0$ we write~$\log^+(x) = \max(\log(x), 0)$,

\bigskip
\onehalfspacing
\begin{policy}[H]
	\caption{F-aMOSS policy~$\check{\pi}$ }\label{poly:faMOSS}
	\textbf{Input:} $\beta > 1/4$ \\
	\For{$t = 1, \hdots, K$}{assign~$\check{\pi}_{t}(Z_{t-1}) = t$}
	\For{$t \geq K+1$}{assign~$\check{\pi}_{t}(Z_{t-1}) = \min \argmax_{i\in \mathcal{I}}\cbr[2]{\mathsf{T}(\hat{F}_{i,t-1})+C \sqrt{\frac{\beta}{S_i(t-1)} \log^+\left(\frac{t-1}{ KS_i(t-1)}\right) }}$}
\end{policy}

\bigskip
\doublespacing
A regret upper bound for the F-aMOSS policy is given next. 
\begin{theorem}\label{RegretBoundfaMOSS}
Under Assumptions~\ref{as:MAIN} and~\ref{as:MB}, the F-aMOSS policy~$\check{\pi}$ satisfies
\begin{align}\label{eqn:RegretfaMOSS}
\sup_{\substack{F^i \in \mathscr{D} \\ i = 1, \hdots, K }}\E [R_n(\check{\pi})]\leq C (4.83 + 6.66 \times d(\beta) + \sqrt{2\beta})\sqrt{\pi} \sqrt{K n} \quad \text{ for every } n \in \N,
\end{align}
for~$d(\beta) = \frac{( \beta W_0(e/4 \beta))^{\frac{1}{2}}}{
1-(4\beta W_0(e/4 \beta))^{\frac{1}{2}-\frac{1}{2 W_0(e/4 \beta)}}}$, and for~$W_0$ the inverse of~$w \mapsto w e^w$ on~$(0, \infty)$.
\end{theorem}
To prove the result, we generalize to the functional setup a novel argument recently put forward by~\cite{garivier2018kl} for obtaining a regret upper bound for the anytime MOSS policy of~\cite{degenne2016anytime}. As in the proof of Theorem~\ref{RegretBound} we need to replace arguments relying on concentration inequalities for the mean, and rely heavily on optional skipping arguments. Furthermore, in contrast to \cite{garivier2018kl}, we do not only consider the case~$\beta = 1/2$, but we show that the argument actually goes through for~$\beta > 1/4$, also expanding the range~$\beta > 1/2$ considered in \cite{degenne2016anytime}.\footnote{Interestingly, in the special case of the mean functional (with~$\mathscr{D} = D_{cdf}([0, 1])$), Theorem~\ref{RegretBoundfaMOSS} shows that the multiplicative constant~$113$ given in Theorem~3 of~\cite{degenne2016anytime} for~$\beta = 2.35/2$ can be improved to $(4.83 + 6.66 \times d(2.35/2) + \sqrt{2.35})\sqrt{\pi} \approx 32.5$.}  This establishes theoretical guarantees for parameter values close to~$1/4$, which turned out best in their numerical results (but for which no regret guarantees were provided). Finally, we note that while the upper bound just given is of the order~$\sqrt{n}$, and improves on the upper bound for the F-UCB policy in this sense, this is bought at a price: the multiplicative constant appearing in the upper bound is larger than that obtained in Theorem~\ref{RegretBound}. 

\section{Numerical illustrations}\label{sec:num}

We now illustrate the theoretical results established in this article by means of simulation experiments. Throughout this section, the treatment outcome distributions~$F^i$ will be taken from the Beta family, a parametric subset of~$D_{cdf}([0, 1])$, which has a long history in modeling income distributions; see, for example, \cite{thurow1970analyzing}, \cite{mcdonald1984some} and \cite{mcdonald2008generalized}. An appealing characteristic of the Beta family is its ability to replicate many ``shapes'' of distributions. We emphasize that the policies investigated do not exploit that the unknown treatment outcome distributions are elements of the Beta family.

Our numerical results cover different functionals~$\mathsf{T}$, with a focus on situations where the policy maker targets the distribution that maximizes welfare, and where we consider the case~$a = 0$ and~$b = 1$. In all our examples the feasible set for the marginal distributions of the treatment outcomes~$\mathscr{D} = D_{cdf}([0, 1])$. 

The specific welfare measures we consider are as follows (and correspond to the \mbox{Gini-,} Schutz- and Atkinson- inequality measure, respectively, through the transformations detailed in Appendix \ref{sec:welfaremeasures}, to which we refer the reader for more background information):

\begin{enumerate}
\item \emph{Gini-index-based welfare measure}: $\mathsf{W}(F)=\mu(F) - \frac{1}{2}\int \int |x_1 - x_2| dF(x_1) dF(x_2)$, where $\mu(F) := \int x dF(x)$ denotes the mean of~$F$. 

[Assumption \ref{as:MAIN} is satisfied with $\mathscr{D}= D_{cdf}([0,1])$ and $C=2$, cf.~the discussion after Lemma \ref{lem:derivedwelfare}.]
\item \emph{Schutz-coefficient-based welfare measure}: $\mathsf{W}(F)=\mu(F) - \frac{1}{2} \int |x - \mu(F)| dF(x)$. 

[Observing that  $\mathsf{W}(F)=\mu(F) -\mathsf{S}_{abs}(F)$ with~$\mathsf{S}_{abs}$ as defined in Equation~\eqref{eqn:schutzabs}, it follows from Lemmas \ref{lem:schutzindex} and \ref{lem:derivedwelfare} that Assumption \ref{as:MAIN} is satisfied with $\mathscr{D}= D_{cdf}([0,1])$ and $C=2$.]
\item \emph{Atkinson-index-based welfare measure}: $\mathsf{W}(F)= [\int x^{1-\eps}dF(x)]^{1/(1-\eps)}$ for a parameter $\eps\in (0,1)\cup (1,\infty)$. 

[Restricting attention to $\eps\in(0,1)$, the mean value theorem along with Example \ref{ex:pmean} in Appendix \ref{sec:LCgeneral} yield that Assumption \ref{as:MAIN} is satisfied with $\mathscr{D}= D_{cdf}([0,1])$ and $C=\frac{1}{1-\eps}$. We shall consider $\eps\in\cbr[0]{0.1,0.5}$.]
\end{enumerate}

In this section we consider two settings: (A) we compare the performance of explore-then-commit policies which do not incorporate~$n$ with the F-UCB and the F-aMOSS policy (which also do not incorporate~$n$); (B) as in (A) but where we now consider explore-then-commit policies that optimally incorporate~$n$. Throughout in this section, we consider the case of~$K = 2$ treatments. In the following, the symbol~$\mathsf{W}$ shall denote one of the welfare measures just defined in the above enumeration.

\subsection{Numerical results in Setting A}

In this setting the total number of assignments to be made is not known from the outset. Thus, the policies we study do not make use of the horizon~$n$. We consider explore-then-commit policies as in Section \ref{sec:ETC}, the F-UCB policy, and the F-aMOSS policy. While the F-UCB policy is implemented as in Policy~\ref{poly:FUCB} of Section~\ref{sec:F-UCBnoCov} with~$\beta = 2.01$, and the F-aMOSS policy is implemented as in Policy~\ref{poly:faMOSS} with~$\beta = 1/3.99$, the concrete development of explore-then-commit policies with certain performance guarantees requires some additional work which we develop next.

\subsubsection{Implementation details for explore-then-commit policies}\label{sec:numimplement}

In all explore-then-commit policies we consider, Treatments 1 and 2 are assigned cyclically in the exploration period. This ensures that the number of assignments to each treatment differs at most by~$1$ (cf.~also Example \ref{rem:RCTex} in Section \ref{sec:ETC}).\footnote{Investigating policies with randomized assignment in the exploration phase would necessitate running the simulations repeatedly, averaging over different draws for the assignments in the exploration phase. The numerical results are already quite computationally intensive, which is why we only investigate a cyclical assignment scheme. This scheme already reflects to a good extent the average behavior of a randomized assignment with equal assignment probabilities.\label{foot:cyc}} Given this specification, the policy maker must still choose i)~the length of the exploration period~$n_1$, and ii)~a commitment rule to be used after the exploration phase. The choice of~$n_1$ (while independent of~$n$) depends on the commitment rule, of which we now develop a test-based and an empirical-success-based variant:

\begin{enumerate}
	\item \textbf{ETC-T}: This policy is built around a \emph{test-based commitment rule}. That is, one uses a test for the testing problem ``equal welfare of treatments,'' i.e.,~$\mathsf{W}(F^1) = \mathsf{W}(F^2)$, in deciding which treatment to choose after the exploration phase. 

	Given a test that satisfies a pre-specified size requirement, the length of the exploration phase is chosen such that the power of the test against a certain deviation from the null (effect size) is at least of a desired magnitude. A typical desired amount of power against the deviation from the null of interest is 0.8 or 0.9.

	The deviation from the null that one wishes to detect is clearly context dependent. We refer to \cite{jacob1988statistical}, \cite{murphy2014statistical} and \cite{athey2017econometrics}, as well as references therein, for in-depth treatments of power calculations. 
	
	To make this approach implementable, we need to construct an appropriate test. Given~$\alpha\in(0,1)$, and for~$n_1 \geq 2$, we shall consider the test that rejects if (and only if)~$|\mathsf{W}(\hat{F}_{1,n_1})-\mathsf{W}(\hat{F}_{2,n_1})|\geq c_\alpha$ with~$c_\alpha=\sqrt{2\log(4/\alpha)C^2/\lfloor n_1/2\rfloor}$. Under the null, i.e., for every pair~$F^1$ and~$F^2$ in~$D_{cdf}([0,1])$ such that~$\mathsf{W}(F^1)=\mathsf{W}(F^2)$, this test has rejection probability at most~$\alpha$ (a proof of this statement is provided in Appendix~\ref{sec:numprovesize}). Hence, the size of this test does not exceed~$\alpha$. 
	
	For this test, in order to detect a deviation of~$\Delta:=|\mathsf{W}(F^1)-\mathsf{W}(F^2)|>0$ with probability at least~$1-\eta$, where~$\eta\in(0,1)$, it suffices that~$n_1 = 2\lceil\frac{8\log(4/\min(\alpha, \eta))C^2}{\Delta^2}\rceil$ (for a proof of this statement, see Appendix~\ref{sec:numprovepow}). 
	
	In our numerical studies we set~$\eta=\alpha=0.1$. We consider~$\Delta\in\cbr[0]{0.15,0.30}$, which amounts to a small and moderate desired detectable effect size, respectively. Note that while choosing~$\Delta$ small allows one to detect small differences in the functionals by the above test, this comes at the price of a larger~$n_1$. Thus, we shall see that neither~$\Delta=0.15$ nor~$\Delta=0.30$ dominates the other uniformly (over~$t\in\N$) in terms of maximal expected regret.   
	The commitment rule applied is to assign~$\argmax_{1\leq i\leq 2} \mathsf{W}(\hat{F}_{i,n_1})$ if the above test rejects, and to randomize the treatment assignment with equal probabilities otherwise. Finally, we sometimes make the dependence of ETC-T on~$\Delta$ explicit by writing ETC-T($\Delta$).
	
	\item \textbf{ETC-ES}: This policy assigns~$\pi_{n}^c(Z_{n_1}): = \min\argmax_{1\leq i\leq K} \mathsf{W}(\hat{F}_{i,n_1})$ to subjects~$t=n_1+1,\hdots,n$, which is an \emph{empirical success commitment rule} inspired by \cite{manski2004statistical} and \cite{manski2016sufficient}. Here, given a~$\delta>0$,~$n_1$ is chosen such that the maximal expected regret for every subject to be treated after the exploration phase is at most~$\delta$; i.e.,~$n_1$ satisfies~$$\sup_{\substack{F^i \in \mathscr{D} \\ i = 1, \hdots, K }}\E\del[1]{\max_{i\in\mathcal{I}}\mathsf{W}(F^i)-\mathsf{W}(F^{\pi_{n}^c(Z_{n_1})}}\leq \delta.$$ We prove in Appendix~\ref{sec:numprovereg} that~$n_1 =  2 \lceil 16C^2/(\delta^2 \exp(1))\rceil$ suffices. 
	
	In our numerical results, we consider~$\delta\in\cbr[0]{0.15,0.30}$, which should be contrasted to the treatment outcomes taking values in~$[0,1]$. Note that the~$n_1$ required to guarantee a maximal expected regret of at most~$\delta$ for every subject treated \emph{after} the exploration phase is decreasing in~$\delta$. Thus, we shall see that it need not be the case that choosing~$\delta$ smaller will result in lower overall maximal expected regret. Finally, we sometimes make the dependence of ETC-ES on~$\delta$ explicit by writing ETC-ES($\delta$).
\end{enumerate}

The following display summarizes the numerical implementation. 

\RestyleAlgo{boxed}

\medskip
\onehalfspacing
\begin{simulation}[H]

\textbf{Input:}~$n=100{,}000,\ r=20$ and\\
$\mathcal{G} = \{0.1,0.425,0.75,0.8,0.85,0.9,0.95,0.9625,0.975,0.9875, 1,1.0125,1.025,$\\$1.0375,1.05,1.10,1.15,1.20,1.25,3.125,5\}$

\For{$p_1\in\mathcal{G}$ such that $p_1 < 5$}{
	\For{$p_2\in\mathcal{G},\ p_2>p_1$}{
		\For{$l=1,\hdots,r$}{Generate~$n$ independent observations from~$\mathsf{Beta}(1,p_1) \otimes \mathsf{Beta}(1,p_2).$
		
		\For{$t=1,\hdots,n$}{Calculate the regret of each policy over all assignments~$s=1,\hdots,t$.
		}
	}
	Estimate expected regret for each policy and for~$t = 1, \hdots, n$ by the arithmetic mean of regret over the~$r$ data sets.
}
}
Estimate, for every~$t = 1, \hdots, n$, the maximal expected regret by maximizing the arithmetic means over the~$|\mathcal{G}|(|\mathcal{G}|-1)/2=210$ parameter vectors~$(p_1,p_2)$.  
\end{simulation}

\medskip

\doublespacing

\RestyleAlgo{ruled}

\medskip

Since maximizing expected regret over all Beta distributions would be numerically infeasible, we have chosen to maximize expected regret over a subset of all Beta distributions indexed by~$\mathcal{G}$ as defined in the previous display. We stress that since none of the three policies above needs to know~$n$, the numerical results also contain the maximal expected regret of the policies for any sample size less than~$n=100{,}000$. 

\subsubsection{Results}

The left panel of Figure~\ref{fig:Gini} illustrates the maximal expected regret for the F-UCB, F-aMOSS, ETC-T and ETC-ES policies in the case of Gini-welfare. Each point on the six graphs is the maximum of expected regret over the~$210$ different distributions considered at a given~$t$. In accordance with Theorems \ref{thm:LBETC}, \ref{RegretBound} and~\ref{RegretBoundfaMOSS}, the maximal expected regret of the policies in the explore-then-commit family is generally higher than the one of the F-UCB and the F-aMOSS policy. For~$t=100{,}000$, the maximal expected regret of F-UCB is~$498$, and~$159$ for F-aMOSS, while the corresponding numbers for ETC-T(0.15), ETC-ES(0.15), ETC-T(0.30) and ETC-ES(0.30) are~$4{,}248$, $777$, $7{,}281$ and~$836$, respectively. Note also that no matter the values of~$\Delta$ and~$\delta$, the maximal expected regret of ETC-ES($\delta$) is much lower than the one of the ETC-T($\Delta$) policy.\footnote{This result on the ranking of test-based vs.~empirical success-based commitment rules is similar to an analogous finding in a non-sequential setting in \cite{manski2016sufficient}.} In fact, we shall see for all functionals considered that the F-aMOSS policy generally incurs the lowest maximal expected regret, followed by the F-UCB policy and subsequently by the ETC-ES policies, which in turn perform much better than ETC-T policies. 

\begin{figure}
\centering
\includegraphics[height=7.5cm,width=7.5cm]{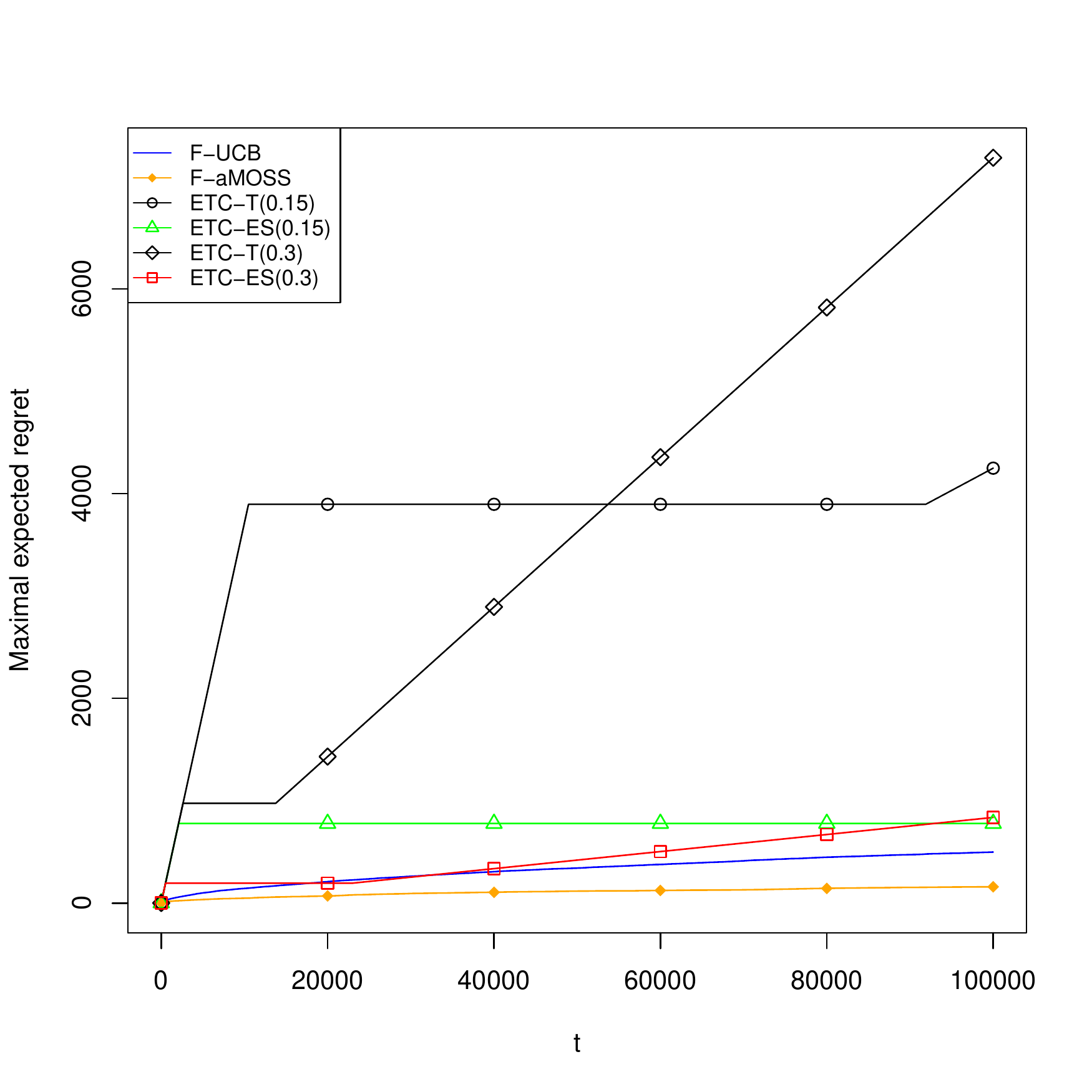}
\includegraphics[height=7.5cm,width=7.5cm]{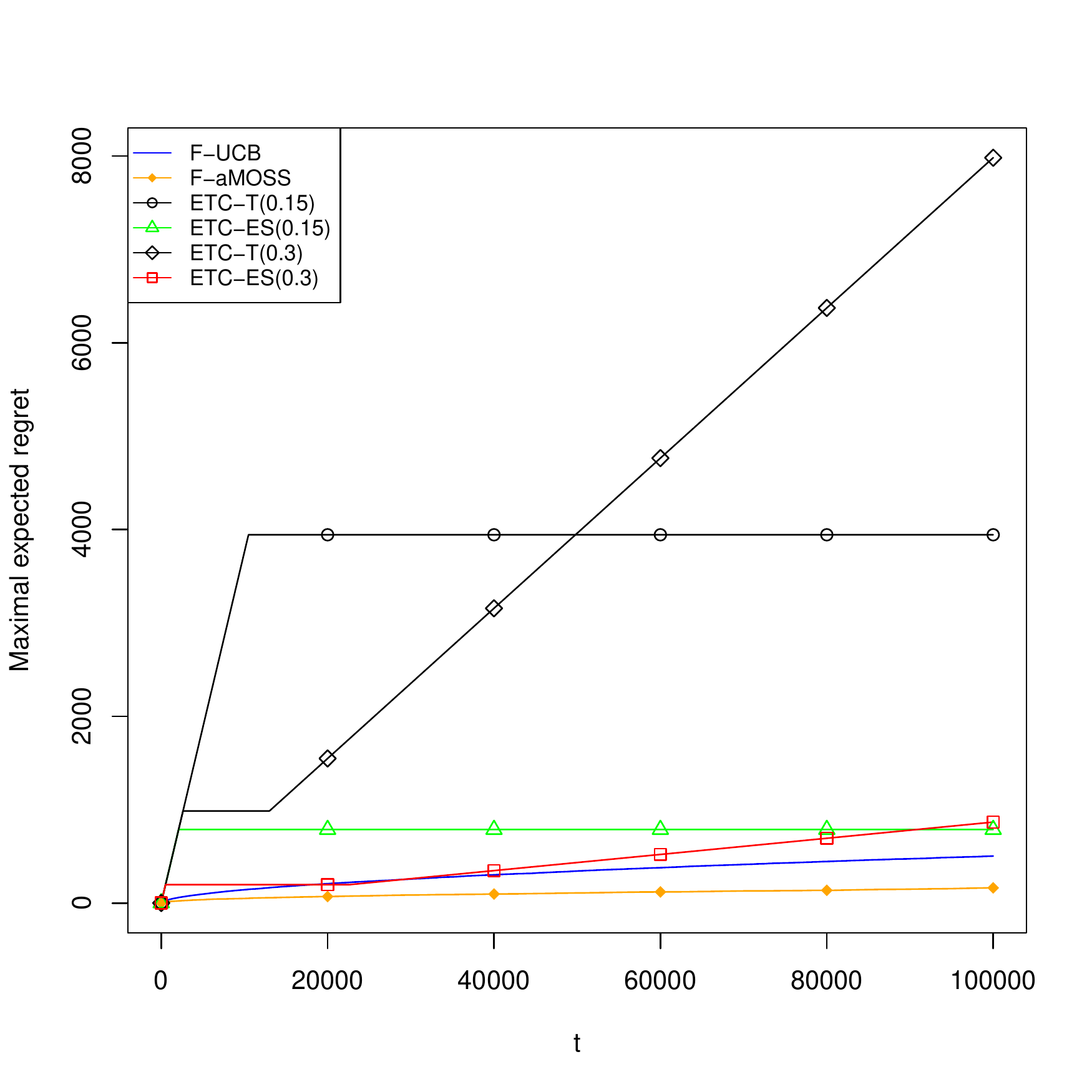}
\caption{\small The figure contains the maximal expected regret for F-UCB, F-aMOSS, ETC-T($\Delta$) with~$\Delta\in\cbr[0]{0.15,0.30}$ and ETC-ES($\delta$) with~$\delta\in\cbr[0]{0.15,0.30}$. The left panel is for Gini-welfare while the right panel is for Schutz-welfare.}
\label{fig:Gini}
\end{figure}

The shape of the graphs of the maximal expected regret of the explore-then-commit policies can be explained as follows: in the exploration phase maximal expected regret is attained by a distribution~$P_1$, say, for which the value of the Gini-welfare differs strongly at the marginals. However, such distributions are also relatively easy to distinguish, such that none of the commitment rules (testing or empirical success) assigns the suboptimal treatment after the exploration phase. This results in no more regret being incurred and thus a horizontal part on the maximal expected regret graph. For~$t$ sufficiently large, however, maximal expected regret will be attained by a distribution $P_2$, say, for which the marginals are sufficiently ``close'' to imply that the commitment rules occasionally assign the suboptimal treatment. For such a distribution, the expected regret curve will have a positive linear increase even after the commitment time~$n_1$ and this curve will eventually cross the horizontal part of the expected regret curve pertaining to~$P_1$. This implies that maximal expected regret increases again (as seen for ETC-T(0.30) around~$t=14{,}000$ and ETC-ES(0.30) around~$t=23{,}000$ in the left panel of Figure~\ref{fig:Gini}). Such a kink also occurs for ETC-T(0.15) and eventually also for ETC-ES(0.15). Thus, the left panel of Figure~\ref{fig:Gini} illustrates the tension between choosing~$n_1$ small in order to avoid incurring high regret in the exploration phase and, on the other hand, choosing~$n_1$ large in order to ensure making the correct decision at the commitment time.   

The right panel of Figure \ref{fig:Gini}, which contains the maximal expected regret for the Schutz-welfare, yields results qualitatively similar to the ones for the Gini-welfare. The best explore-then-commit policy again has a terminal maximal expected regret that is more than~$4.8$ times that of the F-aMOSS policy. 

\begin{figure}
\centering
\includegraphics[height=7.5cm,width=7.5cm]{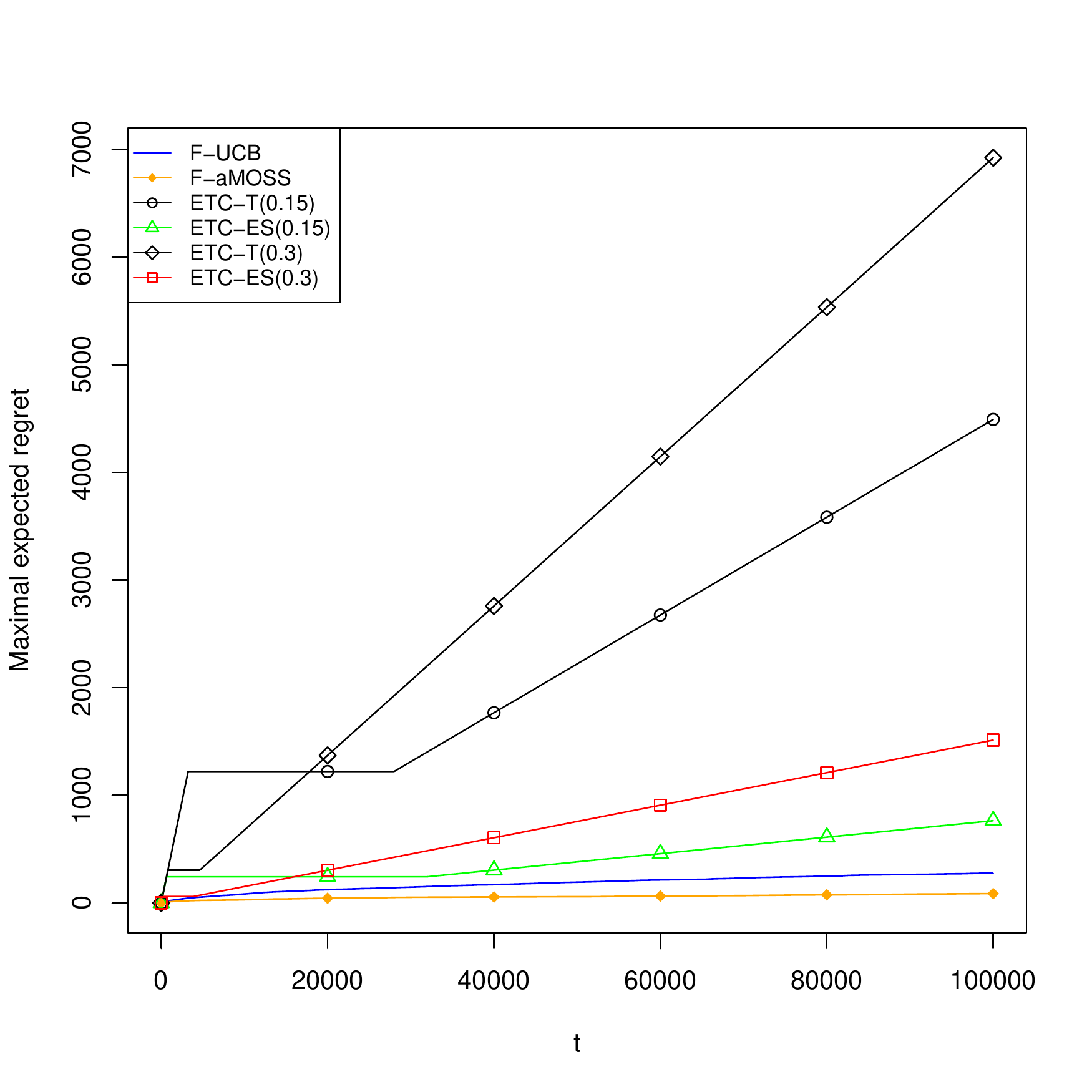}
\includegraphics[height=7.5cm,width=7.5cm]{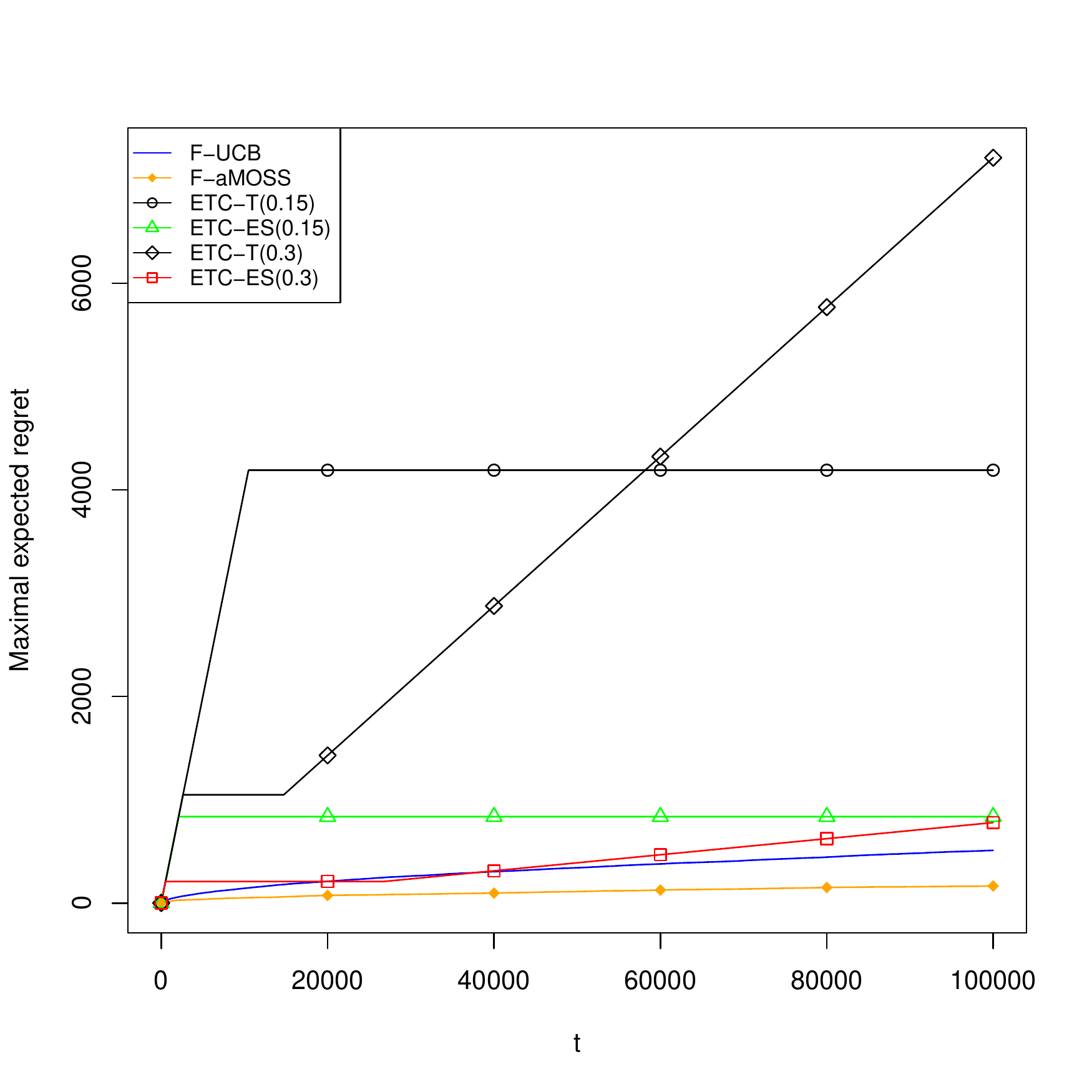}
\caption{\small The figure contains the maximal expected regret for F-UCB, F-aMOSS, ETC-T($\Delta$) with~$\Delta\in\cbr[0]{0.15,0.30}$ and ETC-ES($\delta$) with~$\delta\in\cbr[0]{0.15,0.30}$ in the case of Atkinson welfare. The left panel is for~$\eps=0.1$, while the right panel is for~$\eps=0.5$.}
\label{fig:Atkinson}
\end{figure}

We next turn to the two welfare measures in the Atkinson family. The left panel of Figure~\ref{fig:Atkinson} contains the results for the case of~$\eps=0.1$. While F-aMOSS incurs the lowest maximal expected regret uniformly over~$t=1,\hdots,100{,}000$, the most remarkable feature of the figure is that maximal expected regret of \emph{all} explore-then-commit policies is eventually increasing within the sample considered. The reason for this is that~$\eps=0.1$ implies a low value of $n_1$ such that i) the steep increase in maximal expected regret becomes shorter and ii) more mistakes are made at the commitment time. The ranking of the families of polices is unaltered with F-UCB and F-aMOSS dominating ETC-ES, which in turn incurs much lower regret than ETC-T. 

The right panel of Figure \ref{fig:Atkinson} considers the case of Atkinson welfare when~$\eps=0.5$. The findings are qualitatively similar to the ones for the Gini- and Schutz-based welfare measures. 

\subsection{Numerical results in Setting B}

\subsubsection{Implementation details}

In this section we compare the explore-then-commit Policy \ref{poly:ES} (but with cyclical assignment in the exploration phase, cf.~Footnote~\ref{foot:cyc}) with the F-UCB policy as implemented as in the previous subsection. Note that Policy \ref{poly:ES} depends on~$n$ in an optimal way, cf.~Theorem~\ref{thm:ETCES}, while the F-UCB policy does not incorporate~$n$.

\subsubsection{Results}

Table \ref{tab:nknown} contains maximal expected regret computations for Policy \ref{poly:ES} \emph{relative} to that of the F-UCB policy for $n\in\cbr[0]{1{,}000;5{,}000;10{,}000;20{,}000;40{,}000;60{,}000}$. Thus, numbers larger than 1 indicate that the F-UCB policy has lower maximal expected regret. Since Policy~\ref{poly:ES} is not anytime, cf.~Remark~\ref{rem:horizonpolicy}, to study its regret behavior it must be implemented and run anew for each $n$; i.e., for each~$n$ we proceed as in the display describing the implementation details for Setting A, but only record the terminal value of the numerically determined maximal expected regret. Producing a plot analogous to Figure~\ref{fig:Gini} but for a policy incorporating~$n$ would require us to run the simulation 100{,}000 times, i.e., one simulation per terminal sample size~$n$, which would be extremely computationally intensive.\footnote{To reduce the computational cost we set~$\mathcal{G} = \cbr[0]{0.1,0.75,0.85,0.95,0.975,1,1.025,1.05,1.15,1.25,5}$ in the results reported in the present section.} As can be seen from Table \ref{tab:nknown}, the F-UCB policy achieves lower expected regret than Policy \ref{poly:ES} at all considered horizons for all welfare measures even though the former policy does not make use of $n$ while the latter does. Note also that the relative improvement of the F-UCB policy over Policy \ref{poly:ES} is increasing in~$n$ as suggested by our theoretical results.

\bigskip

\begin{table}[h]
\centering
\begin{tabular}{ccccccc}
\toprule
$n$ &1{,}000&5{,}000&10{,}000& 20{,}000& 40{,}000& 60{,}000\\
\midrule
 Gini&1.96&2.30&2.49&2.64& 2.89 & 3.10  \\
 Schutz&2.02&2.34&2.47& 2.68 & 2.86  & 3.10  \\
 Atkinson, $\eps=0.1$&3.46&3.94&4.16& 4.77 & 5.34 & 5.20  \\
 Atkinson, $\eps=0.5$&2.21&2.48	&2.65& 2.86 & 3.05 & 3.27\\ 
\bottomrule 
\end{tabular}
\caption{Maximal expected regret of the explore-then-commit Policy \ref{poly:ES} relative to that of the F-UCB for $n\in\cbr[0]{1{,}000;5{,}000;10{,}000;20{,}000;40{,}000;60{,}000}$.}
\label{tab:nknown}
\end{table}
As shown by the simulation results reported in the previous section, using the F-aMOSS policy as a benchmark instead of the F-UCB policy would lead to even larger relative improvements over the explore-then-commit Policy \ref{poly:ES}.

\section{Illustrations with empirical data}\label{sec:app}

We here compare the performance of the policies using three (non-sequentially generated) empirical data sets, each containing the outcomes of a treatment program. From every data set we generate \emph{synthetic sequential data} by sampling from the empirical cdfs corresponding to the treatment/control groups. That is, the empirical cdfs in the data sets are taken as the respective (unknown) treatment outcome distributions~$F^1, \hdots, F^K$, from which observations are then drawn sequentially. This approach allows us to study the policies' performance on cdfs resembling specific characteristics arising in large scale empirical applications. The data sets considered are as follows; cf.~also Appendix \ref{sec:Data}.

\begin{enumerate}[leftmargin=1cm]
	\item \emph{The Cognitive Abilities}  program studied in \cite{hardy2015enhancing}. In this RCT, the participants were split into a treatment group who participated in an online training program targeting various cognitive capacities and an active control group solving crossword puzzles. Thus,~$K=2$. The outcome variable is a neuropsychological performance measure. 
	\item \emph{The Detroit Work First} program studied in \cite{autor2010temporary} and \cite{autor2017effect}. Here low-skilled workers took temporary help jobs, direct hire jobs or exited the program. Thus,~$K=3$. The outcome variable is the total earnings in quarters 2--8 after the start of the program. 
	\item \emph{The Pennsylvania Reemployment Bonus} program studied originally in \cite{bilias2000sequential} and also in, e.g., \cite{chernozhukov2018double}. The participants in the program are unemployed individuals who are either assigned to a control group, or to one of five treatment groups who receive a cash bonus if they find and retain a job within a given qualification period. Thus,~$K=6$. The size of the cash bonus and the length of the qualification period vary across the five treatment groups. The outcome variable is unemployment duration which varies from~1 to~52 weeks. 
\end{enumerate}

To facilitate the comparision with the other results in the previous section, all data were scaled to~$[0,1]$, and we consider the Gini-, Schutz- and two Atkinson-welfare measures. We focus on the Gini-based-welfare, and report the results for the remaining functionals in Appendix~\ref{sec:Data}. 
The reported expected regrets are averages over~$100$ replications, with~$n=100{,}000$ in each setup. When interpreting the results, it is important to keep in mind that in contrast to the maximal expected regret studied in Section \ref{sec:num}, where for each~$t$ the worst-case regret over a certain family of distributions is reported, the focus is now on three particular data sets, i.e., three instances of \emph{pointwise} expected regret w.r.t.~fixed distributions. 

For the Detroit Work First Program, for which~$K=3$, the ETC-ES policies are implemented as in Section \ref{sec:numimplement} with~$\delta\in\cbr[0]{0.15,0.30}$. For the Pennsylvania Reemployment Bonus experiment, where~$K=6$, this rule led to exploration periods exceeding~$n=100{,}000$. Hence, we instead considered exploration periods assigning~$250,\ 500,\ 750$ and~$1{,}000$ observations to each of the six arms, respectively. We only implemented the ETC-T policies for the cognitive ability program for which~$K=2$.\footnote{This is justified by the fact that these policies were always inferior in Section \ref{sec:num}. Furthermore, implementing the ETC-T policies when~$K>2$ would require taking a stance on how to control the size of the (multiple) testing problem at the commitment time.} The F-UCB and F-aMOSS policies are implemented as in Section~\ref{sec:num}.

\begin{figure}
	\centering
	\includegraphics[height=7.5cm,width=7.5cm]{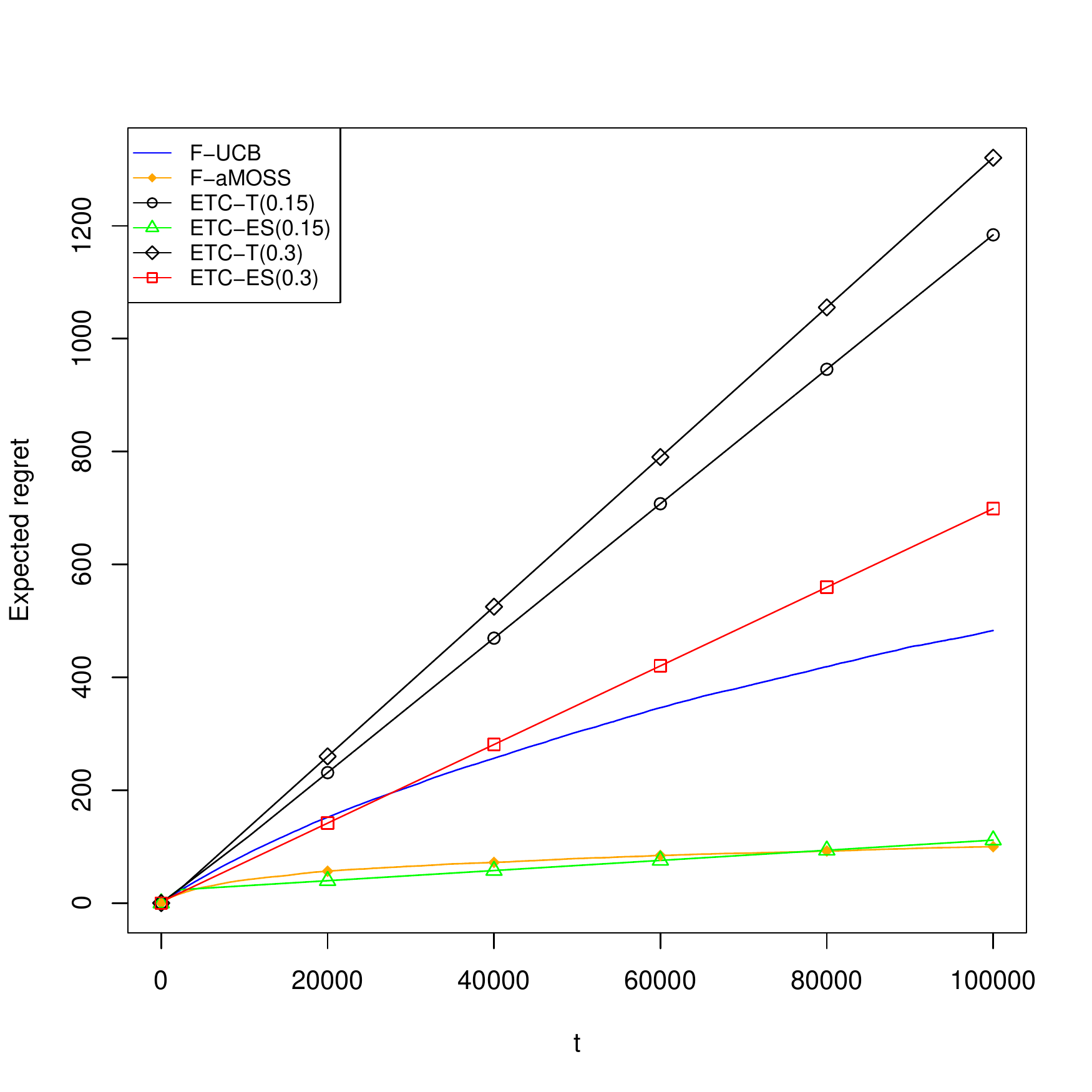}
	\includegraphics[height=7.5cm,width=7.5cm]{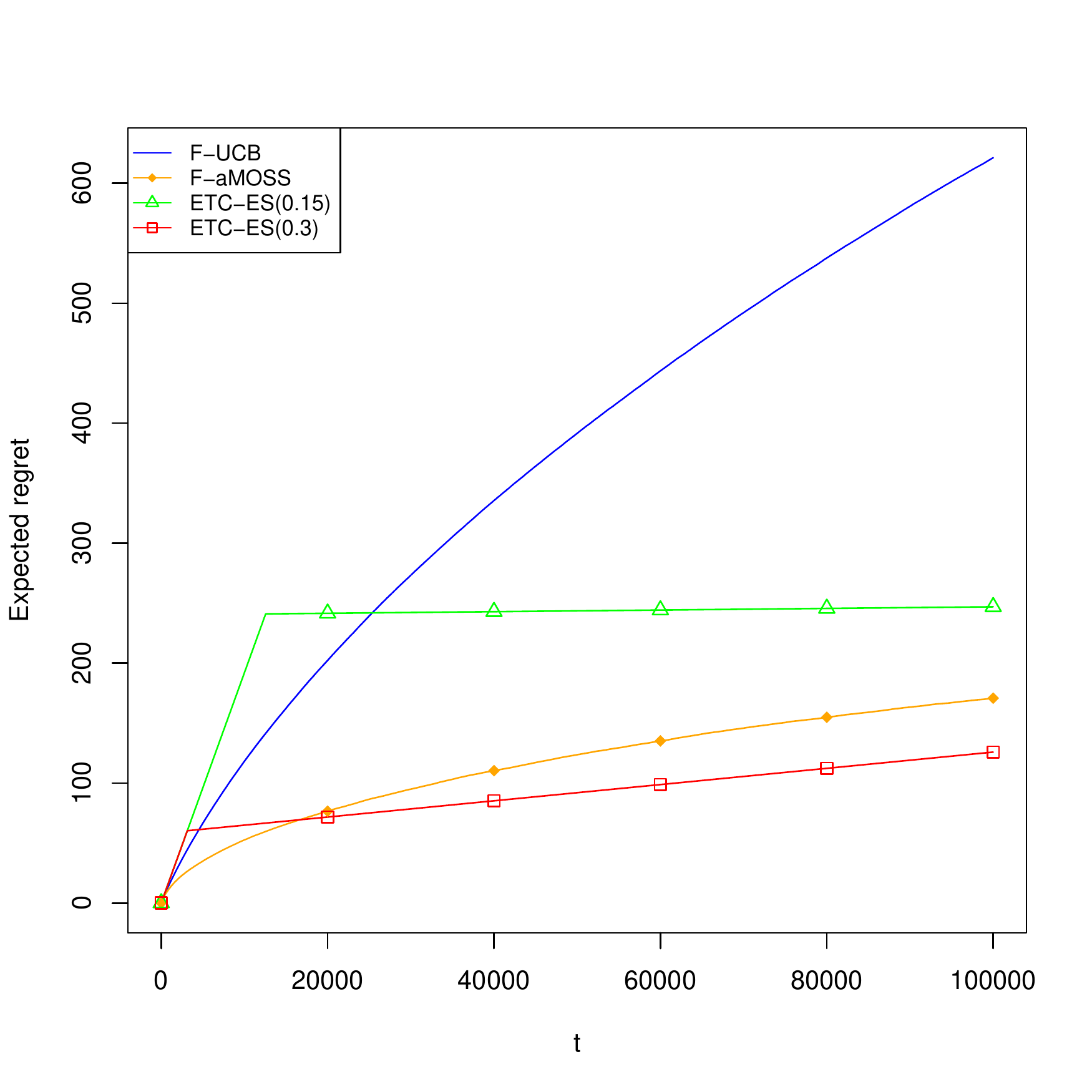}
	\includegraphics[height=7.5cm,width=7.5cm]{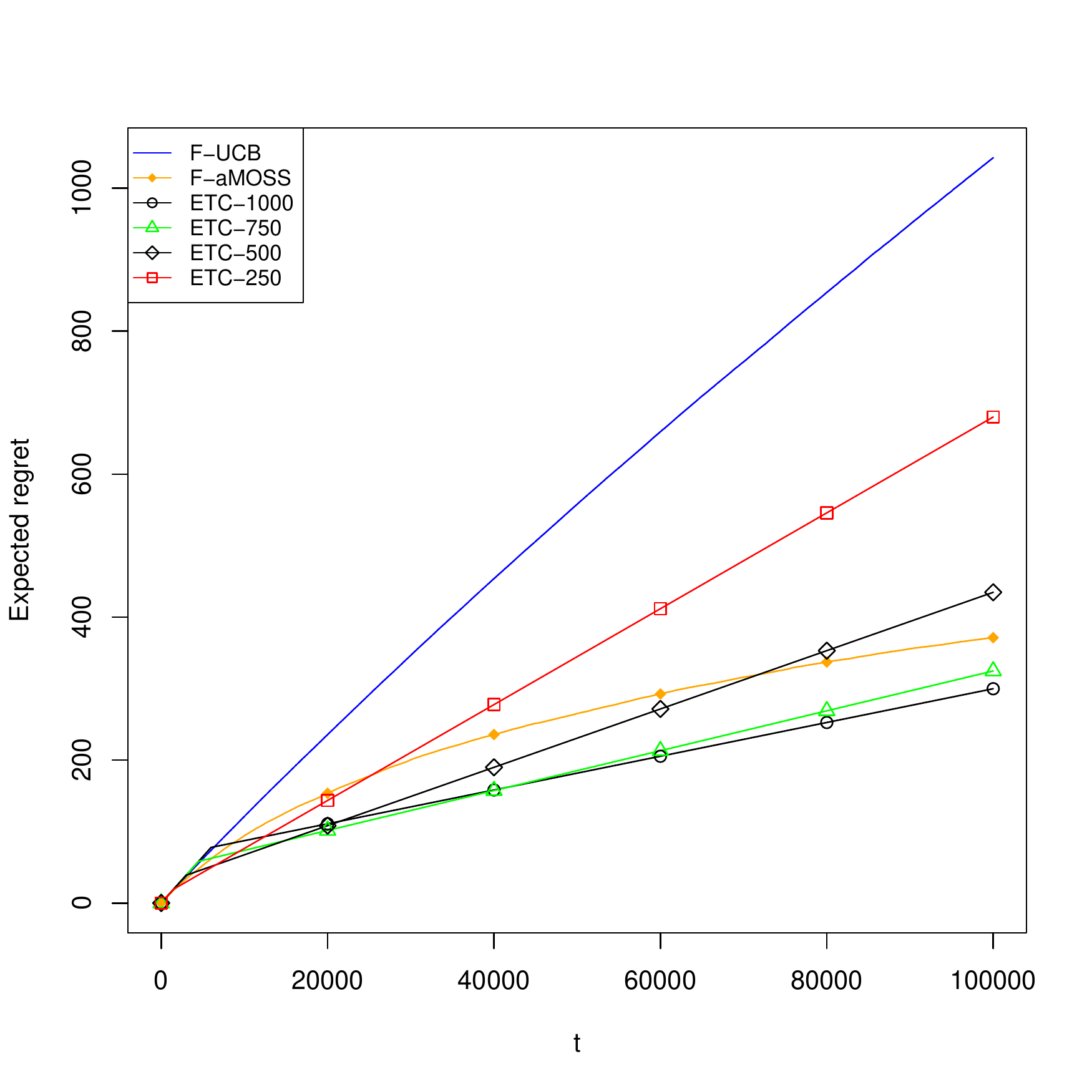}
	\caption{\small The figure contains the expected regret for the Gini-based-welfare measure. Top-left: Cognitive abilities program, top-right: Detroit work first program, bottom: Pennsylvania reemployment bonus program.}
	\label{fig:GiniApp}
\end{figure}
The results for the Gini-welfare measure are summarized in Figure \ref{fig:GiniApp}. The main take-aways are: i) F-aMOSS performs solidly across all data sets. ii) Except for the cognitive training program data, the F-UCB policy is not among the best. Note that this is not in contradiction to the theoretical results of this paper (nor the simulations in Section \ref{sec:num}) as these are concerned with the \emph{worst-case} performance of the policies. iii) There always exists an exploration horizon such that an ETC-ES policy incurs a low expected regret (over the sample sizes considered). However, this horizon is data dependent: Note that for the cognitive and Pennsylvania data sets long exploration horizons are preferable, while for the Work First data the opposite is the case. From the Pennsylvania data it is also seen that the optimal length of the exploration horizon depends on the length of the program. 

The figures containing the results for the remaining functionals are contained in Appendix~\ref{sec:Data}. For the Schutz- and Atkinson-welfare measure with~$\eps=0.5$ the results are qualitatively similar to those of the Gini-based welfare. Regarding the Atkinson-based welfare with~$\eps=0.1$, F-aMOSS and F-UCB now even incur the lowest regret for the cognitive data. For the Work First and Pennsylvania data F-aMOSS remains best (at the end of the program). It is interesting that for the Work First data the ordering of the two ETC-ES policies is reversed at the end of the treatment period compared to the remaining functionals. The latter observation again underscores the difficulty in getting the length of the exploration period ``right." This echoes our theoretical results showing that there is no way of constructing an ETC-based rule that would uniformly dominate the UCB-type policies.

\section{Conclusion}
In this paper we have studied the problem of a policy maker who assigns sequentially arriving subjects to treatments. The quality of a treatment is measured through a functional of the potential outcome distributions. Drawing crucially on the results and the framework developed in the present paper, the companion paper~\cite{kpv2} studies how the setting and regret notion can be adapted to allow for covariates, and  explores how those can be optimally incorporated in the decision process. 

\singlespacing

\bibliographystyle{ecta}	% (uses file "plain.bst")
\bibliography{mergereferences}		% expects file "myrefs.bib"

\newpage
\doublespacing
\appendix
\renewcommand\appendixpagename{Supplementary Online Appendices}
\appendixpage

Throughout the appendices, the (unique) probability measure on the Borel sets of $\R$ corresponding to a cdf~$F\in D_{cdf}(\R)$ will be denoted by~$\mu_{F}$ (cf., e.g., \cite{folland},~p.35). 

We shall freely use standard notation and terminology concerning stochastic kernels (also referred to as Markov kernels or probability kernels) and semi-direct products (i.e., the joint distribution corresponding to a stochastic kernel and a probability measure) see, e.g., Appendix~A.3 of \cite{liese} in particular their Equation~A.3. Furthermore, the random variables and vectors appearing in the proofs are defined on an underlying probability space~$(\Omega, \mathcal{A}, \P)$ with corresponding expectation~$\E$, which is (without loss of generality) assumed to be rich enough to support all random variables we work with. Furthermore, we shall denote by~$\omega$ a generic element of~$\Omega$.

We also recall from, e.g., Definition 2.5 in \cite{tsybakov2009introduction}, that the Kullback-Leibler divergence between two probability measures~$P$ and~$Q$ on a measurable space~$(\mathcal{X}, \mathfrak{Y})$ is defined as
\begin{equation}\label{eqn:DefKL}
\mathsf{KL}(P, Q) := 
\begin{cases}
\int_{\mathcal{X}} \log(dP/dQ) dP & \text{ if } P \ll Q, \\
\infty & \text{ else}.
\end{cases}
\end{equation}
The integral appearing in this definition is well-defined, because the negative part of the integrand is $P$-integrable. The positive part of the integrand is not necessarily~$P$-integrable. Therefore,~$\mathsf{KL}(P, Q) = \infty$ might hold even in case~$P \ll Q$. Furthermore,~$\mathsf{KL}(P, Q)$ is non-negative, and equals~$0$ if and only if~$P = Q$. Proofs for the just-mentioned facts can be found in Section~2.4 of \cite{tsybakov2009introduction}. Note that the definition of~$\mathsf{KL}$ does not depend on how one defines~$\log(0)$ (for completeness, we set~$\log(0):=0$ in the sequel).

\section{Auxiliary results}\label{sec:Aux}
This section develops some auxiliary lemmas that will be used in Appendix~\ref{sec:maintextproofs}. The following result is a general ``chain rule'' for Kullback-Leibler divergences. Although well-documented under stronger assumptions, we could not find a reference containing a proof of the following general statement. 

\begin{lemma}[``Chain rule'' for Kullback-Leibler divergence]\label{lem:CHAIN}
Let~$(\mathcal{X}, \mathfrak{A})$ and~$(\mathcal{Y}, \mathfrak{B})$ be measurable spaces. Suppose that~$\mathfrak{B}$ is countably generated. Let~$\mathsf{A}, \mathsf{B} : \mathcal{B} \times \mathcal{X} \to [0, 1]$ be stochastic kernels, and let~$P$ and~$Q$ be probability measures on~$(\mathcal{X}, \mathfrak{A})$. Then,
\begin{equation}\label{eqn:CHAIN}
\mathsf{KL}(\mathsf{A} \otimes P, \mathsf{B} \otimes Q) = \int_{\mathcal{X}} \mathsf{KL}(\mathsf{A}(\cdot,x), \mathsf{B}(\cdot,x))dP(x) + \mathsf{KL}(P, Q) = \mathsf{KL}(\mathsf{A} \otimes P, \mathsf{B} \otimes P) + \mathsf{KL}(P, Q).
\end{equation}
\end{lemma}
\begin{remark}\label{rem:countgen}
Inspection of the proof of Lemma~\ref{lem:CHAIN} shows that the assumption of~$\mathfrak{B}$ being countably generated is only used to verify (via Proposition~1.95 in \cite{liese}) that~(i)~$x \mapsto \mathsf{KL}(\mathsf{A}(\cdot,x), \mathsf{B}(\cdot,x))$ is  measurable, and~(ii) that~$\int_{\mathcal{X}} \mathsf{KL}(\mathsf{A}(\cdot,x), \mathsf{B}(\cdot,x))dP(x)$ coincides with~$\mathsf{KL}(\mathsf{A} \otimes P, \mathsf{B} \otimes P)$. In situations where~$\mathfrak{B}$ fails to be countably generated, the conclusion in the previous lemma still holds if~(i) and~(ii) are satisfied.
\end{remark}

\begin{proof}
We conclude from Proposition~1.95 in \cite{liese} that the integral $\int_{\mathcal{X}} \mathsf{KL}(\mathsf{A}(\cdot,x), \mathsf{B}(\cdot,x))dP(x)$ appearing in Equation~\eqref{eqn:CHAIN} is well-defined (i.e., the non-negative integrand~$x \mapsto \mathsf{KL}(\mathsf{A}(\cdot,x), \mathsf{B}(\cdot,x))$ is measurable), and coincides with~$\mathsf{KL}(\mathsf{A} \otimes P, \mathsf{B} \otimes P)$. This proves the second equality in Equation~\eqref{eqn:CHAIN}. 

To prove the first equality in Equation~\eqref{eqn:CHAIN}, assume first that~$\mathsf{A} \otimes P \not \ll \mathsf{B} \otimes Q$. Then,~$\mathsf{KL}(\mathsf{A} \otimes P, \mathsf{B} \otimes Q) = \infty$ by definition of the~$\mathsf{KL}$-divergence.
Observe that if~$P \ll Q$ and~$\mathsf{A}(\cdot, x) \ll \mathsf{B}(\cdot, x)$ for~$P$-almost every~$x$ would hold, then~$\mathsf{A} \otimes P \ll \mathsf{B} \otimes Q$ would follow. Therefore, either~$P \not \ll Q$ holds, or~$P \ll Q$ and~$A(\cdot,x) \not \ll B(\cdot,x)$ for all~$x$ in a set of positive~$P$ measure. In both cases the statement in the first equality in Equation~\eqref{eqn:CHAIN} holds true by definition and non-negativity of the~$\mathsf{KL}$-divergence. 

Consider now the case~$\mathsf{A} \otimes P \ll \mathsf{B} \otimes Q$. Corollary~1.71 of \cite{liese} implies~$\mathsf{KL}(\mathsf{A} \otimes P, \mathsf{B} \otimes Q) \geq \mathsf{KL}(P, Q)$. Therefore, if~$\mathsf{KL}(P, Q) = \infty$ Equation~\eqref{eqn:CHAIN} holds true. Hence, we can assume that~$\mathsf{KL}(P, Q) < \infty$. Choose a density~$0 \leq a := d(\mathsf{A} \otimes P)/d(\mathsf{B} \otimes Q)$, and let~$p := dP/dQ$ denote the corresponding (marginal)~$Q$-density of~$P$. Denote by~$[\log(a)]^+$ and~$[\log(a)]^-$ the positive and negative parts, respectively, of~$\log(a)$. The negative-part~$[\log(a)]^-$ is~$\mathsf{A} \otimes P$-integrable (cf.~the discussion immediately after Equation~\eqref{eqn:DefKL}). Furthermore, since~$\mathsf{KL}(P, Q)$ is finite,~$\log(p)$ is~$\mathsf{A} \otimes P$-integrable, implying that~$[\log(a)]^- + \log(p)$ is~$\mathsf{A} \otimes P$-integrable, and we can thus write~$\mathsf{KL}(\mathsf{A} \otimes P, \mathsf{B} \otimes Q) - \mathsf{KL}(P, Q)$ (the first summand might be infinite) as
\begin{align*}
&\int_{\mathcal{X} \times \mathcal{Y}} [\log(a)]^+ d(\mathsf{A} \otimes P)  -\left[  \int_{\mathcal{X} \times \mathcal{Y}} [\log(a)]^- d(\mathsf{A} \otimes P) + \int_{\mathcal{X} \times \mathcal{Y}} \log(p) d(\mathsf{A} \otimes P) \right] \\
= &
\int_{\mathcal{X} \times \mathcal{Y}} [\log(a)]^+ d(\mathsf{A} \otimes P)  + \int_{\mathcal{X} \times \mathcal{Y}} -\left([\log(a)]^- + \log(p)\right) d(\mathsf{A} \otimes P),
\end{align*}
which, since~$[\log(a)]^+$ is clearly non-negative and measurable, equals (cf., e.g., Theorem~4.1.10 in \cite{dudley})
\begin{equation*}
\int_{\mathcal{X} \times \mathcal{Y}} [\log(a) - \log(p)] d(\mathsf{A} \otimes P) = \int_{\mathcal{X} \times \mathcal{Y}} \log(a/p) \mathds{1}\{p > 0\} d(\mathsf{A} \otimes P),
\end{equation*}
the equality following from~$\{(x,y): a(y,x) = 0 \text{ or } p(x) = 0\}$ being an~$\mathsf{A} \otimes P$-null set. Since~$(a/p) \mathds{1}\{p > 0\} = d(\mathsf{A}\otimes P)/d(\mathsf{B} \otimes P)$, the right-hand side in the previous display equals~$\mathsf{KL}(\mathsf{A}\otimes P,\mathsf{B}\otimes P)$, establishing that
\begin{equation}
\mathsf{KL}(\mathsf{A} \otimes P, \mathsf{B} \otimes Q) = \mathsf{KL}(\mathsf{A}\otimes P,\mathsf{B}\otimes P) + \mathsf{KL}(P, Q).
\end{equation}
The already established second equality in Equation~\eqref{eqn:CHAIN} thus establishes the first.
\end{proof}

\begin{lemma}\label{lem:mixKL}
Consider probability measures~$\mu_i$ and~$\nu_i$ for~$i = 0, \hdots, m$ on a countably generated measurable space~$(\mathcal{Y}, \mathfrak{B})$. Set~$\mu := \sum_{i = 0}^m p_i \mu_i$ and~$\nu := \sum_{i = 0}^m q_i \nu_i$, where~$p_i \geq 0$ and~$q_i > 0$ hold for every~$i = 0, \hdots, m$ and~$\sum_{i = 0}^n p_i = 1 = \sum_{i = 0}^n q_i$. Then~$$\mathsf{KL}(\mu, \nu) 
\leq
\sum_{i = 0}^m \left(p_i \mathsf{KL}(\mu_i, \nu_i) +  (p_i - q_i)^2/q_i\right).$$
\end{lemma}

\begin{proof}
Define stochastic kernels~$\mathsf{A}: \mathcal{B} \times \{0, \hdots, m\} \to [0, 1]$ and~$\mathsf{B}: \mathcal{B} \times \{0, \hdots, m\} \to [0, 1]$ via~$\mathsf{A}(A, i) = \mu_i(A)$ and~$\mathsf{B}(A, i) = \nu_i(A)$, respectively. Let~$P$ be the measure on the power set of~$\{0, \hdots, m\}$ defined via~$P(i) = p_i$, and let~$Q$ be the measure on the power set of~$\{0, \hdots, m\}$ defined via~$Q(i) = q_i$. From Corollary~1.71 in \cite{liese} and the Chain Rule from Lemma~\ref{lem:CHAIN} we obtain 
\begin{equation}
\mathsf{KL}(\mu, \nu) \leq \mathsf{KL}(\mathsf{A} \otimes P, \mathsf{B} \otimes Q) = \sum_{i= 0}^m p_i \mathsf{KL}(\mu_i, \nu_i) + \mathsf{KL}(P, Q).
\end{equation}
But~$\mathsf{KL}(P, Q)$ is not greater than~$\chi^2(P,Q)$, the~$\chi^2$-divergence between~$P$ and~$Q$ (cf., e.g., Lemma~2.7 in~\cite{tsybakov2009introduction}), the latter being equal to~$\sum_{i = 0}^m (p_i-q_i)^2/q_i$.
\end{proof}

\begin{lemma}\label{lem:LBdist}
Suppose Assumption~\ref{as:NCONSTLINE} holds. Then there exist~$H$~and~$H'$ in~$D_{cdf}([a,b])$,~$c_- > 0$ and~$\varepsilon \in (0, 1/2)$ such that the following properties hold:
\begin{enumerate}
\item Letting~$H_v:=(1/2-v)H +(1/2+v)H'$, the set $\mathcal{H}:=\{H_v: v\in[-1/2,1/2]\}$ is contained in $\{J_{\tau}:\tau\in[0, 1]\}$.
\item The function $v\mapsto \mathsf{T}(H_v)$ defined on~$[-1/2, 1/2]$ is Lipschitz continuous.
\item For every~$v\in [0,\eps]$ it holds that
\begin{align}\label{eq:TLB}
\mathsf{T}(H_0)-\mathsf{T}(H_{-v})\geq c_{-}v\quad \text{ and }\quad\mathsf{T}(H_v)-\mathsf{T}(H_0)\geq c_{-}v,
\end{align}
and that
\begin{align}\label{eq:KLsubQuad}
\mathsf{KL}^{1/2}(\mu_{H_{-v}},\mu_{H_{v}})\leq \frac{2}{\sqrt{0.5^2-\eps^2}}v.
\end{align}
\end{enumerate} 
\end{lemma}
\begin{proof}
Without loss of generality, we can assume that $H_1$ and $H_2$ in Assumption \ref{as:NCONSTLINE} satisfy~$\mathsf{T}(H_2) < \mathsf{T}(H_1)$; otherwise swap the indices. From~Assumption~\ref{as:MAIN}, which is imposed through Assumption~\ref{as:NCONSTLINE}, it follows that the function~$h(\tau) := \mathsf{T}(J_{\tau})$ for~$\tau \in [0, 1]$ is Lipschitz continuous (recall the definition of~$J_{\tau}$ from Equation~\eqref{eqn:NCONSTLINE}), and hence almost everywhere differentiable. Furthermore, since~$\mathsf{T}(H_2) < \mathsf{T}(H_1)$, the derivative of~$h$ must be positive at some point~$\tau^* \in (0, 1)$, say. Consequently, there exists a~$c > 0$ (e.g., half the derivative of~$h$ at~$\tau^*$) and an~$\varepsilon \in (0, 1/2)$ satisfying $[\tau^*-\eps, \tau^* + \eps] \subseteq (0, 1)$, such that
\begin{equation}\label{eqn:subplin}
\frac{h(\tau) - h(\tau^*)}{\tau - \tau^*} \geq c \quad \text{ for every } \tau \in [\tau^*-\eps, \tau^* + \eps] \setminus \{\tau^*\}.
\end{equation}
Finally, let~$H := J_{\tau^* - \eps}$ and~$H' := J_{\tau^* + \eps}$ and set~$c_- = 2c\eps$. Note that~$H_v = J_{\tau^* + 2v\eps}$ for every~$v \in [-1/2, 1/2]$. Hence, the first part of the present lemma follows. The second part follows from~$\mathsf{T}(H_v) = \mathsf{T}(J_{\tau^* + 2v\eps}) = h(\tau^* + 2v\eps)$, recalling that~$h$ is Lipschitz continuous. The statements in Equation~\eqref{eq:TLB} follow immediately from Equation~\eqref{eqn:subplin}. Lemma~\ref{lem:mixKL} (applied with~$m = 1$,~$\mu_0 = \nu_0 = \mu_{H}$ and~$\mu_1 = \nu_1 = \mu_{H'}$, $p_0 = 1/2 + v$ and~$q_0 = 1/2 - v$) and a simple calculation shows that~$\mathsf{KL}(\mu_{H_{-v}}, \mu_{H_v}) \leq \frac{4v^2}{0.5^2-\varepsilon^2}$, which establishes~\eqref{eq:KLsubQuad}.
\end{proof}

\section{Proofs of results in Sections~\ref{sec:ETC},~\ref{sec:F-UCBnoCov} and~\ref{sec:num}}\label{sec:maintextproofs}

\subsection{Proofs of results in Section~\ref{sec:ETC}}

\subsubsection{Proof of Theorem~\ref{thm:LBETC}}

Let~$\pi$ be an explore-then-commit policy as in Definition~\ref{def:etc} that satisfies the corresponding exploration condition with~$\eta \in (0, 1)$. Fix~$n\geq 2$, and fix the randomization measure~$\P_G$ (i.e., a probability measure on the Borel sets of~$\R$). Since~$n$ is fixed, we shall abbreviate~$\pi_{n,t}=\pi_{t}$ in the sequel. Furthermore, we write~$\pi^c = \pi_n^c$. By Lemma~\ref{lem:LBdist} there exists a one-parametric family~$\mathcal{H}=\cbr[0]{H_v:v\in[-1/2,1/2]}\subseteq \{J_{\tau}: \tau \in [0, 1]\} \subseteq \mathscr{D}$, a real number~$c_- > 0$, and an~$\eps\in (0,1/2)$, such that for every~$v\in [0,\eps]$
\begin{equation}\label{eq:Hprop}
\mathsf{T}(H_0)-\mathsf{T}(H_{-v})\geq c_{-}v, \mathsf{T}(H_v)-\mathsf{T}(H_0)\geq c_{-}v, \text{ and }  \mathsf{KL}^{1/2}(\mu_{H_{-v}},\mu_{H_{v}})\leq \frac{2}{\sqrt{0.5^2-\eps^2}}v.
\end{equation}
Fix~$v \in (0, \varepsilon]$. We need some further notation: For~$j \in \{-v,v\}$ and every~$t = 1, \hdots, n$, we denote by~$\P_{\pi,j}^t$ the distribution of~$Z_t$ (as defined in Section~\ref{sec:setup}) on the Borel sets of~$\R^{2t}$, for~$Y_t$ i.i.d.~$\mu_{H_0} \otimes \mu_{H{j}}$ and~$G_t$ i.i.d.~$\P_G$. The expectation corresponding to~$\P_{\pi,j}^t$ will be denoted by~$\E_{\pi,j}^t$. We shall use~$z_t \in \R^{2t}$ as a generic symbol for a realization of~$Z_t$, and~$g_t\in\R$ as a generic symbol for a realization of~$G_t$. Furthermore, we denote by~$R_n^j({\pi})$ the regret of policy~$\pi$ under~$Y_t$ i.i.d.~$\mu_{H_0} \otimes \mu_{H{j}}$ and~$G_t$ i.i.d.~$\P_G$. We abbreviate~$n_1(n) = n_1$ and~$n_2 = n-n_1$. 

From Equation~\eqref{eq:Hprop} we conclude~$\mathsf{T}(H_{-v}) < \mathsf{T}(H_0) < \mathsf{T}(H_{v})$. Hence, Treatment~2 is inferior under~$\mu_{H_0} \otimes \mu_{H_{-v}}$, but superior under~$\mu_{H_0} \otimes \mu_{H_{v}}$. Therefore, recalling the definition of $S_{i,n}(t)$ and the corresponding notational convention in case~$t = n$ from Equation~\eqref{eqn:Sintdef} and using the expression for $R_n(\pi)$ given in Equation~\eqref{eq:regret2new}, we obtain (with some abuse of notation\footnote{Here and at many other places in the appendices, it is occasionally convenient to interpret quantities such as~$R_n^j(\pi)$ and~$S_{1}(n)$ as functions on the image space of $(Z_{n-1}, G_n)$, as opposed to the random variables obtained by plugging~$(Z_{n-1}, G_n)$ into these functions.})
\begin{align*}
\sup_{j\in\{-v,v\}}\E_{\pi,j}^nR_n^j(\pi)
&\geq
\frac{1}{2}\del[1]{\E_{\pi,-v}^nR_n^{-v}(\pi)+\E_{\pi,v}^nR_n^{v}(\pi)}\\
&=
\frac{1}{2} \left(
(\mathsf{T}(H_0)-\mathsf{T}(H_{-v})) \E_{\pi,-v}^nS_{2}(n) +(\mathsf{T}(H_v)-\mathsf{T}(H_{0}))\E_{\pi,v}^nS_{1}(n)\right)\\
&\geq
\frac{c_-v}{2}\del[2]{\E_{\pi,-v}^nS_{2}(n)+\E_{\pi,v}^nS_{1}(n)},
\end{align*}
where the third inequality follows from~\eqref{eq:Hprop}. Using Definition~\ref{def:etc}, the last expression equals
\begin{equation}
\frac{c_-v}{2} \left( \E_{\pi,-v}^{n_1}S_{2, n}(n_1)+\E_{\pi,v}^{n_1}S_{1,n}(n_1) \right) +
\frac{c_-v}{2}n_2\del[2]{\E_{\pi,-v}^{n_1}\mathds{1}\cbr[0]{\pi^c(z_{n_1})=2}+\E_{\pi,v}^{n_1}\mathds{1}\cbr[0]{\pi^c(z_{n_1})=1}};
\end{equation}
furthermore, for~$j\in\cbr[0]{-v,v}$ and~$i=1,2$, it holds that~$\E_{\pi,j}^{n_1}S_{i,n}(n_1)\geq \eta n_1$. Therefore, 
\begin{align}\label{eq:ObservationAfterEq}
\frac{c_-v}{2} \left( \E_{\pi,-v}^{n_1}S_{2, n}(n_1)+\E_{\pi,-v}^{n_1}S_{1,n}(n_1) \right)
\geq
c_-v\eta n_1.
\end{align}
%Thus, in case~$n_1 = n$, upon setting~$v=\eps/2$, it follows that~$\sup_{j\in\{-v,v\}}\E_{\pi,j}^nR_n^j(\pi) \geq c_-\eps \eta n/2$. We shall therefore assume~$n_1<n$ in the sequel. 
Noting that
\begin{align*}
&\E_{\pi,-v}^{n_1}\mathds{1}\cbr[0]{\pi^c(z_{n_1})=2}+\E_{\pi,v}^{n_1}\mathds{1}\cbr[0]{\pi^c(z_{n_1})=1} = \E_{\pi,-v}^{n_1}\mathds{1}\cbr[0]{\pi^c(z_{n_1})=2}+1-\E_{\pi,v}^{n_1}\mathds{1}\cbr[0]{\pi^c(z_{n_1})=2},
\end{align*}
which is the sum of Type 1 and Type 2 errors of the test~$\mathds{1}\cbr[0]{\pi^c(z_{n_1})=2}$ for the testing problem~$\P_{\pi,-v}^{n_1}$ against~$\P_{\pi,v}^{n_1}$, it follows from Theorem 2.2(iii) in \cite{tsybakov2009introduction} that
\begin{align*}
\E_{\pi,-v}^{n_1}\mathds{1}\cbr[0]{\pi^c(z_{n_1})=2}+\E_{\pi,v}^{n_1}\mathds{1}\cbr[0]{\pi^c(z_{n_1})=1}
\geq \frac{1}{4}
\exp\del[1]{-\mathsf{KL}(\P_{\pi,-v}^{n_1},\P_{\pi,v}^{n_1})}.
\end{align*}
Summarizing, we obtain
\begin{equation}\label{eqn:ETCKLlowerbound}
\begin{aligned}
\sup_{\substack{F^i \in \{J_{\tau} : \tau \in [0,1]\} \\ i = 1, 2}}\E [R_n(\pi)]
&\geq \frac{c_-v\eta}{8} \left[n_1 + (n-n_1) \exp\del[1]{-\mathsf{KL}(\P_{\pi,-v}^{n_1},\P_{\pi,v}^{n_1})}\right].
\end{aligned}
\end{equation}

To obtain an upper bound on~$\mathsf{KL}(\P_{\pi,-v}^{n_1}, \P_{\pi,v}^{n_1})$ we argue as follows: Let~$j \in \{-v,v\}$, let~$Y_t$ be i.i.d.~$\mu_{H_0} \otimes \mu_{H{j}}$, and let~$G_t$ be i.i.d.~$\P_G$. Let~$t \in \{1, \hdots, n_1\}$. It is easy to verify that the stochastic kernel 
\begin{equation}
(A, (g_{t},z_{t-1})) \mapsto \mu_{H_0}(A) \mathds{1}_{\cbr[0]{\pi_{t}(z_{t-1},g_{t})=1}}+ \mu_{H_j}(A) \mathds{1}_{\cbr[0]{\pi_{t}(z_{t-1},g_{t})=2}} 
\end{equation}
defines a regular conditional distribution (as defined in, e.g., \cite{liese} Definition~A.36) of~$Y_{\pi_{t}(Z_{t-1}, G_{t}), t}$ given~$(G_{t}, Z_{t-1})$ (dropping the quantities with index~$t-1$ in case~$t = 1$). Now, since the joint distribution of~$(G_{n_1}, Z_{n_1-1})$ is~$\P_G  \otimes \mathbb{P}_{\pi, j}^{n_1-1}$, we can write~$\mathbb{P}_{\pi, j}^{n_1}$, the joint distribution of~$(Y_{\pi_{n_1}(Z_{n_1-1}, G_{n_1}), n_1}, G_{n_1}, Z_{n_1 - 1})$, as the semi-direct product
\begin{equation*}
\mathbb{P}_{\pi, j}^{n_1} = \left(\mu_{H_0}\mathds{1}_{\cbr[0]{\pi_{n_1}(z_{n_1-1},g_{n_1})=1}}+ \mu_{H_j} \mathds{1}_{\cbr[0]{\pi_{n_1}(z_{n_1-1},g_{n_1})=2}} \right) \otimes (\P_G \otimes \mathbb{P}_{\pi, j}^{n_1-1} ),
\end{equation*}
where in case $n_1 = 1$ the arguments~$z_{n_1-1}$ and the factor~$\P_{\pi, j}^{n_1-1}$ need to be dropped. The chain rule in Lemma~\ref{lem:CHAIN} applied multiple times (and Tonelli's theorem) hence implies
\begin{align*}
\mathsf{KL}(\P_{\pi,-v}^{n_1}, \P_{\pi,v}^{n_1})
&=
\mathsf{KL}(\P_G \otimes \P_{\pi,-v}^{n_1-1} , \P_G \otimes \P_{\pi,v}^{n_1-1} )+\E_{\pi,-v}^{n_1-1}\E_G \left(
\mathds{1}_{\cbr[0]{\pi_{n_1}(z_{n_1-1},g_{n_1})=2}} \right) 
\mathsf{KL}(\mu_{H_{-v}},\mu_{H_{v}})\\
&\leq
\mathsf{KL}(\P_{\pi,-v}^{n_1-1}, \P_{\pi,v}^{n_1-1})+ 
\mathsf{KL}(\mu_{H_{-v}},\mu_{H_{v}}).
\end{align*}
By induction, it follows that~$$\mathsf{KL}(\P_{\pi,-v}^{n_1}, \P_{\pi,v}^{n_1})
\leq n_1 \mathsf{KL}(\mu_{H_{-v}},\mu_{H_{v}})  \leq  c^+v^2n_1$$ for~$c^+= c^+(\eps) := \frac{4}{(0.5^2-\eps^2)}$, the second estimate following from \eqref{eq:Hprop}. This upper bound on~$\mathsf{KL}(\P_{\pi,-v}^{n_1}, \P_{\pi,v}^{n_1})$ and Equation~\eqref{eqn:ETCKLlowerbound} imply that for every~$v \in (0, \varepsilon]$ we have
\begin{equation}\label{eq:LBaux}
\begin{aligned}
\sup_{\substack{F^i \in \{J_{\tau} : \tau \in [0,1]\} \\ i = 1, 2}}\E [R_n(\pi)] &\geq \frac{c_-v\eta}{8} \left[n_1 + (n-n_1) \exp\del[1]{-c^+v^2n_1}\right] \\
&\geq \frac{c_-v\eta}{8} n \exp\del[1]{-c^+v^2n_1}.
\end{aligned}
\end{equation}

To establish the first claim in the theorem, we use that the supremum in Equation~\eqref{eq:LBaux} is bounded from below by the average of the first lower bound appearing in that Equation applied to $v = \eps$ and to $v = \eps/\sqrt{n_1}$. In particular,  after dropping two nonnegative terms in this average, the supremum is found to be bounded from below by
\begin{equation}\label{eqn:dropav}
\frac{c_- \eta \eps}{16}  \left[ n_1 + \frac{n-n_1}{\sqrt{n_1}} \exp(-c^+ \eps^2) \right].
\end{equation}
We consider two cases: On the one hand, if~$n_1 \geq n/2$, the quantity in~\eqref{eqn:dropav} is not smaller than~$\frac{c_- \eta \eps}{32} n$. On the other hand, if~$n_1 < n/2$, then the quantity in~\eqref{eqn:dropav} is not smaller than
\begin{equation}
\frac{c_- \eta \eps}{16}  \left[ n_1 + \frac{n}{2\sqrt{n_1}} \exp(-c^+ \eps^2) \right] \geq 
\frac{c_- \eta \eps}{16} \inf_{z \in (0, \infty)} \left[ z^2 n + \frac{\sqrt{n}}{2z} \exp(-c^+ \eps^2) \right].
\end{equation}
The infimum is attained at~$z^* = c(\eps) n^{-1/6}$ for~$c(\eps) := 4^{-1/3} \exp(-c^+(\eps)\eps^2/3)$, implying the lower bound
\begin{equation}
\frac{c_- \eta \eps}{16} \left[ c^2(\eps) + \frac{1}{2 c(\eps)} \exp(-c^+(\eps) \eps^2) \right] n^{2/3}.
\end{equation}
Combining the two cases proves the first statement with constant $c_l = \frac{c_-  \eps}{16} \min(0.5, c^2(\eps) + \frac{1}{2 c(\eps)} \exp(-c^+(\eps) \eps^2))$.

Upon replacing~$n_1$ by~$n^*$ and setting~$v=\eps$  in the second line in Equation \eqref{eq:LBaux}, the second statement in the theorem follows with constant~$c_l(n^*) = \frac{c_-\eps}{8}\exp(-c^+(\eps) \eps^2n^*)$.

\subsubsection{Proof Theorem~\ref{thm:ETCES}}

Let~$Y_t$ be i.i.d.~such that the marginal~$Y_{i,t}$ has cdf~$F^i \in \mathscr{D}$ for~$i = 1, \hdots, K$. We denote~$\bar{\Delta} := \{i : \Delta_i > 0\}$. If~$\bar{\Delta} = \emptyset$ there is nothing to prove. Thus, we assume henceforth that $\bar{\Delta} \neq \emptyset$. Let~$n \in \N$ be fixed. In the following, we will abbreviate~$\tilde{\pi}_{n,t} = \tilde{\pi}_t$ for~$t = 1, \hdots,n$, and will write~$n_1 = n_1(n) :=  \min(K \lceil n^{2/3} \rceil, n)$. We consider two cases:

1) Suppose that~$n \leq K\lceil n^{2/3}\rceil$: Note that trivially~$R_n(\tilde{\pi}) \leq Cn\leq CK\lceil n^{2/3}\rceil \leq 2CK n^{2/3}$, where we used that Assumption~\ref{as:MAIN} implies~$\Delta_i\leq C$. Hence, Equation~\eqref{eqn:regretc23} holds.

2) Suppose that~$n> K\lceil n^{2/3} \rceil$: Note that~$n_1 = K\lceil n^{2/3} \rceil$. Equations~\eqref{eqn:Sintdef} and~\eqref{eq:regret2new} show that
\begin{equation}
R_n(\tilde{\pi})
= \sum_{i \in \bar{\Delta}} \Delta_i S_i(n) = 
\sum_{i \in \bar{\Delta}} \Delta_i \sum_{t = 1}^n 
\mathds{1} \{ \tilde{\pi}_{t}(Z_{t-1}, G_t) = i \}.
\end{equation}
Decomposing the last sum, and using~$\Delta_i \leq C$ yields
\begin{equation}
R_n(\tilde{\pi}) \leq C \sum_{t = 1}^{n_1} 
\mathds{1} \{ \tilde{\pi}_{t}(Z_{t-1}, G_t) \in \bar{\Delta} \} + \sum_{i \in \bar{\Delta}} \Delta_i \sum_{t = n_1 + 1}^{n} 
\mathds{1} \{ \tilde{\pi}_{t}(Z_{t-1}, G_t) = i \}.
\end{equation}
We thus obtain
\begin{equation}
\E R_n(\tilde{\pi}) \leq  C n_1 +   \sum_{t = n_1 + 1}^{n} 
\sum_{i \in \bar{\Delta}} \Delta_i \P (\tilde{\pi}_{t}(Z_{t-1}, G_t) = i).
\end{equation}
By definition, for~$t = n_1 + 1, \hdots, n$,
\begin{equation}\label{eqn:recallETC}
\tilde{\pi}_{t}(Z_{t-1}, G_t) = \min \arg \max \{\mathsf{T}(\hat{F}_{i,n_1, n}): S_{i, n}(n_1) > 0 \},
\end{equation}
which, in particular, is constant in~$t = n_1 + 1, \hdots, n$. Hence,
\begin{equation}\label{eqn:longargstart}
\E R_n(\tilde{\pi}) \leq  C n_1 +   (n-n_1)
\sum_{i \in \bar{\Delta}} \Delta_i \P (\tilde{\pi}_{n_1 + 1}(Z_{n_1}, G_{n_1 +1}) = i).
\end{equation}

We now develop an upper bound for the probabilities appearing in the previous display. Note that~$\bar{\Delta} \neq \{1, \hdots, K\}$, fix~$i^* \in \bar{\Delta}^c := \{1, \hdots, K\} \backslash \bar{\Delta}$, and let~$i \in \bar{\Delta}$. From Equation~\eqref{eqn:recallETC} it follows that~$\P (\tilde{\pi}_{n_1 + 1}(Z_{n_1}, G_{n_1 +1}) = i)$ is bounded from above by
\begin{equation}
\P(\mathsf{T}(\hat{F}_{i,n_1, n}) \geq \mathsf{T}(\hat{F}_{i^*,n_1, n}), S_{i^*, n}(n_1) > 0, S_{i, n}(n_1) > 0) + \P(S_{i^*, n}(n_1) = 0).
\end{equation}
For~$j \in \{i, i^*\}$, we denote by~$\iota_j$ the~$n_1$-dimensional random vector with $t$-th coordinate equal to~$1$ if $G_t = j$, and equal to~$0$ otherwise. Furthermore, for every~$c \in \{0, 1\}^{n_1}$ such that~$c \neq 0$, we define the empirical cdf~$G(c, j) := \|c\|_1^{-1} \sum_{t: c_t \neq 0} \mathds{1} \{Y_{j,t} \leq \cdot \}$,~$\|\cdot\|_1$ denoting the 1-norm. We now define the event $\{\mathsf{T}(\hat{F}_{i,n_1, n}) \geq \mathsf{T}(\hat{F}_{i^*,n_1, n})\} = : M(i,i^*)$, and write the first probability in the previous display as
\begin{equation}\label{eqnETCdoubles}
\sum_{a \in \{0, 1\}^{n_1} \backslash \{0\}} \sum_{b \in \{0, 1\}^{n_1} \backslash \{0\}} \P(M(i, i^*), \iota_i = a, \iota_{i^*} = b)
\end{equation}
Recall that for every~$t = 1, \hdots, n_1$ we have
\begin{equation}\label{eqn:recallETC1}
\tilde{\pi}_{t}(Z_{t-1}, G_t) = G_t, \quad \text{ with } G_t \text{ uniformly distributed on } \mathcal{I} = \{1, \hdots, K\}.
\end{equation}
On the event where~$\iota_i = a$ and~$\iota_{i^*} = b$, we can use Equations~\eqref{eq:Fhat} and~\eqref{eqn:recallETC1} to write
\begin{equation}
M(i, i^*) = \big\{ \mathsf{T}(G(a, i)) \geq
\mathsf{T}(G(b, i^*))  \big\}.
\end{equation}
Because, for~$j \in \{i,i^*\}$, the random vector~$\iota_j$ is a measurable function of~$G^{n_1} = (G_1, \hdots, G_{n_1})$, and since~$Y_1, \hdots, Y_{n_1}$ is independent of~$G^{n_1}$, it follows that~$\{ \mathsf{T}(G(a, i)) \geq
\mathsf{T}(G(b, i^*))\}$ and $\{\iota_i = a, \iota_{i^*} = b\}$ are independent, and we can write the double sum in Equation~\eqref{eqnETCdoubles} as
\begin{equation}\label{eqnETCdoubles2}
\sum_{a \in \{0, 1\}^{n_1} \backslash \{0\}} \sum_{b \in \{0, 1\}^{n_1} \backslash \{0\}} \P\big(\mathsf{T}(G(a, i)) \geq 
\mathsf{T}(G(b, i^*)) \big) \P( \iota_i = a, \iota_{i^*} = b).
\end{equation}
Since~$\Delta_i = \mathsf{T}(F^{i^*}) - \mathsf{T}(F^i)$, we can bound every~$\P\big(\mathsf{T}(G(a, i)) \geq 
\mathsf{T}(G(b, i^*)) \big)$ from above by
\begin{align}
&\P\big(|\mathsf{T}(G(a, i)) - \mathsf{T}(F^i)| + |\mathsf{T}(F^{i^*}) - \mathsf{T}(G(b, i^*))| \geq \Delta_i
\big) \\
\leq   ~~ &
\P\big(|\mathsf{T}(G(a, i)) - \mathsf{T}(F^i)| \geq \Delta_i/2\big) + \P\big(|\mathsf{T}(G(b, i^*)) - \mathsf{T}(F^{i^*})| \geq \Delta_i/2\big).
\end{align}
Using Assumption~\ref{as:MAIN}, we can bound the latter sum by
\begin{equation}
\P\big( \|G(a, i)- F^i\|_{\infty} \geq \Delta_i/(2C)\big) + \P\big(\|G(b, i^*) - F^{i^*}\|_{\infty} \geq \Delta_i/(2C)\big)
\end{equation}
Hence, the double sum in Equation~\eqref{eqnETCdoubles2} is seen to be bounded from above by
\begin{equation}
\begin{aligned}\label{eqn:sumDKWup}
&\sum_{a \in \{0, 1\}^{n_1} \backslash \{0\}} 
\P\big( \|G(a, i)- F^i\|_{\infty} > \Delta_i/(2C)\big) \P(\iota_i = a) \\
+ &\sum_{b \in \{0, 1\}^{n_1} \backslash \{0\}}
\P\big(\|G(b, i^*) - F^{i^*}\|_{\infty} > \Delta_i/(2C)\big) \P(\iota_{i^*} = b).
\end{aligned}
\end{equation}
The Dvoretzky-Kiefer-Wolfowitz-Massart inequality (note that Equation~1.5 in \cite{massart1990} obviously remains valid if ``$>$'' is replaced by ``$\geq$'') implies that the first sum in Equation~\eqref{eqn:sumDKWup} is bounded from above by
\begin{align}
\sum_{m = 1}^{n_1} 
2 e^{-m \Delta_i^2/(2C^2)} \P(\|\iota_i\|_1 = m) 
\leq 
\frac{\sqrt{2}C}{\Delta_i}  
\sum_{m = 1}^{n_1} 
\frac{1}{\sqrt{m}} \P(\|\iota_i\|_1 = m),
\end{align}
where, to obtain the inequality, we used that~$z e^{-z^2} \leq (2e)^{-1/2} < 1/2$ for every~$z > 0$. An analogous upper bound holds for the second sum in Equation~\eqref{eqn:sumDKWup}. Noting that~$\|\iota_j\|_1 = S_{j,n}(n_1)$, and since the distribution of~$S_{j,n}(n_1)$ does not depend on~$j \in \{i, i^*\}$ (cf.~Equation~\eqref{eqn:recallETC1}), we therefore see that the double sum in Equation~\eqref{eqn:sumDKWup} is bounded from above by~$\frac{4C}{\Delta_i}  
\sum_{m = 1}^{n_1} 
\frac{1}{\sqrt{2m}} \P(\|\iota_{i^*}\|_1 = m)$. Summarizing the argument we started after~Equation~\eqref{eqn:longargstart}, we now obtain
\begin{equation}
\E R_n(\tilde{\pi}) \leq Cn_1 + (n-n_1) \sum_{i \in \bar{\Delta}} \left[4C \left(\sum_{m = 1}^{n_1} \frac{1}{\sqrt{2m}} \P(S_{i^*,n}(n_1) = m)\right) + \Delta_i  \P(S_{i^*,n}(n_1) = 0) \right],
\end{equation}
which (using that~$\sqrt{2m} \geq \sqrt{m+1}$ for~$m \geq 1$,~$\Delta_i \leq C$, and Jensen's inequality) gives
\begin{equation}
\E R_n(\tilde{\pi}) \leq  C n_1 + 4 C(n-n_1)
K \left[\E[1/(S_{i^*, n}(n_1) + 1)]\right]^{1/2}.
\end{equation}
From Equations~\eqref{eqn:Sintdef} and~\eqref{eqn:recallETC1} it follows that~$S_{i^*, n}(n_1)$ is Bernoulli distributed with success probability $K^{-1}$ and ``sample size'' $n_1$. Equation~3.4 in~\cite{chao1972negative} establishes
\begin{equation}
\E\left(1/[S_{i^*, n}(n_1) + 1]\right) = K\frac{1-(1-K^{-1})^{n_1+1}}{n_1+1} \leq \frac{K}{n_1}.
\end{equation}
Therefore, 
\begin{equation}
\E R_n(\tilde{\pi}) \leq  C n_1 + 4C(n-n_1)
K^{3/2} n_1^{-1/2} \leq 2 C K n^{2/3}  + 4 CK(n^{2/3}-n^{1/3}) \leq 6CKn^{2/3},
\end{equation}
where the second inequality was obtained from~$2 Kn^{2/3} \geq n_1 = K\lceil n^{2/3} \rceil \geq K n^{2/3}$.

\subsection{Proofs of results in Section~\ref{sec:F-UCBnoCov}}

\subsubsection{An optional skipping result for Functional UCB-type policies}\label{app:OSkipD}

In this section we discuss an optional skipping result for policies as in Policy~\ref{UCBt}.

\medskip
\onehalfspacing
\begin{policy}[H]
\caption{Functional UCB-type policy~$\pi$ without randomization}\label{UCBt}
\textbf{Input}: Function~$\rho: \N \times \N \to [0, \infty)$ \\

\For{$t = 1, \hdots, K$}{assign~$\pi_{t}(Z_{t-1}) = t$}
\For{$t \geq K+1$}{assign~$\pi_{t}(Z_{t-1}) = \min \argmax_{i\in \mathcal{I}}\cbr[2]{\mathsf{T}(\hat{F}_{i,t-1})+C \rho(S_{i}(t-1), t)}$}
\end{policy}

\medskip
\doublespacing

\begin{lemma}[Optional skipping for Functional UCB-type policies without randomization]\label{lem:welldefrvMUCBt}
Let Assumptions~\ref{as:MAIN} and~\ref{as:MB} hold. Suppose~$\pi$ is a functional UCB-type policy without randomization as in Policy~\ref{UCBt}. Suppose the function~$\rho: \N \times \N \to [0,\infty)$ is non-increasing in its first argument, and satisfies, for every~$\kappa \in \N$, that
\begin{equation}\label{eqn:asOS}
\sup_{m' \geq \kappa} \limsup_{t \to \infty} \left(\rho(\kappa, t) - \rho(m', t)\right) > 2.
\end{equation}
Then, for every~$i \in \{1, \hdots, K\}$, every~$r \in \N$, and every~$\omega \in \Omega$, we have
\begin{equation}\label{eqn:deftlromega}
t_{i,r}(\omega) := \inf \big\{ s \in \N: \sum_{j = 1}^{s} \mathds{1}_{\{\pi_j(Z_{j-1}) = i\}}(\omega) = r \big\} \in \N,
\end{equation}
and for every~$i \in \{1, \hdots, K\}$ and every~$m \in \N$, the joint distributions of~$Y_{i, 1}, \hdots, Y_{i, m}$ and of~$ Y_{i, t_{i,1}}, \hdots, Y_{i, t_{i,m}}$ coincide.
\end{lemma}

\begin{remark}\label{rem:exOS}
The F-UCB Policy~\ref{poly:FUCB} (for any~$\beta > 0$) and the F-aMOSS Policy~\ref{poly:faMOSS} (for any~$\beta > 1/4$) satisfy the requirements of Lemma~\ref{lem:welldefrvMUCBt} (e.g., note that these two policies satisfy the monotonicity requirement, and also satisfy~$\rho(\kappa, t) - \rho(4\kappa, t) \geq \frac{1}{2} \rho(\kappa, t)$, which together with~$\rho(\kappa, t)\to \infty$ as~$t \to \infty$ verifies~\eqref{eqn:asOS}.).
\end{remark}

\begin{proof}
For the first statement, we argue by contradiction. Suppose there would exist an~$\ell \in \mathcal{I}$ and an~$\omega \in \Omega$ such that~$1 \leq \sum_{j = 1}^{\infty} \mathds{1}_{\{\pi_j(Z_{j-1}) = \ell\}}(\omega) =: \kappa(\omega) < \infty$. For notational convenience, we do not show in our notation that all random variables in this proof are evaluated at~$\omega$. For an arbitrary~$F \in \mathscr{D}$, Assumption~\ref{as:MAIN} implies that for every~$t \geq K+1$
\begin{equation*}
|\mathsf{T}(\hat{F}_{\pi_t(Z_{t-1}),t-1}) - 	\mathsf{T}(\hat{F}_{\ell,t-1}) | \leq |\mathsf{T}(\hat{F}_{\pi_t(Z_{t-1}),t-1}(\omega)) - \mathsf{T}(F)| + |\mathsf{T}(F) - 	\mathsf{T}(\hat{F}_{\ell,t-1})|\\ \leq 2C.
\end{equation*}
From the definition of~$\pi$, for all~$t \geq K+1$ large enough such that~$S_{\ell}(t-1) = \kappa$, it thus follows that
\begin{equation}\label{eqn:OSdis2}
2 \geq \rho(\kappa, t) - \rho(S_{\pi_{t}(Z_{t-1})}(t-1), t).
\end{equation}
By~\eqref{eqn:asOS}, there exists an~$m' \geq \kappa$ such that~$\limsup_{t \to \infty} \left(\rho(\kappa, t) - \rho(m', t)\right) > 2$. By a pigeonhole argument, it follows that~$S_{\pi_{t}(Z_{t-1})}(t-1) \to \infty$ as~$t \to \infty$. Thus,~$S_{\pi_{t}(Z_{t-1})}(t-1) \geq m'$ eventually. Since $\rho$ is non-increasing in its first argument, Equation~\eqref{eqn:OSdis2} implies~$2 \geq \rho(\kappa, t) - \rho(m', t)$ for all~$t$ large enough, a contradiction.

For the second claim we apply Doob's optional skipping theorem, cf.~Proposition~4.1 in \cite{Kallsymm}. To verify the conditions there, denote by~$\mathcal{F}$ the natural filtration corresponding to the i.i.d.~sequence~$(Y_{t})_{t \in \N}$. For every~$r \in \N$, the (measurable) function~$t_{i,r}$ takes its values in~$\N$; furthermore, it is easy to see that~$\{t_{i, r} = s \} \in \mathcal{F}_{s-1}$ for every~$s \in \N$. Therefore,~$t_{i, r}$ is an~$\mathcal{F}$-predictable time. In addition,~$t_{i,1} < \hdots < t_{i,m}$ holds by definition. The statement in the lemma now follows from Proposition~4.1 in \cite{Kallsymm} (the remaining assumptions there following immediately as~$(Y_{t})_{t \in \N}$ is i.i.d).
\end{proof}

\subsubsection{Proof of Theorem~\ref{RegretBound}}

Let~$n \in \N$ be fixed, and let~$F^i \in \mathscr{D}$ for~$i = 1, \hdots, K$. We need to show that~$\E[R_n(\hat{\pi})] \leq c (Kn\overline{\log}(n))^{0.5}$ for~$c = c(\beta, C)$ as defined in the statement of the theorem. Note that this inequality trivially holds if~$\mathsf{T}(F^1) = \hdots = \mathsf{T}(F^K)$. Therefore, we will assume that~$\mathsf{T}(F^i)$ is not constant in~$i \in \{1, \hdots, K\}$.
%; also implying that~$C$ from Assumption~\ref{as:MAIN} must satisfy~$C > 0$
Because~$\hat{\pi}$ is an anytime policy, we shall make use of the notational simplifications discussed right after Assumption~\ref{as:MB} (e.g., we write~$S_i(t)$ instead of~$S_{i,n}(t)$ and~$\hat{F}_{i,t}$ instead of~$\hat{F}_{i,t,n}$).

We now claim that for every~$i$ with~$\Delta_i>0$ it holds that
\begin{equation}\label{NumberBound}
\mathbb{E}[S_i(n)] \leq  \frac{2 C^2 \beta \log(n)}{\Delta_i^2}+\frac{\beta+2}{\beta-2}.
\end{equation}
Before proving this claim, recall from Equation~\eqref{eq:regret2new} that~$\E[R_n(\hat{\pi})]=\sum_{i: \Delta_i>0} \Delta_i \mathbb{E}[S_i(n)]$, which, together with the claim in Equation~\eqref{NumberBound} and~$\Delta_i \leq C$ (by Assumption~\ref{as:MAIN}), yields
\begin{align*}
\E[R_n(\hat{\pi})]&= \sum_{i:\Delta_i>0}  \sqrt{\Delta_i^2\mathbb{E}[S_i(n)]} \sqrt{\mathbb{E}[S_i(n)]} \\
& \leq \sqrt{ 2 C^2 \beta \log(n) + C^2 (\beta+2)/(\beta-2)} \sum_{i:\Delta_i>0} \sqrt{\mathbb{E}[S_i(n)]}  \leq \left[c(\beta, C) \sqrt{ \overline{\log}(n)} \right] \sqrt{K n},
\end{align*}
the last inequality following from the Cauchy-Schwarz inequality and
$\sum_{i:\Delta_i>0} \mathbb{E}[S_i(n)] \leq n.$
Therefore, to conclude the proof, it remains to prove the statement in Equation~\eqref{NumberBound}. 

%Before doing that, we also observe that  Equations~\eqref{NumberBound} and~\eqref{ExpRegret} give the regret bound
%%
%\begin{align}
%\E[R_N(\hat{\pi})]\leq\sum_{i: \Delta_i>0} \left( \frac{2 C^2 \beta \log(n)}{\Delta_i}+\frac{\beta+2}{\beta-2} \Delta_i \right).
%\label{eq:adaptive_regret}
%\end{align}

To this end, let~$i$ be such that~$\Delta_i>0$. We first note that if~$n \leq K$, then~$S_i(n) \leq 1$, hence Equation~\eqref{NumberBound} is trivially satisfied in this case. Consider now the case where~$n > K$. We set~$i^* := \min \argmax_{i = 1, \hdots, K} \mathsf{T}(F^i)$. For~$t > K$, from the definition of~$\hat{\pi}$, it follows that~$S_i(t-1) \geq 1$. We now abbreviate~$\{\hat{\pi}_t(Z_{t-1})  = i\}$ by~$\{\hat{\pi}_t = i\}$, and will argue that for~$t> K$ we have~$\{ \hat{\pi}_t=i \} \subseteq A_{t} \cup B_{i,t} \cup C_{i,t}$, where
\begin{equation}\label{eqn:ABC}
\begin{aligned}
A_{t} &:= \Big \{ \mathsf{T}(\hat{F}_{i^{\ast}, t-1})+C\sqrt{ \beta \log(t)/(2 S_{i^{\ast}}(t-1))} \leq \mathsf{T}(F^{i^{\ast}}) \Big \}, \\
B_{i,t} &:= \Big \{ \mathsf{T}(\hat{F}_{i, t-1}) > \mathsf{T}(F^i) +C\sqrt{ \beta \log(t)/(2 S_i(t-1))} \Big \}, \\
C_{i,t} &:= \Big \{\Delta_i < 2C \sqrt{\beta \log(n)/(2S_i(t-1))}  \Big \} = \left\{S_i(t-1) < 2\beta C^2 \log(n)/\Delta_i^2\right\}.
\end{aligned}
\end{equation}
Indeed, on the complement of~$A_t \cup B_{i,t} \cup C_{i,t}$ we have
\begin{align*}
\mathsf{T}(\hat{F}_{i^{\ast}, t-1})+C\sqrt{ \beta \log(t)/(2 S_{i^{\ast}}(t-1))} > \mathsf{T}(F^{i^{\ast}}) &= \mathsf{T}(F^i) + \Delta_i \\
&\geq  \mathsf{T}(F^i)+ 2 C\sqrt{ \beta \log(n)/(2 S_i(t-1))} \\
&\geq  \mathsf{T}(F^i)+ 2 C\sqrt{ \beta \log(t)/(2 S_i(t-1))} \\
& \geq \mathsf{T}(\hat{F}_{i, t-1})+C\sqrt{ \beta \log(t)/(2 S_i(t-1))},
\end{align*}
which implies~$\hat{\pi}_t(Z_{t-1}) \neq i$.
%
%\begin{align*}
%T(F_{i^{\ast}, S_{i^{\ast}}(t-1)})+C\sqrt{ \beta \log(t)/2 S_{i^{\ast}}(t-1)} &> T(F_{i^{\ast}}) \\
%&\geq  T(F_i)+ 2 C\sqrt{ \beta \log(t)/2 S_i(t-1)} \\
%& \geq T(F_{i, S_{i}(t-1)})+C\sqrt{ \beta \log(t)/2 S_i(t-1)}.
%\end{align*}
%Note that the inequality~$\{ \pi_t=i \} \subseteq A_{i,t} \cup B_{i,t} \cup C_{i,t}$ is also valid for~$t \leq K$ via the notation
%$A_{i,t} \equiv B_{i,t} \equiv \emptyset$ and~$C_{i,t}\equiv\Omega$ for all~$t \leq K$ and~$i \in \mathcal{I}.$
Hence,~$\{\hat\pi_t=i \} \subseteq A_{t} \cup B_{i,t} \cup C_{i,t}$ for~$t > K$. Setting~$u:=\left \lceil 2 C^2 \beta \log(n)/\Delta_i^2\right \rceil,$ we therefore obtain (recalling that~$n\geq K+1$, and by definition of~$\hat{\pi}$) 
\begin{equation}\label{eqn:upSABC}
\begin{aligned}
S_i(n)
&=
\sum_{t=1}^K \mathds{1}_{ \{\hat\pi_t=i \}}+\sum_{t=K+1}^n\mathds{1}_{ \{\hat\pi_t=i \}}
=
1+\sum_{t=K+1}^n \mathds{1}_{ \{\hat\pi_t=i \}} \\
&=
1+\sum_{t=K+1}^n \mathds{1}_{ \{\hat\pi_t=i \}\cap C_{i, t}}+\sum_{t=K+1}^n \mathds{1}_{ \{\hat\pi_t=i \}\cap C_{i, t}^c}
\leq
u+\sum_{t=K+1}^n \mathds{1}_{A_{t} \cup B_{i,t}},
\end{aligned}
\end{equation}
where, to obtain the inequality, we used~$1+\sum_{t=K+1}^n \mathds{1}_{ \{\hat\pi_t=i \}\cap C_{i, t}}\leq u$. To see the latter inequality, we consider two cases: On the one hand, if~$\omega \in \Omega$ is such that~$\omega \notin C_{i,t}$ for every~$t = K+1, \hdots, n$, then the inequality trivially holds, because~$u \geq 1$. On the other hand, denoting by~$t^*$ the largest~$t \in \{K+1, \hdots, n\}$ such that~$\omega \in C_{i, t}$, it follows that
\begin{equation}
1+\sum_{t=K+1}^n \mathds{1}_{ \{\hat\pi_t=i \}\cap C_{i, t}}(\omega) \leq \sum_{t=1}^{t^*} \mathds{1}_{ \{\hat\pi_t=i \}}(\omega)  = S_i(t^*)(\omega) \leq S_i(t^*-1)(\omega)+1 \leq u,
\end{equation}
where, for the last inequality, we used the second expression for~$C_{i,t}$ in Equation~\eqref{eqn:ABC}. From the upper bound in Equation~\eqref{eqn:upSABC} we get
\begin{align*}
\mathbb{E}[S_i(n)]
\leq 
u+\sum_{t=K+1}^n 
\left[\p(A_{t})+ \p(B_{i,t})\right].
\end{align*}
We will show further below that for~$t = K+1, \hdots, n$ we have:
\begin{equation}\label{eqn:banditclaims}
\begin{aligned}
\p(A_t)  \leq \sum_{s=1}^t \p \big(\mathsf{T}( F_{i^*,s}) + C \sqrt{ \beta \log(t)/(2 s)} \leq \mathsf{T}(F^{i^*})) \\
\p(B_{i,t})  \leq \sum_{s=1}^t \p \big(\mathsf{T}( F_{i,s})> \mathsf{T}(F^i)+ C \sqrt{ \beta \log(t)/(2 s)}),
\end{aligned}
\end{equation}
where for every~$s \in \{1, \hdots, t\}$ and every~$l \in \{i,i^*\}$ we define~$F_{l, s} := s^{-1} \sum_{j = 1}^s \mathds{1}_{\{Y_{l, j} \leq \cdot\}}$.
From Equation~\eqref{eqn:banditclaims}, Assumption~\ref{as:MAIN} and the Dvoretzky-Kiefer-Wolfowitz-Massart inequality (note that Equation~1.5 in \cite{massart1990} obviously remains valid if ``$>$'' is replaced by ``$\geq$''), we then obtain 
\begin{align*}
\p(A_{t}) \leq \sum_{s=1}^t \p (||F_{i^*,s}-F^{i^*}||_{\infty}\geq\sqrt{ \beta \log(t)/(2 s)} ) \leq 2 \sum_{s=1}^t \frac{1}{t^{\beta}}=\frac{2}{t^{\beta-1}} \\
\p(B_{i,t}) \leq \sum_{s=1}^t \p (||F_{i,s}-F^i||_{\infty}>\sqrt{ \beta \log(t)/(2 s)} ) \leq 2 \sum_{s=1}^t \frac{1}{t^{\beta}}=\frac{2}{t^{\beta-1}}.
\end{align*}
The integral bound
\begin{equation}\label{SumIntegralBdd}
\sum_{t=K+1}^n \frac{1}{t^{\beta-1}}  \leq \int_{K}^{\infty} \frac{1}{x^{\beta-1}} dx=\frac{1}{(\beta-2) K^{\beta-2}} \leq \frac{1}{\beta-2}
\end{equation}
combined with~$u \leq 1+2  C^2 \beta \log(n)/\Delta_i^2$ now establishes \eqref{NumberBound}.

It remains to verify the two inequalities in Equation~\eqref{eqn:banditclaims}. To this end, let~$t \in \{K+1, \hdots, n\}$, and note that 
\begin{align*}
\p(B_{i,t}) &= \p\left(\mathsf{T}(\hat{F}_{i, t-1}) > \mathsf{T}(F^i) +C\sqrt{ \beta \log(t)/(2 S_i(t-1))}\right) \\
&= \sum_{s = 1}^t 
\p\left(\mathsf{T}(\hat{F}_{i, t-1}) > \mathsf{T}(F^i) +C\sqrt{ \beta \log(t)/(2 s)}, S_i(t-1) = s\right).
\end{align*}
On the event~$\{S_i(t-1) = s\}$, we have~$\hat{F}_{i, t-1} = s^{-1} \sum_{j = 1}^s \mathds{1}\{Y_{i, t_{i,j}}\leq \cdot\}$ (cf.~Equation~\eqref{eqn:deftlromega}). Hence, the sum in the second line of the previous display is not greater than
\begin{align*}
\sum_{s = 1}^t 
\p\bigg[\mathsf{T}(s^{-1} \sum_{j = 1}^s \mathds{1}\{Y_{i, t_{i,j}} \leq \cdot\}) > \mathsf{T}(F^i) +C\sqrt{ \beta \log(t)/(2 s)}\bigg].
\end{align*}
Lemma~\ref{lem:welldefrvMUCBt} and Remark~\ref{rem:exOS} show that the joint distribution of~$Y_{i,t_{i,1}}, \hdots, Y_{i, t_{i,s}}$ coincides with the joint distribution of~$Y_{i,1}, \hdots, Y_{i, s}$. It thus follows that we can replace~$Y_{i,t_{i,1}}, \hdots, Y_{i, t_{i,s}}$ by~$Y_{i,1}, \hdots, Y_{i, s}$ in the previous display. In other words, we can replace~$s^{-1} \sum_{j = 1}^s \mathds{1}\{Y_{i, t_{i,j}}\leq \cdot\}$ by~$F_{i,s}$ (defined after Equation~\eqref{eqn:banditclaims}), from which the upper bound on~$\mathbb{P}(B_{i,t})$ in Equation~\eqref{eqn:banditclaims} follows. The upper bound on~$\mathbb{P}(A_{t})$ is obtained analogously.

\subsubsection{A high-probability bound for the F-UCB policy}\label{sec:HPB}

In this subsection, we use the following conventions: A sum over an empty index set is to be interpreted as~$0$, and a union over the empty index set is to be interpreted as the empty set. Recall furthermore from Equation~\eqref{NumberBound} in the proof of Theorem~\ref{RegretBound} that for every~$i$ with~$\Delta_i>0$ it holds that
\begin{equation}
\mathbb{E}[S_i(n)] \leq  \frac{2 C^2 \beta \log(n)}{\Delta_i^2}+\frac{\beta+2}{\beta-2},
\end{equation}
implying the pointwise expected regret upper bound
\begin{equation}
\E[R_n(\hat{\pi})] \leq \sum_{i: \Delta_i > 0}  \frac{2 C^2 \beta \log(n)}{\Delta_i}+ \Delta_i \frac{\beta+2}{\beta-2},
\end{equation}
which helps interpreting the multiplicative factor that appears in the subsequent high-probability bound concerning the regret of the F-UCB policy. For the mean functional, a corresponding bound is given in Theorem~8 in~\cite{audibert2009exploration}.
\begin{theorem}\label{thm:HPB}
Under Assumptions \ref{as:MAIN} and \ref{as:MB}, the F-UCB policy~$\hat{\pi}$ satisfies
\begin{align}\label{eq:probres}
\P\del[3]{R_n(\hat{\pi})> \sum_{i:\Delta_i>0}\del[2]{\frac{2C^2\beta\log(n)}{\Delta_i}+\Delta_i}x}
\leq
\frac{2K}{n^{\beta x-1}}+2\sum_{i:\Delta_i>0}\frac{\del[1]{\frac{2C^2\beta\log(n)}{\Delta_i^2}x}^{1-\beta}}{\beta-1}
\end{align}	
for every~$n\in\N$ and~$x\geq 1$.
\end{theorem}

\begin{proof}
If~$\Delta_i = 0$ for every~$i$, the claimed inequality trivially holds (due to the conventions introduced before the theorem statement). Hence, we shall now assume that~$\Delta_i > 0$ for at least one index~$i$. Observe that if~$n\leq K$, then the left hand side of~\eqref{eq:probres} is zero. Thus, in what follows~$n>K$. For~$i\in\mathcal{I},\ t>K,\ 1\leq s\leq t$ define
%V
\begin{align*}
V_{i,s,t}:=\mathsf{T}(\hat{F}_{i,t-1})+C\sqrt{\beta\log(t)/(2s)},
\end{align*}
and observe that for any~$\tau\in\N$ (recall the convention made before the theorem statement)
\begin{align*}
\cbr[0]{S_i(n)>\tau}
&=
\bigcup_{t=\tau+K}^n\cbr[0]{S_i(t-1)=\tau}\cap\cbr[0]{\hat{\pi}_t(Z_{t-1})=i}\\
&\subseteq
\bigcup_{t=\tau+K}^n\sbr[2]{\cbr[0]{S_i(t-1)=\tau}\cap\del[1]{\cbr[0]{V_{i,\tau,t}> \mathsf{T}(F^{i^*})}\cup \cbr[0]{V_{i^*,S_{i^*}(t-1),t}\leq \mathsf{T}(F^{i^*})}}},
\end{align*}
where~$i^* := \min \argmax_{i = 1, \hdots, K} \mathsf{T}(F^i)$. Set~$\tau_i=\lfloor (\frac{2C^2\beta\log(n)}{\Delta_i^2}+1)x \rfloor$. Then, using that~$R_n(\hat{\pi})=\sum_{i:\Delta_i>0}\Delta_iS_i(n)$,
\begin{align*}
\P\del[3]{R_n(\hat{\pi})> \sum_{i:\Delta_i>0}\del[2]{\frac{2C^2\beta\log(n)}{\Delta_i}+\Delta_i}x}
\leq 
\sum_{i:\Delta_i>0}\P\del[3]{S_i(n)> \del[2]{\frac{2C^2\beta\log(n)}{\Delta_i^2}+1}x},
\end{align*}
which, by the penultimate display, can be further bounded by the sum of
\begin{align}\label{eq:partA}
\sum_{i:\Delta_i>0}\sum_{t=\tau_i+K}^n\P\del[2]{V_{i,\tau_i,t}> \mathsf{T}(F^{i^*}),\ S_i(t-1)=\tau_i},
%&+\sum_{i:\Delta_i>0}\sum_{t=\tau_i+1}^n\P\del[1]{V_{i^*,S_{i^*}(t-1),t}\leq \mathsf{T}(F^{i^*}),\ S_i(t-1)=\tau_i}.
\end{align}
and
\begin{align}\label{eq:partB}
\sum_{i:\Delta_i>0}\P\del[4]{\bigcup_{t=\tau_i+K}^n\bigcup_{s=1}^{t-\tau_i-K+1}\cbr[0]{V_{i^*,S_{i^*}(t-1),t}\leq \mathsf{T}(F^{i^*}),\ S_{i^*}(t-1)=s}}.
\end{align}
We proceed by bounding~\eqref{eq:partA} and~\eqref{eq:partB}. To this end, note that each summand in~\eqref{eq:partA} equals 
\begin{align*}
\P\del[3]{\mathsf{T}(\hat{F}_{i,t-1})-\mathsf{T}(F^i)> \Delta_i-C\sqrt{\frac{\beta\log(t)}{2\tau_i}},\ S_i(t-1)=\tau_i}.
\end{align*}
Observing that~$\tau_i\geq \frac{2C^2\beta\log(n)}{\Delta_i^2}x\geq \frac{2C^2\beta\log(n)}{\Delta_i^2}$ such that~$C\sqrt{\frac{\beta\log(t)}{2\tau_i}}\leq \Delta_i/2$ for~$\tau_i+K\leq t\leq n$, an optional skipping argument (cf.~Lemma~\ref{lem:welldefrvMUCBt} and Remark~\ref{rem:exOS}), Assumption~\ref{as:MAIN}, and the DKWM-inequality yield
\begin{align*}
\sum_{t=\tau_i+K}^n\P\del[1]{V_{i,\tau_i,t}> \mathsf{T}(F^{i^*}),\ S_i(t-1)=\tau_i}
\leq
2\sum_{t=\tau_i+K}^n\exp\del[1]{-\frac{\Delta_i^2}{2C^2}\tau_i}
\leq
\frac{2}{n^{\beta x-1}}.
\end{align*}
Second, to bound~\eqref{eq:partB}, note that by interchanging unions
\begin{align*}
\bigcup_{t=\tau_i+K}^n\bigcup_{s=1}^{t-\tau_i-K+1}\cbr[0]{V_{i^*,S_{i^*}(t-1),t}\leq \mathsf{T}(F^{i^*}),\ S_{i^*}(t-1)=s},
\end{align*}
can be written as
\begin{align*}
\bigcup_{s=1}^{n-\tau_i-K+1}\bigcup_{t=s+\tau_i+K-1}^n\cbr[0]{V_{i^*,S_{i^*}(t-1),t}\leq \mathsf{T}(F^{i^*}),\ S_{i^*}(t-1)=s},
\end{align*}
which is contained in
\begin{align}\label{eq:auxunion}
\bigcup_{s=1}^{n}\bigcup_{t=s+\tau_i}^n\cbr[0]{V_{i^*,S_{i^*}(t-1),t}\leq \mathsf{T}(F^{i^*}),\ S_{i^*}(t-1)=s}.
\end{align}
Defining~$\bar{F}_{i^*,s}(\cdot):=\frac{1}{s}\sum_{r=1}^s\mathds{1}\cbr[0]{Y_{i^*,t_{i^*,r}}\leq \cdot}$, observe that for~$1\leq s\leq n$ and~$s+\tau_i\leq t\leq n$
\begin{align*}
&\cbr[1]{V_{i^*,S_{i^*}(t-1),t}\leq \mathsf{T}(F^{i^*}),\ S_{i^*}(t-1)=s}\\
\subseteq&
\cbr[3]{|\mathsf{T}(\hat{F}_{i^*,t-1})-\mathsf{T}(F^{i^*})|\geq C\sqrt{\frac{\beta\log(t)}{2S_{i^*}(t-1)}},\ S_{i^*}(t-1)=s}\\
\subseteq&
\cbr[3]{|\mathsf{T}(\bar{F}_{i^*,s})-\mathsf{T}(F^{i^*})|\geq C\sqrt{\frac{\beta\log(s+\tau_i)}{2s}}}.
\end{align*}
Hence,~\eqref{eq:auxunion} is contained in
\begin{align*}
\bigcup_{s=1}^{n}\cbr[3]{|\mathsf{T}(\bar{F}_{i^*,s})-\mathsf{T}(F^{i^*})|\geq C\sqrt{\frac{\beta\log(s+\tau_i)}{2s}}},
\end{align*}
and we conclude that the probability of the event in~\eqref{eq:auxunion} is no larger than
\begin{align*}
\P\del[3]{\bigcup_{s=1}^{n}\cbr[3]{|\mathsf{T}(\bar{F}_{i^*,s})-\mathsf{T}(F^{i^*})|\geq C\sqrt{\frac{\beta\log(s+\tau_i)}{2s}}}}
\leq
2\sum_{s=1}^{n}(\tau_i+s)^{-\beta}
\leq 
2\frac{\tau_i^{1-\beta}}{\beta-1},
\end{align*}
where we used a union bound, an optional skipping argument (cf.~Lemma~\ref{lem:welldefrvMUCBt} and Remark~\ref{rem:exOS}), Assumption~\ref{as:MAIN}, the DKWM-inequality, and an integral bound. Inserting for~$\tau_i$ and collecting terms yields~\eqref{eq:probres}.
\end{proof}

\subsubsection{Proof of Theorem~\ref{thm:LB_NoCov}}
\begin{proof}[Proof of Theorem~\ref{thm:LB_NoCov}]
Let~$\pi$ be a policy and let~$n \in \N$. Fix the randomization measure~$\P_G$. Since~$n$ is fixed, we shall abbreviate~$\pi_{n,t}=\pi_{t}$ in the sequel. As in the proof of~Theorem~\ref{thm:LBETC}, we obtain from~Lemma~\ref{lem:LBdist} a one-parametric family~$\mathcal{H} \subseteq \{J_{\tau}: \tau \in [0,1]\}$, a~$c_- > 0$ and an~$\eps \in (0, 1/2)$, such that the statement in Equation~\eqref{eq:Hprop} holds for every~$v\in [0, \eps]$, from which it follows that 
\begin{equation}\label{eqn:aslb}
\mathsf{KL}^{1/2}(\mu_{H_{-v}}, \mu_{H_v}) \leq \frac{2}{c_-\sqrt{0.5^2-\eps^2}} \min_{j \in \{-v, v\}}|\mathsf{T}(H_j) - \mathsf{T}(H_{0})| \quad \text{ for every } v \in [0, \varepsilon].
\end{equation}
For ease of notation, we set~$\zeta := \frac{2}{c_-\sqrt{0.5^2-\varepsilon^2}}$ and define~$f(v) := \zeta \min_{j \in \{-v, v\}}|\mathsf{T}(H_j) - \mathsf{T}(H_{0})|$ for every~$v \in [0, \varepsilon]$. Note that~$f(0) = 0$, and that~$f(\varepsilon) \geq 2\varepsilon/\sqrt{0.5^2-\varepsilon^2} > 0$ by \eqref{eq:Hprop}.
By continuity of~$f$ (following from the continuity of~$v\mapsto\mathsf{T}(H_v)$ as guaranteed by Lemma~\ref{lem:LBdist}) and the intermediate-value theorem, we can choose~$a_n \in (0, \varepsilon]$ such that~$f^2(a_n) =f^2(\varepsilon)/n$. 

As in the proof of Theorem~\ref{thm:LBETC}, for~$j \in \{-a_n,a_n\}$ and every~$t = 1, \hdots, n$, we denote by~$\P_{\pi,j}^t$ the distribution induced by~$Z_t$ for~$Y_t$ i.i.d.~$\mu_{H_0} \otimes \mu_{H{j}}$ and~$G_t$ i.i.d.~$\P_G$; the expectation corresponding to~$\P_{\pi,j}^t$ being denoted by~$\E_{\pi,j}^t$. We denote by~$R_n^j({\pi})$ the regret of policy~$\pi$ (under~$Y_t$ i.i.d.~$\mu_{H_0} \otimes \mu_{H_{j}}$ and~$G_t$ i.i.d.~$\P_G$). 

Arguing similarly as around the second display in the proof of Theorem~\ref{thm:LBETC}, we obtain
\begin{equation}\label{eqn:basicfacts}
\begin{aligned}
\E_{\pi,-a_n}^n R^{-a_n}_n({\pi}) 
\geq \frac{f(a_n)}{\zeta} \E_{\pi,-a_n}^nS_2(n) \quad \text{ and } \quad
\E_{\pi,a_n}^n R^{a_n}_n({\pi})
\geq \frac{f(a_n)}{\zeta}(n - \E_{\pi, a_n}^nS_2(n)),
\end{aligned}
\end{equation}
and
\begin{align*}
\sup_{j \in \{-a_n, a_n\}} \E_{\pi,j}^n R^{j}_n({\pi})
\geq 
\frac{f(a_n)}{2\zeta} \left(\E_{\pi,-a_n}^n S_2(n) +\left[n-\E_{\pi,a_n}^n S_2(n) \right]\right) 
\geq 
\frac{f(a_n)n}{8\zeta} e^{-\mathsf{KL}(\P_{\pi,-a_n}^n, \P_{\pi,a_n}^n)},
\end{align*}
the second inequality following from Theorem 2.2(iii) in \cite{tsybakov2009introduction} (and its proof), using that~$n^{-1} S_2(n)$ is a test for~$H_0: \P_{\pi, -a_n}^n$ against~$H_1: \P_{\pi, a_n}^n$. 

By the same argument as used after Equation~\eqref{eqn:ETCKLlowerbound} (but now with~$n$ instead of~``$n_1$,'' and with~$a_n$ instead of~``$v$''), we obtain $\mathsf{KL}(\P_{\pi,-a_n}^n, \P_{\pi,a_n}^n)
= \mathsf{KL}(\mu_{H_{-a_n}},\mu_{H_{a_n}}) n \leq f^2(a_n) n$, the inequality being a consequence of Equation~\eqref{eqn:aslb}. This shows that the supremum in~\eqref{eqn:suplownoco} is bounded from below by
\begin{align*}
\frac{f(a_n)n}{8\zeta}
\exp\left(-f^2(a_n)n \right) = \frac{f(\varepsilon) }{8 \zeta} \exp(-f^2(\varepsilon)) \sqrt{n}.
\end{align*}
\end{proof} 

\subsubsection{Proof of Theorem~\ref{RegretBoundfaMOSS}}\label{app:mossproof}

In the following lemma we write~$\max(x, 0) := (x)^+$ for every~$x \in \R$.
\begin{lemma}\label{lem:rm++}
Let~$X_n$,~$n\in \N$, be a sequence of i.i.d.~random variables with cdf~$F$. Denote by~$\hat{F}_n$ the empirical cdf based on~$X_1, \hdots, X_n$. Then, for every~$x \geq 0$ and every~$n \in \N$,
\begin{equation}\label{eqn:ineq4SW}
\P\left(\sup\nolimits_{m \geq n} \|\hat{F}_m - F\|_{\infty} \geq x \right) \leq  e^{-4nx^2}+ \sqrt{32\pi n} xe^{-2 n x^2 },
\end{equation}
and
\begin{equation}\label{eqn:ineq5SW}
\E \left[\sup\nolimits_{m \geq n} \left(\|\hat{F}_m - F\|_{\infty} - x \right)^+\right]
\leq \sqrt{\frac{\pi}{n}} \left(
\frac{1}{4} e^{-4 n x^2}
+ \sqrt{2}
e^{-2 n x^2 }\right).
\end{equation}
\end{lemma}

\begin{remark}\label{rem:rm++}
If in the context of Lemma~\ref{lem:rm++} Assumptions~\ref{as:MAIN} and~\ref{as:MB} hold, and if~$F\in \mathscr{D}$, then, for every~$x \geq 0$ and every~$n \in \N$,
\begin{equation}\label{eqn:ineq6SW}
\begin{aligned}
\E \left[\sup\nolimits_{m \geq n} \left( \mathsf{T}(F) - \mathsf{T}(\hat{F}_m) - x \right)^+\right]
&\leq C \sqrt{\frac{\pi}{n}} \left(
\frac{1}{4} e^{-\frac{4 n x^2}{C^2}}
+ \sqrt{2}
e^{-\frac{2 n x^2}{C^2} }\right) \\
&\leq C \sqrt{\frac{\pi}{n}} \left( \frac{1}{4} + \sqrt{2}\right)
e^{-\frac{2 n x^2}{C^2} };
\end{aligned}
\end{equation}
the same statement holds if the roles of~$\hat{F}_m$ and~$F$ are interchanged. Furthermore, the DKWM-inequality immediately implies that for every~$x \geq 0$ and every~$n \in \N$, 
\begin{equation}\label{eqn:DKWMKPV}
\E \left[\left( \mathsf{T}(\hat{F}_n) - \mathsf{T}(F) - x \right)^+\right]
\leq C \sqrt{\frac{\pi}{n}} \frac{1}{\sqrt{2}} e^{-\frac{2 n x^2}{C^2} };
\end{equation}
again, the same statement holds if the roles of~$\hat{F}_n$ and~$F$ are interchanged
\end{remark}

\begin{proof}[Proof of Lemma~\ref{lem:rm++}:]
Denote by~$\mathscr{F}_n$ the sigma-algebra generated by the collection of random variables~$\{X_{1:n}, \hdots, X_{n:n}\} \cup \{X_i : i > n\}$, where~$X_{1:n} \leq \hdots \leq X_{n:n}$ are the order statistics of~$X_1, \hdots, X_n$.
Note that~$\mathscr{F}_n \supseteq \mathscr{F}_{n+1}$, and that~$\| \hat{F}_n - F\|_{\infty}$ is~$\mathscr{F}_n$-measurable. Arguing as in the proof of Proposition~4 on p.~138 of~\cite{shorack2009empirical}, one verifies that~$(\| \hat{F}_n - F\|_{\infty}, \mathscr{F}_n)_{n \in \N}$ is a reverse submartingale, i.e., that~$\mathbb{E}[ \|\hat{F}_n - F\|_{\infty} | \mathscr{F}_{n+1}]	\geq \|\hat{F}_{n+1} - F\|_{\infty}$ for every~$n \in \N$. Hence, for every~$a \geq 0$,~$(\exp(a\|\hat{F}_n - F\|_{\infty}); \mathscr{F}_n)_{n \in \N}$ is a reverse submartingale. Therefore, for every pair of natural numbers~$N > n$, and~$s \geq 0$, it follows that~$(\exp(sn\|\hat{F}_{N-m} - F\|_{\infty}); \mathscr{F}_{N-m})_{m = 
0, \hdots, N-n}$ is a submartingale.  Doob's submartingale inequality shows that
\begin{equation}
\P\left(
\max_{m = n}^N e^{sn\|\hat{F}_m - F\|_{\infty}} \geq \varepsilon
\right) = 
\P\left(
\max_{m = 0}^{N-n} e^{sn\|\hat{F}_{N-m} - F\|_{\infty}} \geq \varepsilon
\right)
\leq 
\varepsilon^{-1}
\mathbb{E}(
e^{sn\|\hat{F}_n - F\|_{\infty}}), \forall \varepsilon > 0.
\end{equation}
Since the upper bound does not depend on~$N$ and since~$N$ was arbitrary, it follows that
\begin{equation}
\P\left(
\sup\nolimits_{m \geq n} e^{sn\|\hat{F}_m - F\|_{\infty}} \geq \varepsilon
\right)
\leq 
\varepsilon^{-1}
\mathbb{E}(
e^{sn\|\hat{F}_n - F\|_{\infty}}), \text{ for every } \varepsilon > 0.
\end{equation}
Now, for~$x \geq 0$, we may write
\begin{equation}
\P\left(\sup\nolimits_{m \geq n} \|\hat{F}_m - F\|_{\infty} \geq x \right) =
\P\left(\sup\nolimits_{m \geq n} e^{sn\|\hat{F}_m - F\|_{\infty}} \geq e^{snx} \right) \leq e^{-sn x} \mathbb{E}(e^{sn\|\hat{F}_n - F\|_{\infty}}),
\end{equation}
and use the second inequality on p.~357 of \cite{shorack2009empirical} (cf.~also Equation~35 in~\cite{kpv3}) to obtain
\begin{equation}
e^{-snx} \mathbb{E}(e^{s n\|\hat{F}_n - F\|_{\infty}}) \leq 
e^{-snx} \left(1+ \sqrt{2\pi}s\sqrt{n} e^{\frac{s^2 n}{8}}\right).
\end{equation}
We may now set~$s = 4x$ to conclude~\eqref{eqn:ineq4SW}. Next, we write the expectation in~\eqref{eqn:ineq5SW} as
\begin{equation}
\int_{0}^{\infty} \P\bigg(
\sup_{m \geq n} \|\hat{F}_m - F\|_{\infty} \geq x + y \bigg) dy
\leq
\int_{0}^{\infty}
e^{-4n(x+y)^2} dy + \sqrt{32\pi n} \int_{0}^{\infty} (x+y)e^{-2 n (x+y)^2 }
dy,
\end{equation}
where we used~\eqref{eqn:ineq4SW}. From~$\int_{0}^{\infty} (x+y) e^{-a(x+y)^2} dy = e^{-ax^2}/(2a)$ for every~$a > 0$ we obtain
\begin{align}
\int_{0}^{\infty}
e^{-4n(x+y)^2} dy + \frac{\sqrt{2\pi }}{\sqrt{n}}
e^{-2 nx^2} 
\leq 
\sqrt{\frac{\pi}{n}} \left(
\frac{1}{4} e^{-4 n x^2}
+ \sqrt{2}
e^{-2 n x^2 }\right).
\end{align}
\end{proof}

Equipped with this auxiliary result, the proof of Theorem~\ref{RegretBoundfaMOSS} is given next. To establish Theorem~\ref{RegretBoundfaMOSS} we adapt the argument developed in the proof of Proposition~17 in~\cite{garivier2018kl}.
\begin{proof}[Proof of Theorem~\ref{RegretBoundfaMOSS}:]
Let~$F^i \in \mathscr{D}$ for~$i = 1, \hdots, K$. Because~$\check{\pi}$ is an anytime policy, we can use the same notational simplifications as in the proof of Theorem~\ref{RegretBound}. 
The proof is based on a repeated application of an \emph{optional skipping result}: Define~$t_{i,r} = \inf \{ s \in \N: \sum_{j = 1}^{s} \mathds{1}_{\{\check{\pi}_j(Z_{j-1}) = i\}} = r\}$. Lemma~\ref{lem:welldefrvMUCBt} together with Remark~\ref{rem:exOS} above shows that the infimum in the definition of~$t_{i,r}$ is a minimum, and that for every~$i \in \{1, \hdots, K\}$ and every~$m \in \N$, the joint distributions of~$Y_{i, 1}, \hdots, Y_{i, m}$ and of~$ Y_{i, t_{i,1}}, \hdots, Y_{i, t_{i,m}}$ coincide. Also note that for every integer~$s > 0$ on the event~$\{S_i(t) = s\}$, we have~$\hat{F}_{i, t} = s^{-1} \sum_{j = 1}^s \mathds{1}\{Y_{i, t_{i,j}}\leq \cdot\} =: \overline{F}_{i, s}$; and also denote~$F_{i,s} := s^{-1} \sum_{j = 1}^s \mathds{1}\{Y_{i, j}\leq \cdot\}$.

Note that the inequality in~\eqref{eqn:RegretfaMOSS} trivially holds for~$n \leq K$. Hence, we let~$n > K$ in what follows. Fix~$i^* \in \argmax_{i \in \mathcal{I}} \mathsf{T}(F^i)$. For every~$t \geq K$ such that~$t\leq n$ and every~$i \in \mathcal{I}$ set
\begin{equation}\label{eqn:Udef}
U_i(t) := \mathsf{T}(\hat{F}_{i,t})+C \sqrt{\frac{\beta\log^+[t/(KS_i(t))]}{S_i(t)}} \leq \mathsf{T}(\hat{F}_{i,t})+C \sqrt{\frac{\beta \log^+[n/(KS_i(t))]}{S_i(t)}} =: \overline{U}_i(t);
\end{equation}
by definition of~$\check{\pi}$, we have~$U_{i^*}(t) \leq \max_{i \in \mathcal{I}} U_i(t) = U_{\check{\pi}_{t+1}(Z_{t})}(t)\leq \overline{U}_{\check{\pi}_{t+1}(Z_{t})}(t)$. Therefore,~$\mathbb{E}(R_n(\check{\pi}))$ is bounded by
\begin{equation}
C (K-1) + \sum_{t = K+1}^n \mathbb{E}\left(
\mathsf{T}(F^{i^*}) - U_{i^*}(t-1) \right) + 
\sum_{t = K+1}^n \mathbb{E}\left(
\overline{U}_{\check{\pi}_{t}(Z_{t-1})}(t-1) - \mathsf{T}(F^{\check{\pi}_t(Z_{t-1})}) \right).
\end{equation}
We now separately bound the two sums in this upper bound, starting with the first. Since~$U_{i^*}(t) = \mathsf{T}(\hat{F}_{i^*, t})$ for~$S_{i^*}(t) \geq \frac{t}{K}$, we can bound~$\mathbb{E}\left(
\mathsf{T}(F^{i^*}) - U_{i^*}(t) \right)$ from above by
\begin{equation}\label{eqn:t2sum}
\mathbb{E}\left[ \left(
\mathsf{T}(F^{i^*}) - \mathsf{T}(\hat{F}_{i^*, t})\right)^+ \mathds{1}_{[t/K, t)}(S_{i^*}(t))\right] + \mathbb{E}\left[ \left(
\mathsf{T}(F^{i^*}) - U_{i^*}(t) \right)^+
\mathds{1}_{[1, t/K)}(S_{i^*}(t))\right].
\end{equation}
Writing~$\{S_{i^*}(t) \geq \frac{t}{K}\}$ as the disjoint union of~$\{S_{i^*}(t) = m\}$ for all integers~$m \in [\frac{t}{K}, t-1]$, using that~$\hat{F}_{i^*, t} = \overline{F}_{i^*, m}$ if~$S_{i^*}(t) = m$, and applying the optional skipping result above, the first expectation in the previous display can be bounded from above by
\begin{equation}\label{eqn:fsu}
\mathbb{E}\left[ \max_{m = \lceil t/K \rceil, \hdots, t-1} \left(
\mathsf{T}(F^{i^*}) - \mathsf{T}(F_{i^*,m})\right)^+ \right] \leq \left(\frac{1}{4} + \sqrt{2}\right) C \sqrt{\pi K/t},
\end{equation}
where we have used Remark~\ref{rem:rm++}. Concerning the second expectation in~\eqref{eqn:t2sum}, let~$\alpha \in (1, 4\beta)$ (which will be fixed further below), define the sequence~$x_l := \alpha^{-l} t/K $ for~$l \in \N \cup \{0\}$, and denote the set of all natural numbers contained in~$[x_{l+1},x_l)$ by~$V_l$. For every~$i \in \mathcal{I}$ and every pair of natural numbers~$s_1$ and~$s_2$, define the random variable~$U_{i, s_1, s_2} := \mathsf{T}(F_{i, s_1}) + C \sqrt{\beta \log^+[s_2/(Ks_1)]/ s_1}$. Now, write~$\mathds{1}_{[1, t/K)}(S_{i^*}(t)) = \sum_{l = 0}^{\infty} \sum_{m \in V_l} \mathds{1}\{S_{i^*}(t) = m\}$ (a sum over an empty index set is set to~$0$), use Tonelli's theorem, and apply the optional skipping result to get (a maximum over an empty index set is set to~$0$)
\begin{equation}
\mathbb{E}\left[ \left(
\mathsf{T}(F^{i^*}) - U_{i^*}(t) \right)^+
\mathds{1}_{[1, t/K)}(S_{i^*}(t))\right] \leq
\sum_{l = 0}^{\infty} 
\mathbb{E}\left[
\max_{m \in V_l} \left(
\mathsf{T}(F^{i^*}) - U_{i^*, m, t} \right)^+\right].
\end{equation}
Applying Equation~\eqref{eqn:ineq6SW} with~$n = x_{l+1}$, and~$x = C \sqrt{\beta\log[t/(Kx_l)]/x_l}$ the expectation~$\mathbb{E}[
\max_{m \in V_l}\big(
\mathsf{T}(F^{i^*}) - 
\mathsf{T}(F_{i^*, m}) -  C \sqrt{\beta\log[t/(Kx_l)]/x_l}
\big)^+]$ (and thus the~$l$-th summand in the upper bound just derived), is seen to be bounded from above by
\begin{equation}
C \sqrt{\frac{\pi}{x_{l+1}}} \left(\frac{1}{4} + \sqrt{2} \right) 
e^{-2 x_{l+1} \beta\log[t/(Kx_l)]/x_l} = C \sqrt{\frac{\pi K}{t}} \left(\frac{1}{4} + \sqrt{2} \right)  \alpha^{\frac{1}{2}} 
\alpha^{l(\frac{1}{2}  - \frac{2\beta}{\alpha})} =: a_l,
\end{equation}
where we used~$x_{l+1}/x_l = \alpha^{-1}$ and~$t/(Kx_l) = \alpha^l$. 
Recalling that~$\alpha \in (1, 4\beta)$, we obtain $\sum_{l = 0}^{\infty} a_l = C (1/4 + \sqrt{2})\left[
\alpha^{-\frac{1}{2}}-\alpha^{-\frac{2\beta}{\alpha}}\right]^{-1} \sqrt{\frac{\pi K}{t}}$. Together with Equation~\eqref{eqn:fsu} and~\eqref{eqn:t2sum}, this shows that~$\mathbb{E}\left(
\mathsf{T}(F^{i^*}) - U_{i^*}(t) \right) \leq \frac{1}{2}c(\alpha, \beta) C \sqrt{\pi K/t}$, where we abbreviated~
\begin{equation}\label{eqn:calbetdef}
c(\alpha, \beta) = (1/2 + \sqrt{8})\left(\left[
\alpha^{-\frac{1}{2}}-\alpha^{-\frac{2\beta}{\alpha}}\right]^{-1} + 1\right).
\end{equation}
The bound~$\sum_{t = K}^{n-1} t^{-1/2} \leq 2(\sqrt{n}-\sqrt{K-1})$ now shows that
\begin{equation}\label{eqn:MOSS1up}
\sum_{t = K}^{n-1}\mathbb{E}\left(
\mathsf{T}(F^{i^*}) - U_{i^*}(t) \right) \leq c(\alpha, \beta) C \sqrt{\pi} \left(\sqrt{K n} - (K-1)\right).
\end{equation}
It remains to bound (abbreviating~$\check{\pi}_{t}(Z_{t-1})$ as~$\check{\pi}_{t}$)
\begin{equation}
\sum_{t = K+1}^n \mathbb{E}\left(
\overline{U}_{\check{\pi}_{t}}(t-1) - \mathsf{T}(F^{\check{\pi}_t}) \right) \leq C \sqrt{\pi Kn} + 
\sum_{t = K+1}^n \mathbb{E}\left[\left(
\overline{U}_{\check{\pi}_{t}}(t-1) - \mathsf{T}(F^{\check{\pi}_t}) - C \sqrt{\pi K/n}\right)^+ \right],
\end{equation}
which, noting that~$\sum_{i = 1}^K \sum_{s = 1}^n \mathds{1}\{S_{i}(t-1) = s, \check{\pi}_t = i\} = 1$ and~$\sum_{t = K+1}^n \mathds{1}\{S_i(t-1) = s, \check{\pi}_t = i\} \leq 1$, and using the optional skipping argument, is upper bounded by
\begin{equation}
C\sqrt{\pi Kn} + 
\sum_{i = 1}^K \sum_{s = 1}^n \mathbb{E}\left[\left(
U_{i,s,n} - \mathsf{T}(F^{i}) - C\sqrt{\pi K/n}\right)^+ \right].	
\end{equation}
Since~$(
U_{i,s,n} - \mathsf{T}(F^{i}) - C\sqrt{\pi K/n})^+ \leq (\mathsf{T}(F_{i, s}) - \mathsf{T}(F^{i}) - C\sqrt{\pi K/n})^+ + C \sqrt{ \beta \log^+[n/(Ks)]/s}$, and~$\log^+[n/(Ks)] = 0$ for~$s \geq n/K$, the sum over~$s$ in the previous display doesn't exceed
\begin{equation}
\sum_{s = 1}^n \mathbb{E}\left[\left(\mathsf{T}(F_{i, s}) - \mathsf{T}(F^{i}) - C\sqrt{\pi K/n}\right)^+\right] + C\sqrt{\beta} \sum_{s = 1}^{\lfloor n/K \rfloor} \sqrt{   \log[n/(Ks)]/s},
\end{equation}
which, by Equation~\eqref{eqn:DKWMKPV}, is further upper bounded by (cf.~also~\cite{garivier2018kl})
\begin{align}
&\frac{1}{\sqrt{2}} C \sqrt{\pi} \sum_{s = 1}^n s^{-1/2}  e^{-2 \frac{\pi K}{n} s}  + C\sqrt{\beta} \sum_{s = 1}^{\lfloor n/K \rfloor} \sqrt{  \log[n/(Ks)]/s} \\
\leq&
\frac{1}{\sqrt{2}} C \sqrt{\pi} \int_{0}^{\infty} s^{-1/2} e^{-2\frac{\pi K}{n} s}ds
+
C\sqrt{\beta} \int_0^{n/K} \sqrt{\log(n/(Ks))/s} ds 
= C \sqrt{\pi}\left( \frac{1}{2}  + \sqrt{2 \beta}\right)  \sqrt{\frac{n}{K}}.
\end{align}
Summarizing,~$\mathbb{E}(R_n(\check{\pi}))$ is bounded from above by
\begin{align}
&C(K-1) + C \sqrt{\pi} \left(\frac{3}{2} + c(\alpha, \beta)   + \sqrt{2\beta}\right) \sqrt{Kn} - c(\alpha, \beta) C \sqrt{\pi}(K-1) \\
\leq ~& 
C \sqrt{\pi} \left(4.83 + \frac{3.33 \times\alpha^{\frac{1}{2}}}{1-\alpha^{\frac{1}{2}-\frac{2\beta}{\alpha}}} + \sqrt{2\beta}
\right) \sqrt{Kn},
\end{align}
where we have used that~$c(\alpha, \beta) \geq 1$ (cf.~\eqref{eqn:calbetdef} for the definition of~$c(\alpha, \beta)$). Now, we set~$\alpha \in (1, 4\beta)$ equal to~$4\beta W_0(e/(4 \beta))$,\footnote{The first order conditions of~$\alpha \mapsto \alpha^{1/2 - 2\beta/\alpha}$ suggest this specific choice of~$\alpha$. Note that~$W_0$ is the principal branch of the Lambert~$W$ function.} noting that~$\beta > 1/4$ implies~$W_0(e/(4 \beta)) \in \left((4\beta)^{-1}, 1\right)$. We thus get the upper bound
\begin{equation}
C \sqrt{\pi} \left(4.83 + \frac{3.33 \times (4 \beta W_0(e/4 \beta))^{\frac{1}{2}}}{
1-(4\beta W_0(e/4 \beta))^{\frac{1}{2}-\frac{1}{2 W_0(e/4 \beta)}}} + \sqrt{2\beta} \right) \sqrt{Kn},
\end{equation}
which proves the result.
\end{proof}

\subsection{Proofs of the claims in Section~\ref{sec:num}} \label{sec:implproofs}

\subsubsection{Null rejection probability of the test used in the ETC-T policy}\label{sec:numprovesize}

Let~$F^1$ and~$F^2$ in~$\mathscr{D} = D_{cdf}([0, 1])$ be such that~$\mathsf{W}(F^1) = \mathsf{W}(F^2)$. Then, for every natural number~$n_1 \geq 2$, we can bound~$\P( |\mathsf{W}(\hat{F}_{1, n_1}) - \mathsf{W}(\hat{F}_{2, n_1})| \geq c_\alpha)$ from above by
\begin{equation}
\sum_{i = 1}^2 \P\left( |\mathsf{W}(\hat{F}_{i, n_1}) - \mathsf{W}(F^i)| \geq \frac{c_\alpha}{2} \right) \leq \sum_{i = 1}^2 \P\left( \|\hat{F}_{i, n_1} - F^i\|_{\infty} \geq \frac{c_\alpha}{2C} \right), 
\end{equation}
where, to obtain the inequality, we used that~$\mathsf{W}$ satisfies Assumption~\ref{as:MAIN} with~$\mathscr{D} = D_{cdf}([0, 1])$ and~$C$. Now, noting that the cyclical assignment rule leads to~$\hat{F}_{i, n_1}$ being based on at least~$\lfloor n_1/2 \rfloor$ independent observations from~$F^i$, we can use the Dvoretzky-Kiefer-Wolfowitz-Massart (DKWM) inequality to further bound the double sum to the right in the previous display by~$4 \exp(- \lfloor n_1/2 \rfloor c^2_\alpha/(2C^2)) = \alpha$, recalling that by definition~$c_\alpha=\sqrt{2\log(4/\alpha)C^2/\lfloor n_1/2\rfloor}$.

\subsubsection{Power guarantee concerning the choice of~$n_1$ in the ETC-T policy}\label{sec:numprovepow}

Let~$F^1$ and~$F^2$ in~$\mathscr{D} = D_{cdf}([0, 1])$ satisfy~$\Delta = |\mathsf{W}(F^1) - \mathsf{W}(F^2)| > 0$, let~$\eta \in (0, 1)$, and set~$n_1 = 2\lceil 8 \log(4/\min(\alpha, \eta)) C^2/\Delta^2 \rceil$. Assume first that~$\Delta = \mathsf{W}(F^1) - \mathsf{W}(F^2)$ (the other case is handled similarly). Then, the probability that the test does not reject equals
\begin{align}
\P( |\mathsf{W}(\hat{F}_{1, n_1}) - \mathsf{W}(\hat{F}_{2, n_1})| < c_\alpha) &\leq \P( \mathsf{W}(\hat{F}_{1, n_1}) - \mathsf{W}(\hat{F}_{2, n_1}) < c_\alpha) \\
&= 
\P(
\mathsf{W}(\hat{F}_{1, n_1}) - \mathsf{W}(F^1) + \mathsf{W}(F^2) - \mathsf{W}(\hat{F}_{2, n_1}) < c_{\alpha} - \Delta
) \\
&\leq
\P(
\mathsf{W}(\hat{F}_{1, n_1}) - \mathsf{W}(F^1) + \mathsf{W}(F^2) - \mathsf{W}(\hat{F}_{2, n_1}) < - \Delta/2
),
\end{align}
where we used~$c_\alpha \leq \Delta/2$ to obtain the last inequality. This can be upper bounded by
\begin{align}
&\P(
\mathsf{W}(\hat{F}_{1, n_1}) - \mathsf{W}(F^1) < - \Delta/4) +
\P( \mathsf{W}(F^2) - \mathsf{W}(\hat{F}_{2, n_1}) < - \Delta/4
) \\
\leq ~ & \sum_{i = 1}^2
\P(
|\mathsf{W}(\hat{F}_{i, n_1}) - \mathsf{W}(F^i)| > \Delta/4)  
~ \leq ~
%\sum_{i = 1}^2
%\P(
%\|\hat{F}_{i, n_1} - F^i\|_{\infty} > \Delta/(4C)) \leq
4 \exp(-\lfloor n_1/2 \rfloor \Delta^2/(8C^2)) \leq \min(\alpha, \eta) \leq \eta,
\end{align}
where (as in the previous subsection) we used Assumption~\ref{as:MAIN} and the DKWM inequality.

\subsubsection{Regret guarantee concerning the choice of~$n_1$ in the ETC-ES policy}\label{sec:numprovereg}

Let~$F^i\in\mathscr{D} = D_{cdf}([0,1])$ be arbitrary for~$i=1,\ldots,K$ and let~$\delta > 0$. Furthermore, set~$n_1=K\lceil 16(K-1)^2C^2/(\exp(1)\delta^2)\rceil$. If~$\Delta_i=0$ for~$i=1,\ldots,K$, then~$\E(\max_{i \in \mathcal{I}} \mathsf{W}(F^i) - \mathsf{W}(F^{\pi_n^c(Z_{n_1})}))=0$. Thus, suppose that there exists an~$i\in\cbr[0]{1,\ldots,K}$ such that~$\Delta_i>0$. Then, denoting by~$i^*$ an element of~$\argmax_{i\in\cbr[0]{1,\ldots,K}}\mathsf{W}(F^i)$, 
\begin{align*}
	\E\del[1]{\max_{i \in \mathcal{I}} \mathsf{W}(F^i) - \mathsf{W}(F^{\pi_n^c(Z_{n_1})})}
	&=
	\sum_{i:\Delta_i>0}\Delta_i\P(\pi_n^c(Z_{n_1})=i)\\
	&\leq
	\sum_{i:\Delta_i>0}\Delta_i\P\del[1]{\mathsf{W}(\hat{F}_{i,n_1})\geq\mathsf{W}(\hat{F}_{i^*,n_1})}\\
	&=\sum_{i:\Delta_i>0}\Delta_i\P\del[1]{\mathsf{W}(\hat{F}_{i,n_1})-\mathsf{W}(F^i)+\mathsf{W}(F^{i^*})-\mathsf{W}(\hat{F}_{i^*,n_1})\geq\Delta_i}.
\end{align*}
By Assumption~\ref{as:MAIN} and the DKWM inequality, each summand on the far right-hand side of the above display is no greater than
\begin{align*}
	4\Delta_i\exp\del[1]{-\lfloor n_1/K\rfloor \Delta_i^2/(2C^2)}
	\leq
	4\max_{z>0}\sbr[2]{ z\exp\del[1]{-\lfloor n_1/K\rfloor z^2/(2C^2)}}
	\leq
	\frac{4C}{\sqrt{\exp(1)\lfloor n_1/K\rfloor}}.
\end{align*}
Thus, as there are at most~$K-1$ summands in the penultimate display, one obtains that
\begin{align*}
	\E\del[1]{\max_{i \in \mathcal{I}} \mathsf{W}(F^i) - \mathsf{W}(F^{\pi_n^c(Z_{n_1})})}
	\leq
	\frac{4C(K-1)}{\sqrt{\exp(1)\lfloor n_1/K\rfloor}}
	\leq 
	\delta.
\end{align*}

%
%Let~$F^1$ and~$F^2$ in~$\mathscr{D} = D_{cdf}([0,1])$ be arbitrary, let~$\delta > 0$, and set~$n_1 = 2 \lceil 16 C^2 /(\delta^2 \exp(1)) \rceil$. If~$\mathsf{W}(F^1) = \mathsf{W}(F^2)$, then~$\E(\max_{i \in \mathcal{I}} \mathsf{W}(F^i) - \mathsf{W}(F^{\pi_n^c(Z_{n_1})})) = 0$. Suppose next that~$\Delta := \mathsf{W}(F^1) - \mathsf{W}(F^2) > 0$. Then,
%%
%\begin{align}
%\E(\max_{i \in \mathcal{I}} \mathsf{W}(F^i) - \mathsf{W}(F^{\pi_n^c(Z_{n_1})})) 
%&= 
%\Delta \P(\mathsf{W}(\hat{F}_{2, n_1}) > \mathsf{W}(\hat{F}_{1, n_1})
%) \\
%&=
%\Delta \P(\mathsf{W}(\hat{F}_{2, n_1}) - \mathsf{W}(F^2) + \mathsf{W}(F^1) -  \mathsf{W}(\hat{F}_{1, n_1}) > \Delta
%),
%\end{align}
%%
%which, by Assumption~\ref{as:MAIN} and the DKWM inequality, is not greater than
%%
%\begin{equation}
%4 \Delta  \exp(- \lfloor n_1/2 \rfloor \Delta^2/(2C^2)) \leq 4 \max_{z > 0} \left[ z  \exp(- \lfloor n_1/2 \rfloor z^2/(2C^2))\right] \leq 4 \sqrt{C^2/\lfloor n_1/2 \rfloor} e^{-1/2} \leq \delta.
%\end{equation}
%%
%The remaining case where~$\mathsf{W}(F^2) > \mathsf{W}(F^1)$ is established analogously.
  
\section{Further details for Section~\ref{sec:app}}\label{sec:Data}
This section provides further details on the data sets used in Section \ref{sec:app}. 

\begin{enumerate}
\item The data used for the cognitive abilities program can be downloaded at \url{https://journals.plos.org/plosone/article?id=10.1371/journal.pone.0134467#sec025}. For both treatments, the outcome is a performance summary score~$x\in\R$ (referred to as the "Grand Index" in the corresponding article~\cite{hardy2015enhancing}) on a neuropsychological assessment battery (at an IQ scale). This summary score was transformed into~$[0,1]$ via computing percentile ranks~$x\mapsto\Phi^{-1}\del[1]{(x-100)/15}$, where~$\Phi(\cdot)$ is the cdf of the standard normal distribution. The number of subjects who were assigned to the cognitive training tasks is~$2{,}667$, while~$2{,}048$ were assigned to solving crossword puzzles.
\item The data used for the Detroit Work First program can be downloaded at \url{https://www.journals.uchicago.edu/doi/suppl/10.1086/687522}. For all three treatments, we removed the~$1\%$ of individuals with highest earnings. Upon doing so, for each treatment we scaled the data into~$[0,1]$ by dividing earnings by the largest earning over all treatment groups. The number of individuals with no job, a temporary job and a direct hire job were~$19{,}084,\ 3{,}593$ and~$14{,}112$, respectively (after removing~$1\%$ of the individuals from each treatment arm as outlined above).  
\item The data used for the Pennsylvania Reemployment Bonus experiment can be downloaded at \url{http://qed.econ.queensu.ca/jae/2000-v15.6/bilias/}. For all treatments, the outcome is an unemployment duration~$x\in \cbr[0]{1,\ldots,52}$, which was scaled into~$[0,1]$ via~$x\mapsto 1-(x-1)/51$, also ensuring that larger values correspond to ``better'' treatment outcomes. The number of observations for each treatment arm are~$3{,}354,\ 1{,}385,\ 2{,}428,\ 1{,}885,\ 3{,}030$ and~$1{,}831$, respectively. A precise description of the treatments can be found in Tables 1 and 2 of \cite{bilias2000sequential} where we note that, like them, we have merged treatments 4 and 6.
\end{enumerate}

\begin{figure}[H]
\centering
\includegraphics[height=7.5cm,width=7.5cm]{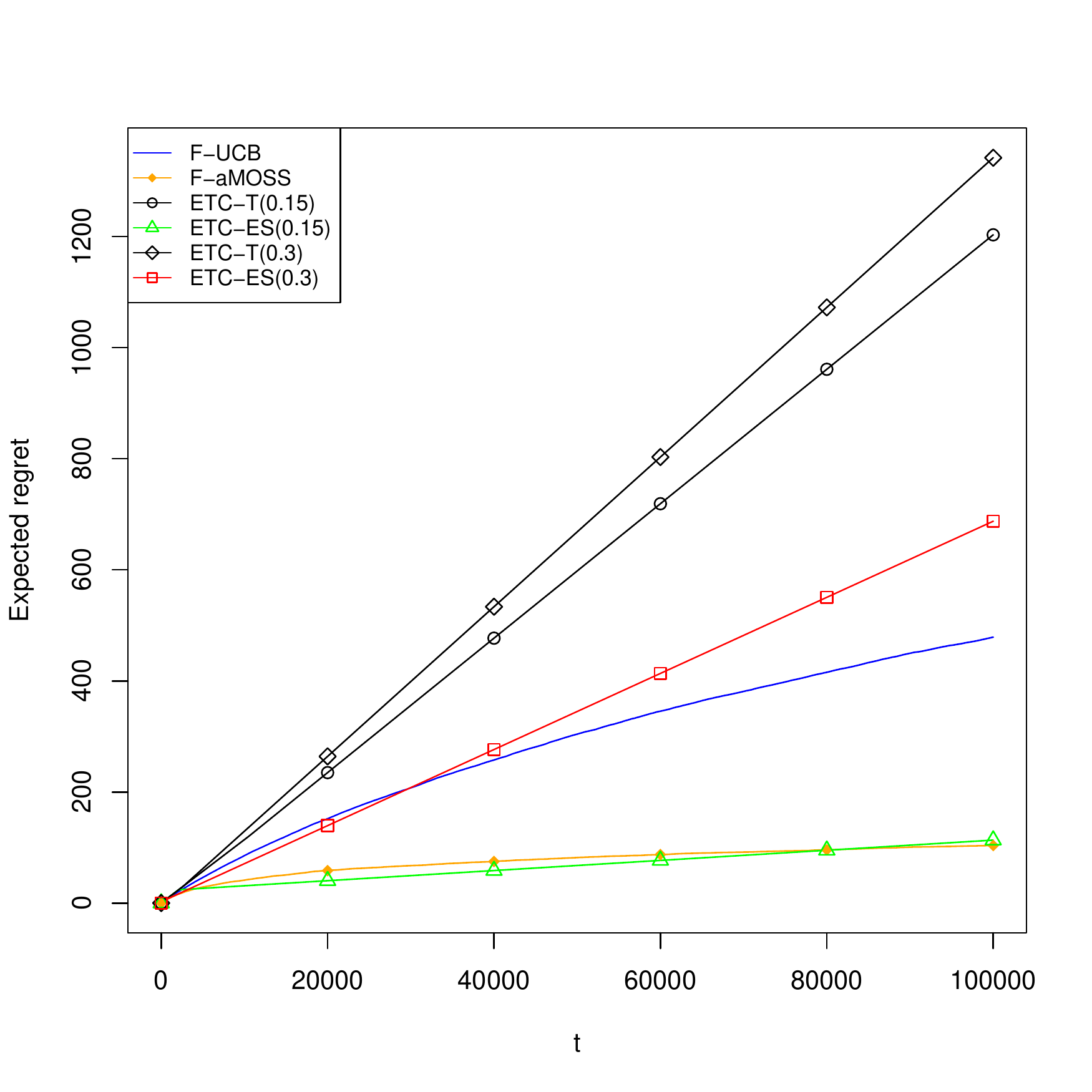}
\includegraphics[height=7.5cm,width=7.5cm]{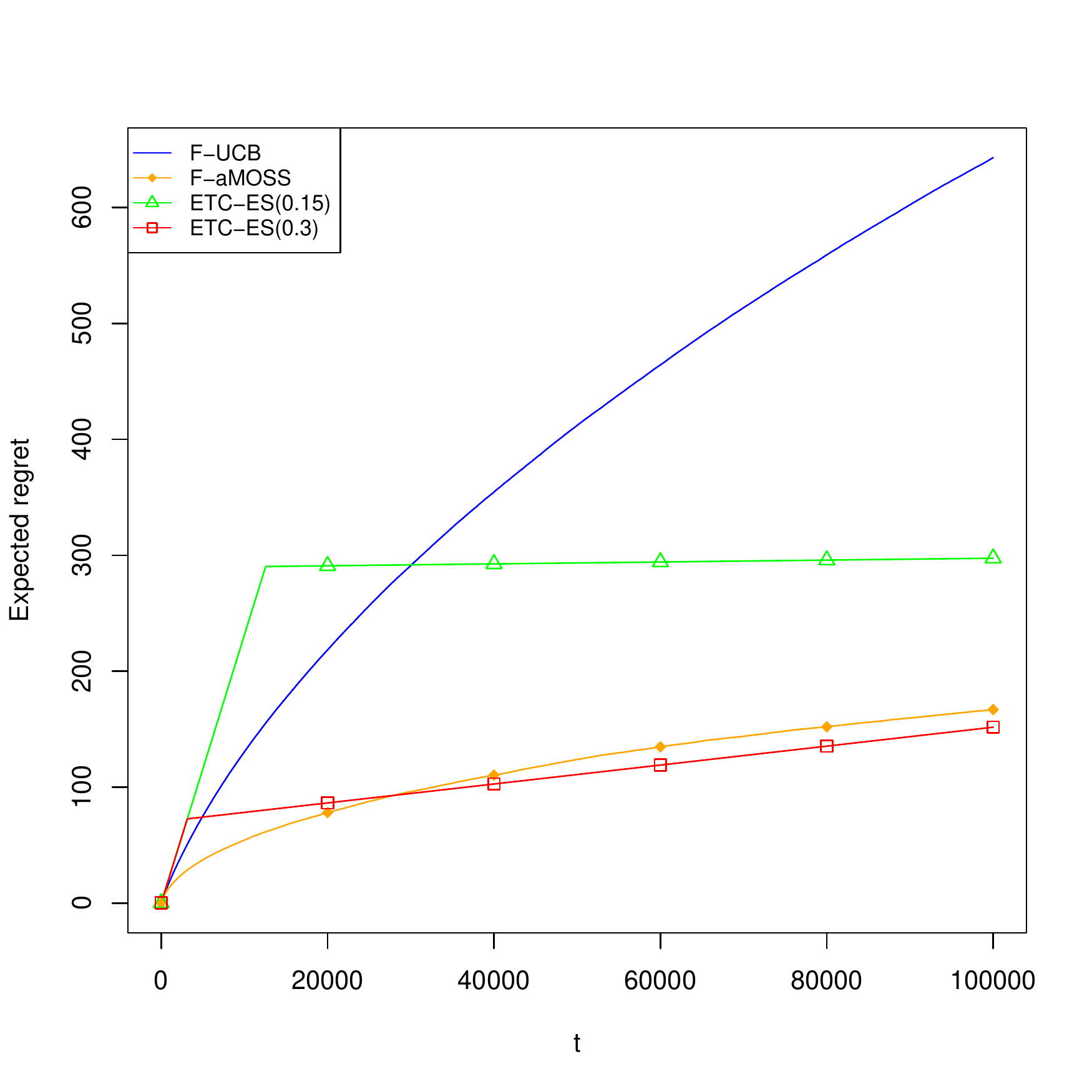}
\includegraphics[height=7.5cm,width=7.5cm]{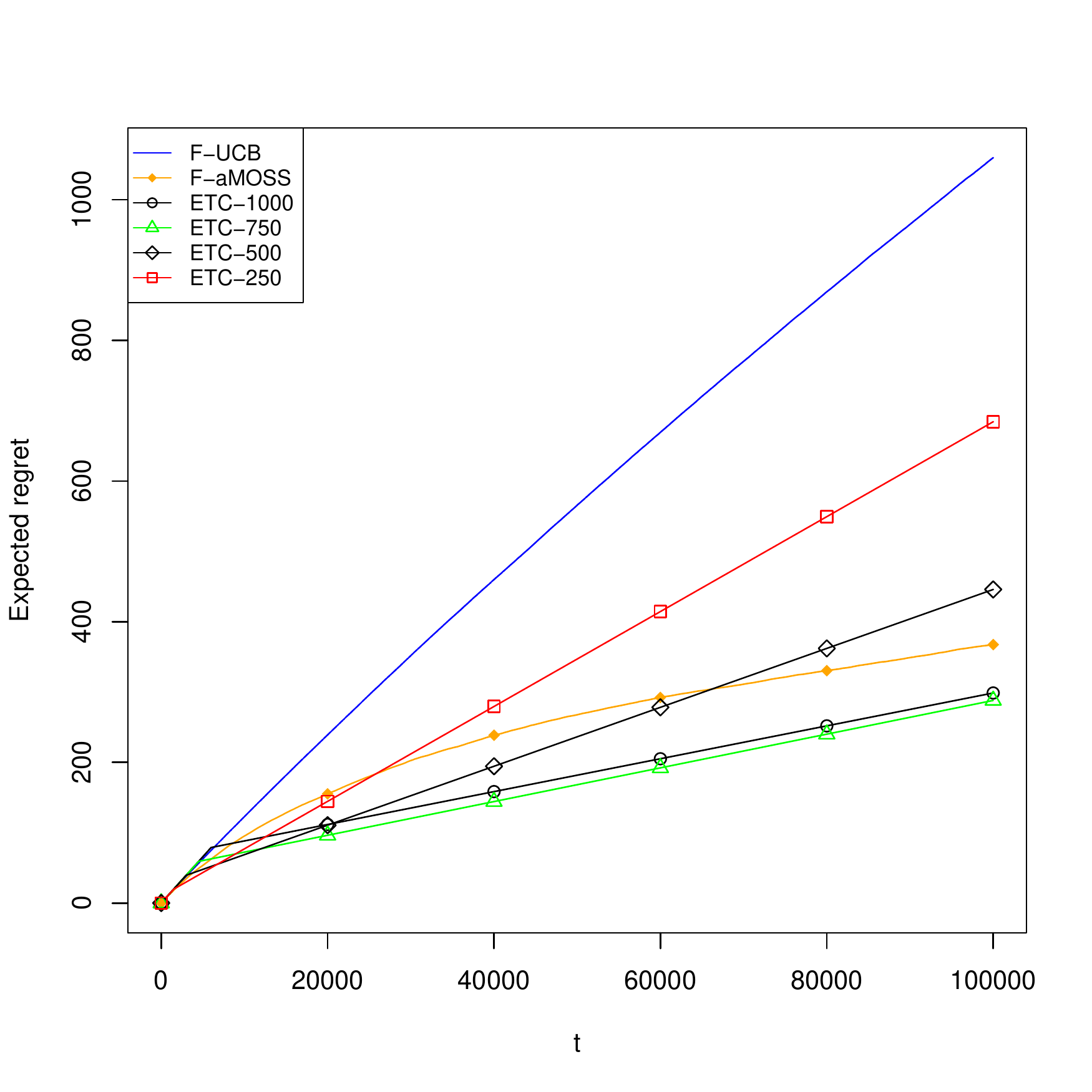}
\caption{\small The figure contains the expected regret for the Schutz-based-welfare measure. Top-left: Cognitive abilities program, top-right: Detroit work first program, bottom: Pennsylvania reemployment bonus program.}
\label{fig:SchutzApp}
\end{figure}

\newpage

\begin{figure}[H]
\centering
\includegraphics[height=7.5cm,width=7.5cm]{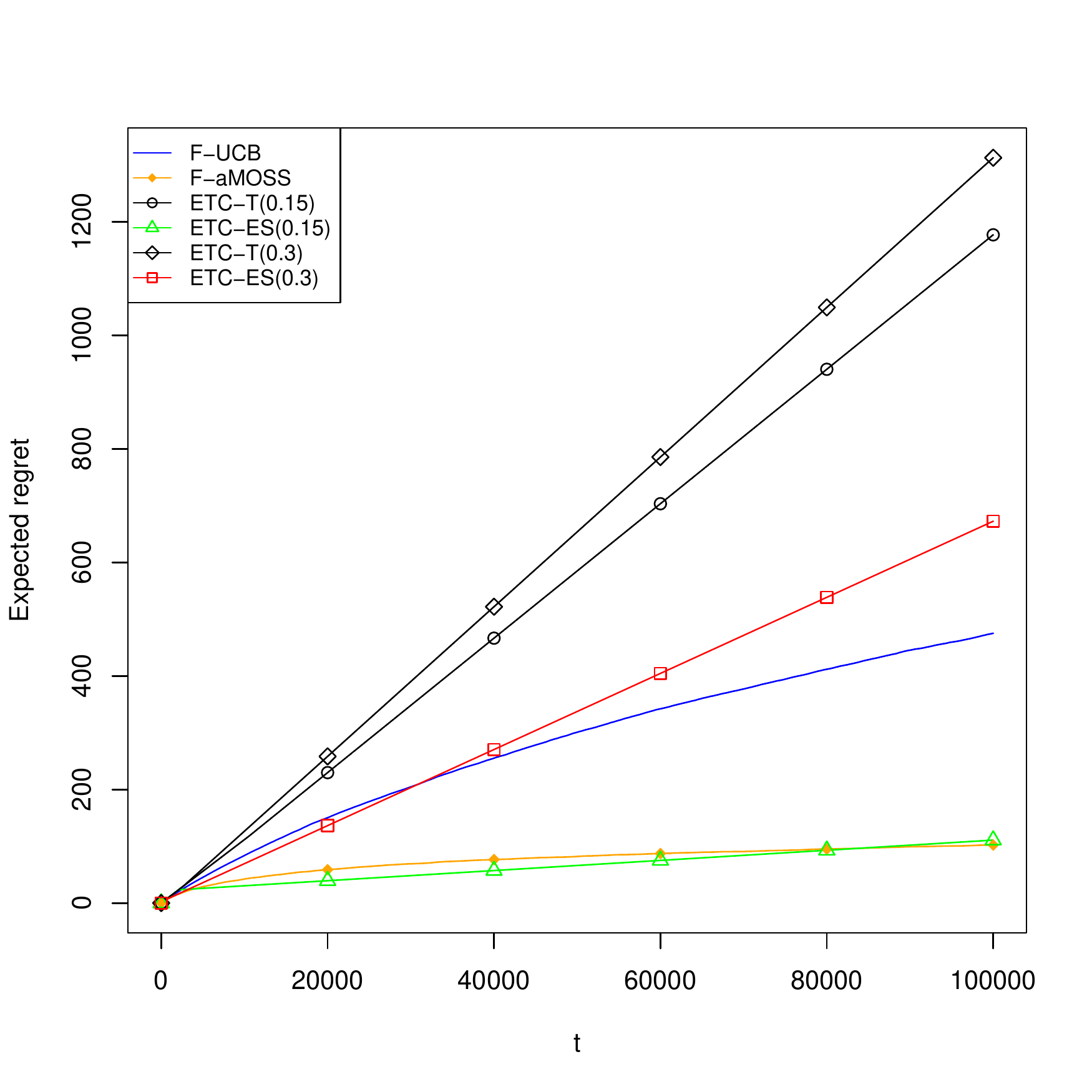}
\includegraphics[height=7.5cm,width=7.5cm]{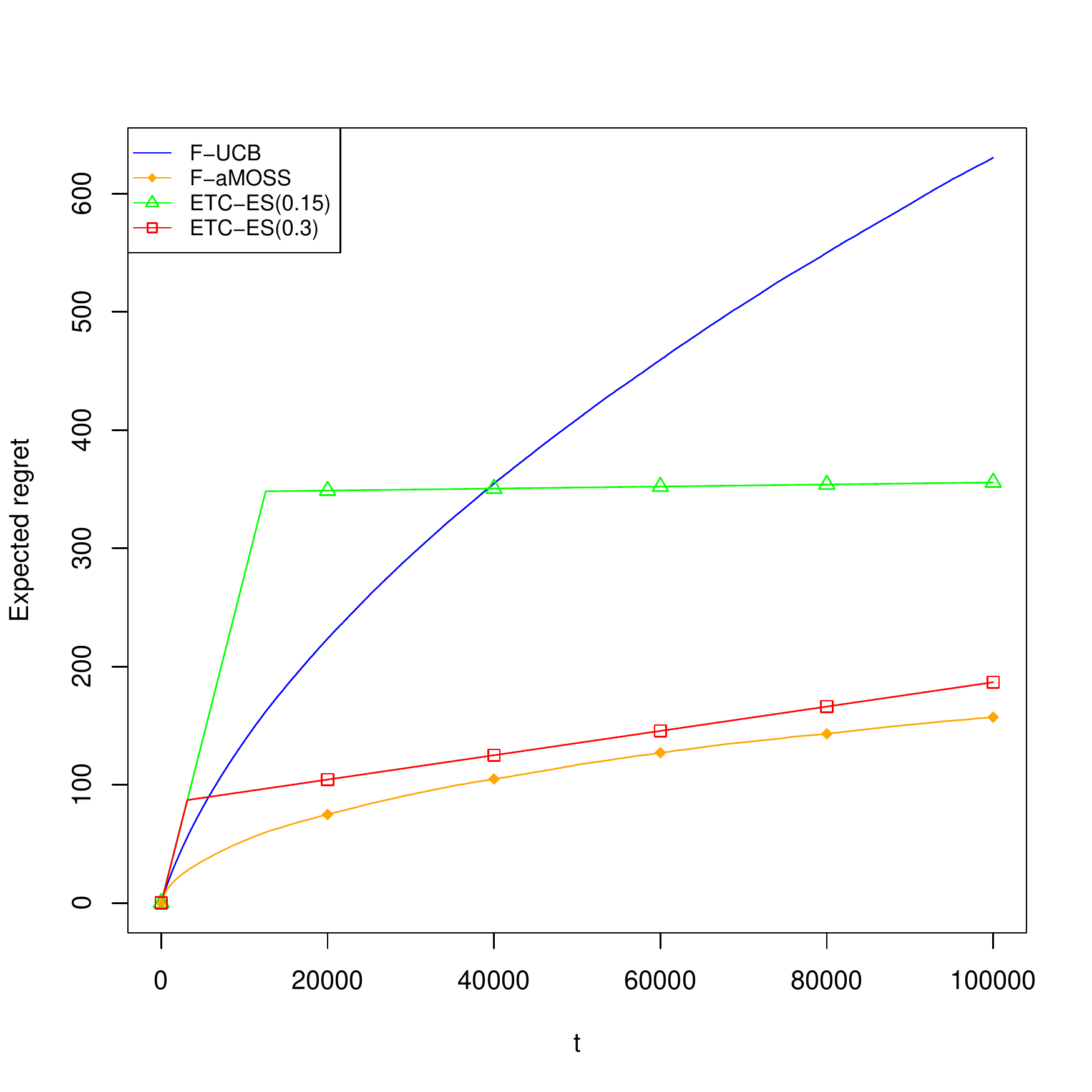}
\includegraphics[height=7.5cm,width=7.5cm]{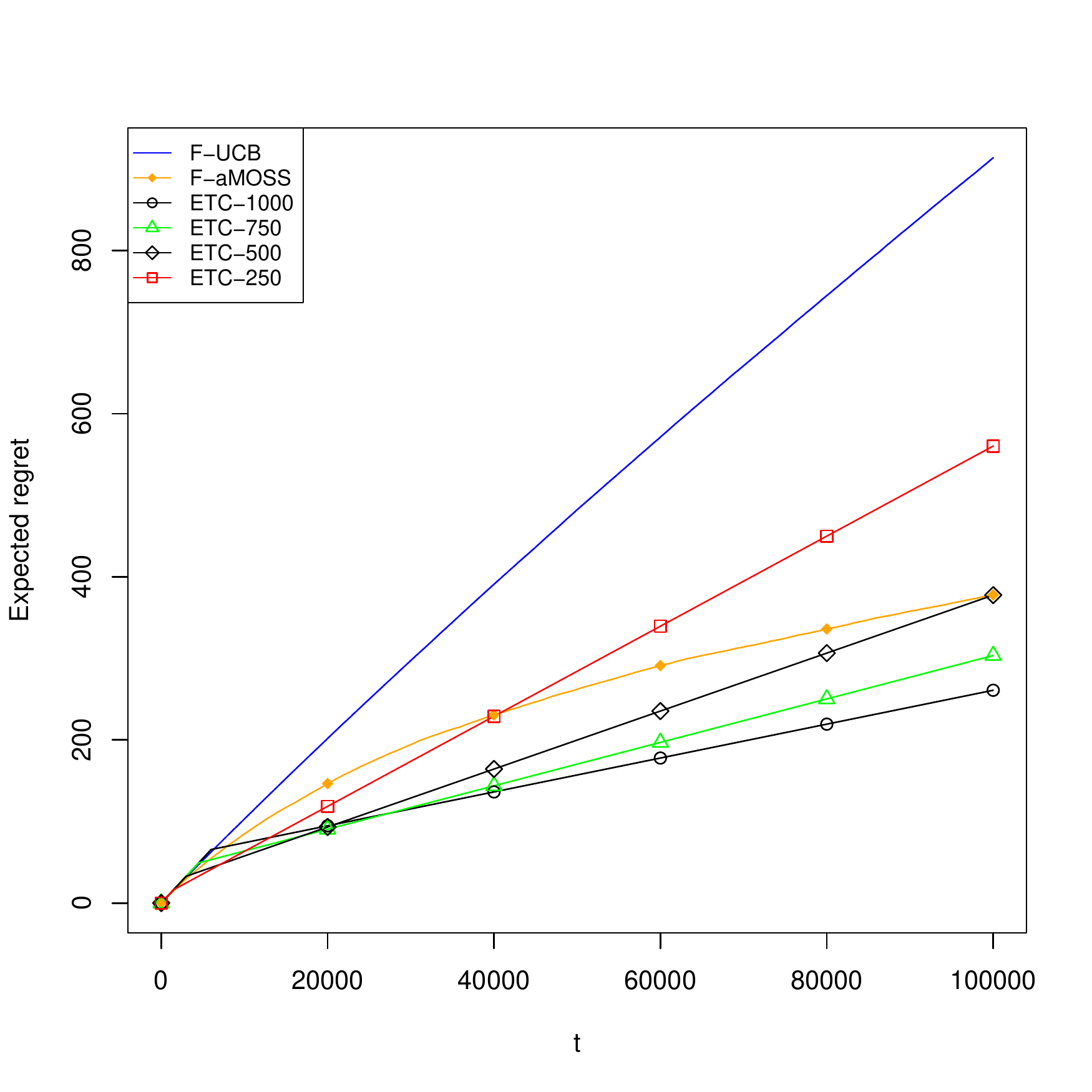}
\caption{\small The figure contains the expected regret for the Atkinson-based-welfare measure with~$\eps=0.5$. Top-left: Cognitive abilities program, top-right: Detroit work first program, bottom: Pennsylvania reemployment bonus program.}
\label{fig:Atkinson05App}
\end{figure}

\newpage

\begin{figure}[H]
\centering
\includegraphics[height=7.5cm,width=7.5cm]{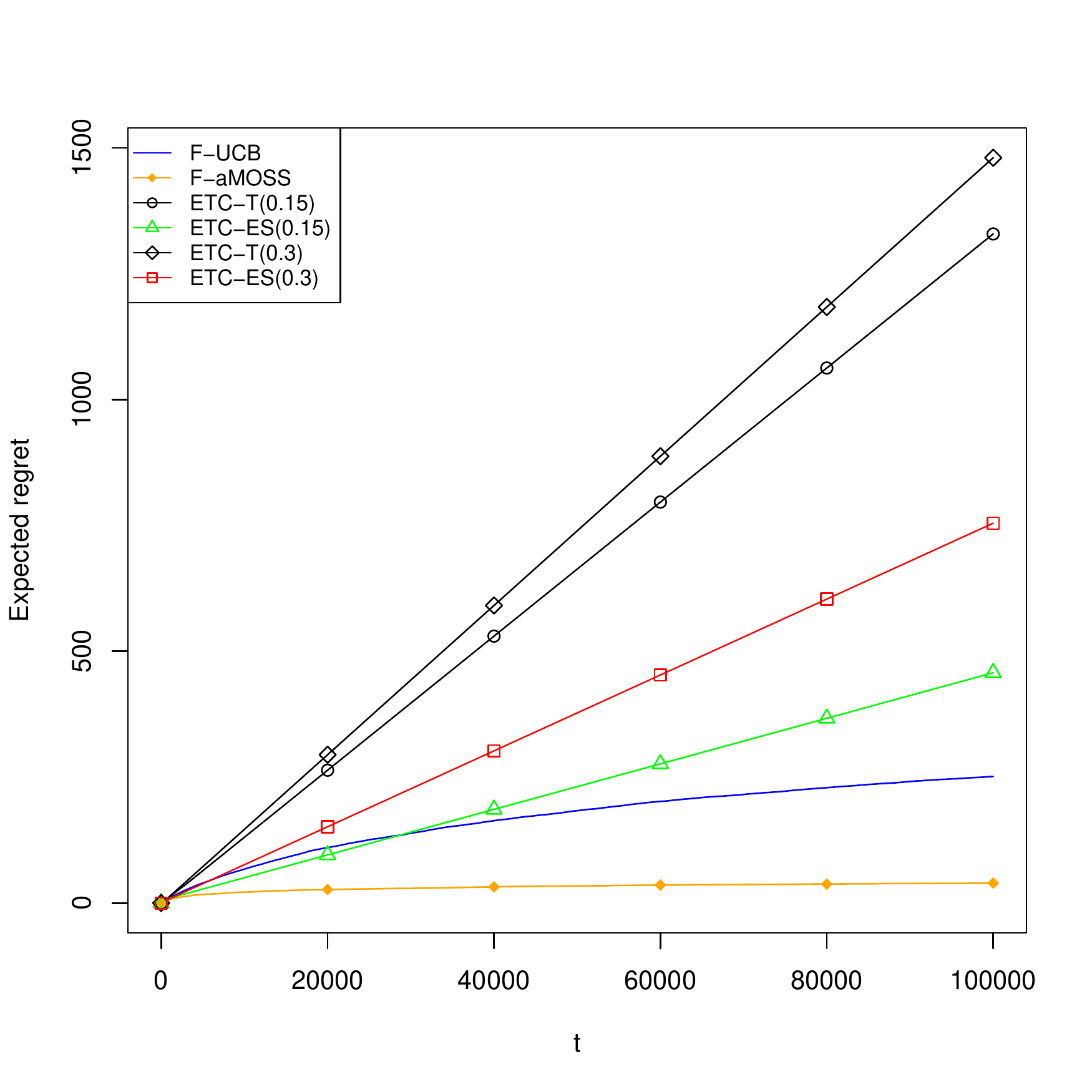}
\includegraphics[height=7.5cm,width=7.5cm]{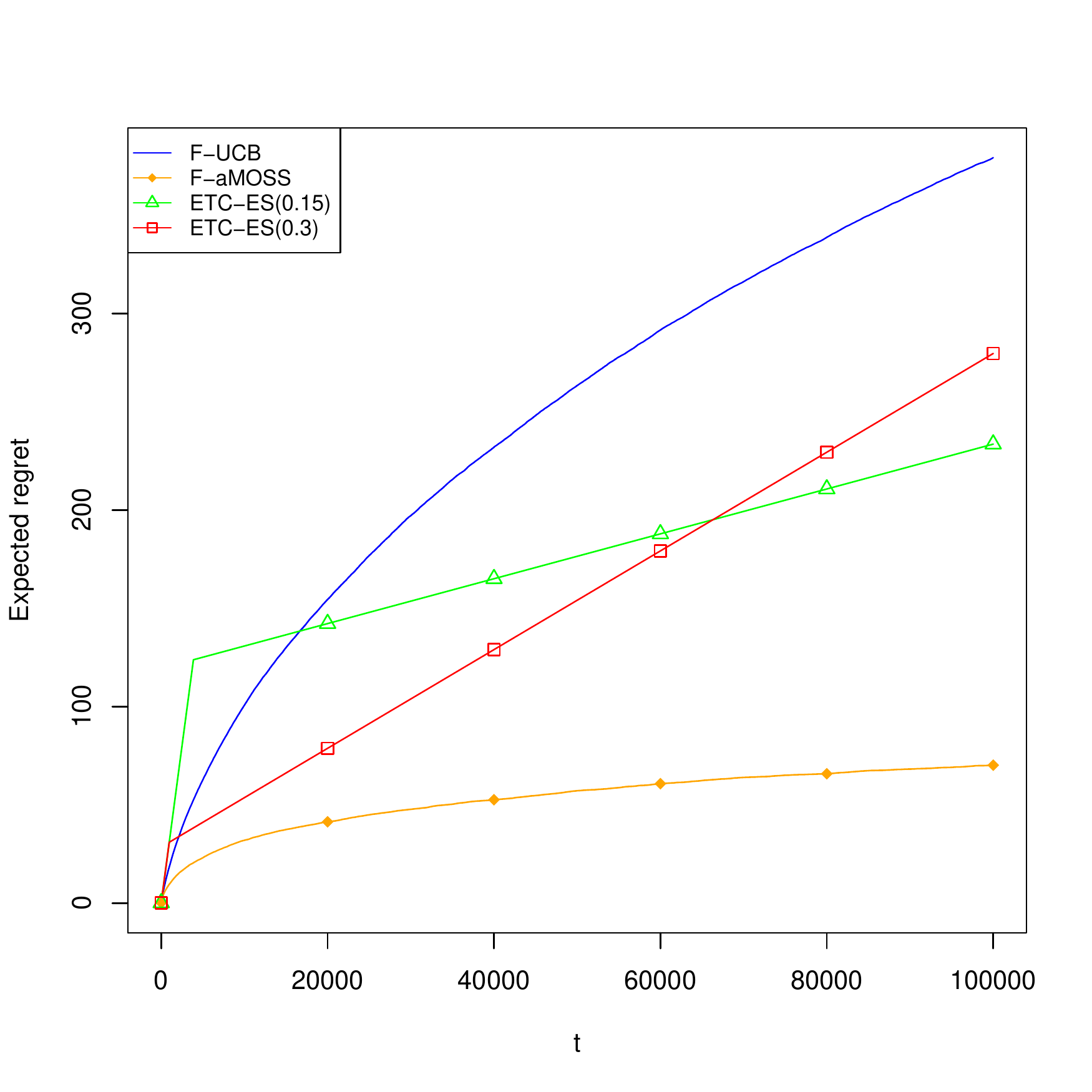}
\includegraphics[height=7.5cm,width=7.5cm]{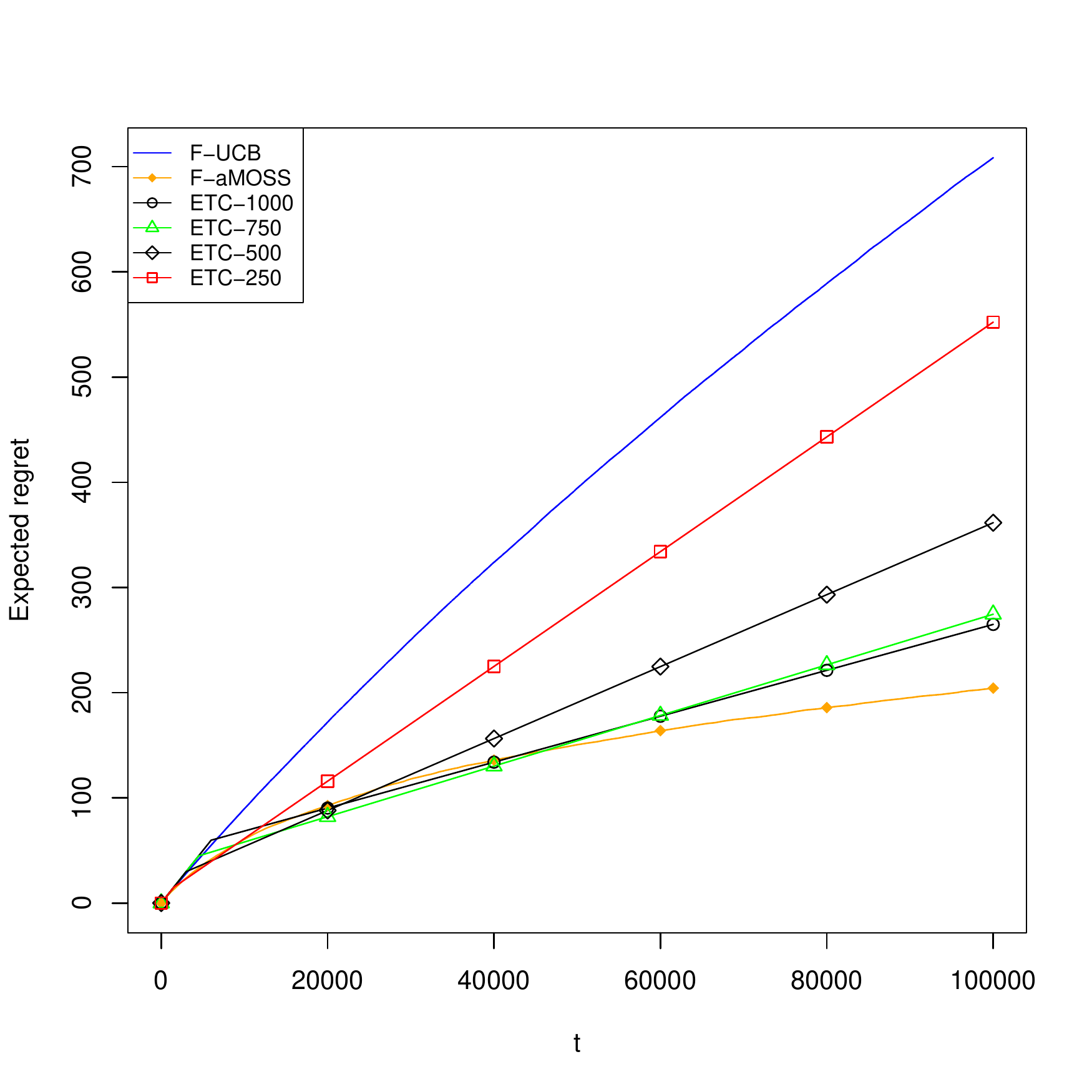}
\caption{\small The figure contains the expected regret for the Atkinson-based-welfare measure with~$\eps=0.1$. Top-left: Cognitive abilities program, top-right: Detroit work first program, bottom: Pennsylvania reemployment bonus program.}
\label{fig:Atkinson01App}
\end{figure}

\section{Assumption~\ref{as:MAIN} in socio-economic applications: inequality, welfare, and poverty measures}\label{app:FUNC}

To illustrate the scope of our results, and to facilitate their implementation in practice, we shall now discuss several functionals of interest in applied economics that satisfy Assumption~\ref{as:MAIN}. We also provide a corresponding set~$\mathscr{D}$ and a constant~$C$. Appendix~\ref{sec:LCgeneral} contains a toolbox of \emph{general} methods for verifying Assumption~\ref{as:MAIN}. The results in the present section are established using these techniques. Therefore, in addition to their intrinsic importance, the following results, and in particular their proofs, also provide a pattern as to how Assumption~\ref{as:MAIN} can be verified for functionals that we do not explicitly discuss.
Before we proceed to these results, we introduce some nonparametric classes of cdfs~$\mathscr{D}$ which will play a major role. The proofs of all results discussed in the present section can be found in Appendix~\ref{ap:verif}.

\subsection{Important classes of cdfs~$\mathscr{D}$}\label{sec:cdfs}

Recall that $a < b$ are throughout assumed to be real numbers. We shall consider the following classes of cdfs.
\begin{enumerate}
\item $\mathscr{D}^s([a,b])$: The subset of all~$F \in D_{cdf}([a,b])$ that are continuous when restricted to~$[a,b]$, and are right-differentiable on~$(a,b)$, with right-sided derivative~$F^+$, say, satisfying~$F^+(x) \leq s$ for all~$x \in (a,b)$.
\item $\mathscr{D}_r([a,b])$: The subset of all~$F \in D_{cdf}([a,b])$ that are continuous when restricted to~$[a,b]$, and right-differentiable on~$(a,b)$, with~$F^+(x) \geq r$ for all~$x \in (a,b)$.
\item $\mathscr{D}_r^s([a,b]) := \mathscr{D}^s([a,b]) \cap \mathscr{D}_r([a,b])$.
\item Furthermore, the subset of all~$F\in \mathscr{D}^s([a,b])$ that are \emph{everywhere} continuous shall be denoted by~$\mathscr{C}^s([a,b])$, and we correspondingly define~$\mathscr{C}_r([a,b])$ and~$\mathscr{C}_r^s([a,b])$.
\end{enumerate}

Note that if $F\in D_{cdf}([a,b])$ is differentiable with a density $f = F'$ that, on~$[a,b]$, is bounded from below by $r$ and from above by $s$, then~$F \in \mathscr{C}_r^s([a,b])$. The set~$\mathscr{C}_r^s([a,b])$ is contained in all classes of cdfs defined in~1.-4. above. Hence, one can think of the (strongest) assumptions imposed above as putting a lower or an upper bound on the unknown densities of the outcome distributions.

\subsection{Inequality measures}\label{sec:inequalitymeasures}

In this section we verify Assumption~\ref{as:MAIN} for functionals that aim to measure the degree of inequality inherent to a (e.g., income, wealth or productivity) distribution~$F$. Such \emph{inequality measures} are relevant in situations where one intends to select that treatment (e.g., one out of several possible taxation schemes) which leads to the most ``equal'' outcome distribution. To avoid possible misunderstandings, we emphasize that it is neither our goal to discuss theoretical foundations of inequality measures, nor to point out their relative advantages and disadvantages. The functional must be chosen by the applied researcher, who can---in making such a choice---rely on excellent book-length treatments, e.g., \cite{lambert}, \cite{chakravarty2009} or \cite{cowell}, as well as the original sources, some of which we shall point out further below. Rather, our goal is to demonstrate that Assumption~\ref{as:MAIN} is satisfied for a broad range of practically relevant functionals. We also emphasize that the  inequality measures discussed in the present section are important building blocks in constructing welfare measures, which will be the topic of discussion in Appendix~\ref{sec:welfaremeasures}.

%We shall first consider two structurally simple inequality measures, which are obtained by evaluating the cdf~$F$ at a ``measure of centrality''. We focus on the two most relevant centrality measures, namely the mean functional~$\mu(F)$, and the median functional, i.e., the~$1/2$-quantile of~$F$, which is denoted by~$q_{1/2}(F)$. The inequality measure corresponding to the mean functional equals~$$\mathsf{MM}(F) := F(\mu(F)),$$ and is known as the \emph{minimal majority index}; the inequality measure corresponding to the median functional equals~$$\mathsf{ES}(F) := F(q_{1/2}(F)),$$ and is known as the \emph{equal share coefficient} (cf.~\cite{cowell}). The following two lemmas provide simple conditions under which the minimal majority index and the equal share coefficient, respectively, satisfy Assumption~\ref{as:MAIN}.
%%{lem:MM}
%\begin{lemma}\label{lem:MM}
%Let~$a < b$ and~$s > 0$ be real numbers. Then, the minimal majority index~$\mathsf{T} = \mathsf{MM}$ satisfies Assumption~\ref{as:MAIN} with~$\mathscr{D} = \mathscr{D}^s([a,b])$ and~$C = s(b-a) + 1$.
%\end{lemma}
%%
%%
%\begin{lemma}\label{lem:ES}
%Let~$a < b$ and~$0 < r < s$ be real numbers. Then, the equal share coefficient~$\mathsf{T} = \mathsf{ES}$ satisfies Assumption~\ref{as:MAIN} with~$\mathscr{D} = \mathscr{C}_r^s([a,b])$ and~$C = \frac{s}{r} + 1$. 
%\end{lemma}
%%
We first discuss inequality measures that derive from the Lorenz curve (cf.~\cite{gastwirth} or Equation~\eqref{def:linineq} below for a formal definition). The first such inequality measure we consider is the \emph{Schutz-coefficient}~$\mathsf{S}_{\text{rel}}$ (cf.~\cite{schutz}, \cite{schutzcomment}), say, which is also known as the \emph{Hoover-index} or the \emph{Robin Hood-index}. Formally, 
\begin{equation}
\mathsf{S}_{\text{rel}}(F) = \frac{1}{2\mu(F)} \int |x - \mu(F)| dF(x),
\end{equation}
provided the mean~$\mu(F) := \int x dF(x)$ exists and is nonzero. The subindex ``rel'' in~$\mathsf{S}_{\text{rel}}(F)$ signifies that this index is defined ``relative'' to the mean. Note that, as a consequence, if one \emph{multiplies} each income by the same (positive) amount this does not result in a change of the inequality index, i.e., the index is \emph{scale independent}. A corresponding ``absolute'' variant, i.e., a measure which remains unchanged if one \emph{adds} to every income the same amount, is obtained by multiplying the relative measure~$\mathsf{S}_{\text{rel}}$ by the mean functional, and is denoted by
\begin{equation}\label{eqn:schutzabs}
\mathsf{S}_{\text{abs}}(F) = \frac{1}{2} \int |x - \mu(F)| dF(x).
\end{equation}
For a discussion of relative and absolute inequality measures we refer to \cite{kolm1, kolm2}, who calls them ``rightist'' and ``leftist,'' respectively. As a general rule, absolute inequality indices require less restrictive assumptions on~$\mathscr{D}$ than their relative counterparts in order to satisfy Assumption~\ref{as:MAIN}. This is due to the fact that division by~$\mu(F)$ is highly unstable for small values of~$\mu(F)$. The following lemma provides conditions under which the relative and absolute Schutz-coefficient satisfy Assumption~\ref{as:MAIN}.

\begin{lemma}\label{lem:schutzindex}
Let~$a < b$ be real numbers. Then the absolute Schutz-coefficient~$\mathsf{T} = \mathsf{S}_{\text{abs}}$ satisfies Assumption~\ref{as:MAIN} with~$\mathscr{D} = D_{cdf}([a,b])$ and~$C = b-a$. Next, assume that~$a \geq 0$, and define for every~$\delta \in (a,b)$ and every~$s > 0$ the set
\begin{equation}
	\mathscr{D}(s, \delta) := \{F \in \mathscr{C}^s([a,b]): \mu(F) \geq \delta\}.
\end{equation}
Then, for every~$\delta \in (a,b)$ and every~$s > 0$, the relative Schutz-coefficient~$\mathsf{T} = \mathsf{S}_{\text{rel}}$ (\emph{defined} as~$0$ for the cdf corresponding to point mass~$1$ at~$0$) satisfies Assumption~\ref{as:MAIN} with~$\mathscr{D} = \mathscr{D}(s, \delta)$ and~$C = (b-a)(2s+\delta^{-1}) + 5$.
\end{lemma}

The next inequality measure we consider is the \emph{Gini-index}. Formally, its relative variant is defined as
\begin{equation}\label{eqn:Ginirel}
\mathsf{G}_{\text{rel}}(F) = \frac{1}{2\mu(F)}\int \int |x_1 - x_2| dF(x_1) dF(x_2),
\end{equation}
provided that the expression is well defined. A corresponding absolute inequality measure is %
\begin{equation}\label{eqn:Giniabs}
\mathsf{G}_{\text{abs}}(F) = \frac{1}{2}\int \int |x_1 - x_2| dF(x_1) dF(x_2).
\end{equation}
The following lemma provides conditions under which Assumption~\ref{as:MAIN} is satisfied for these two Gini-indices.

\begin{lemma}\label{lem:giniindex}
Let~$a < b$ be real numbers. Then the absolute Gini-index~$\mathsf{T} = \mathsf{G}_{\text{abs}}$ satisfies Assumption~\ref{as:MAIN} with~$\mathscr{D} = D_{cdf}([a, b])$ and~$C = b-a$. Next, assume that~$a \geq 0$, and define for every~$\delta \in (a,b)$ the set
\begin{equation}\label{eqn:DDelta}
	\mathscr{D}(\delta) := \{F \in D_{cdf}([a,b]): \mu(F) \geq \delta\}.
\end{equation}
Then, for every~$\delta \in (a,b)$, the relative Gini-index~$\mathsf{T} = \mathsf{G}_{\text{rel}}$ (defined as~$0$ for the cdf corresponding to point mass~$1$ at~$0$) satisfies Assumption~\ref{as:MAIN} with~$\mathscr{D} = \mathscr{D}(\delta)$ and~$C = 2 \delta^{-1} (b-a)$.
\end{lemma}
The Gini-index belongs to the class of \emph{linear inequality measures} introduced by \cite{mehran} (cf.~in particular Equation 3 there). An inequality measure is called \emph{linear}, if it is a functional of the form
\begin{equation}\label{def:linineq}
F \mapsto \int_{[0, 1]} (u - L(F, u)) dW(u), \quad \text{ where } \quad L(F, u) := \mu(F)^{-1} \int_{[0, u]} q_{\alpha}(F)d\alpha,
\end{equation}
and where~$W$ is a function on~$[0,1]$ that is fixed (i.e., independent of~$F$) with finite total variation. Here~$q_{\alpha}(F) := \inf\{ x \in \R : F(x) \geq \alpha \}$ is the usual $\alpha$-quantile of the cdf~$F$, and~$L(F; u)$ is the Lorenz curve corresponding to~$F$ evaluated at~$u$ (cf.~also the discussion around our Equation~\eqref{def:Lorenz}).  The following lemma provides conditions under which a linear inequality measure satisfies Assumption~\ref{as:MAIN}. The result relies on properties of the Lorenz curve established in Lemma~\ref{lem:lorenz} in Appendix~\ref{sec:LCgeneral}. The class of linear inequality measures is large, and the lemma thus applies quite generally. However, the generality is bought at the price of adding further regularity conditions on~$\mathscr{D}$; in particular $a > 0$ has to be assumed. This trade-off in generality and strength of assumptions becomes apparent by comparing the regularity conditions to the ones in Lemma~\ref{lem:giniindex}. Nevertheless, the result shows that Assumption~\ref{as:MAIN} can be expected to be quite generically satisfied.
\begin{lemma}\label{lem:linin}
Let~$a < b$ be positive real numbers and let~$r > 0$.   Assume that~$W: [0, 1] \to \R$ has finite total variation~$\kappa$, say. Then the functional defined in Equation~\eqref{def:linineq} satisfies Assumption~\ref{as:MAIN} with~$\mathscr{D} = \mathscr{C}_r([a,b])$ and~$C = \kappa a^{-1} ( r^{-1} + (b-a)a^{-1}b ).$
\end{lemma}

An absolute version of the linear inequality measure in Equation~\eqref{def:linineq} can be obtained through multiplication by~$\mu(F)$, i.e., 
\begin{equation}\label{def:abslinineq}
F \mapsto \int_{[0, 1]} (\mu(F)u - Q(F, u)) dW(u), \quad \text{ where } \quad Q(F, u) := \int_{[0, u]} q_{\alpha}(F)d\alpha.
\end{equation}
The following result provides conditions under which such absolute linear inequality measures satisfy~Assumption~\ref{as:MAIN}. As usual, the regularity conditions on~$\mathscr{D}$ required are weaker than the ones needed for the relative version. In particular~$a>0$ does \emph{not} need to be assumed.
\begin{lemma}\label{lem:abslinin}
Let~$a < b$ be real numbers and let~$r > 0$. Assume that~$W: [0, 1] \to \R$ has finite total variation~$\kappa$, say. Furthermore, denote~$|\int_{[0,1]} u dW(u)| =: c$. Then the functional defined in Equation~\eqref{def:abslinineq} satisfies Assumption~\ref{as:MAIN} with~$\mathscr{D} = \mathscr{C}_r([a,b])$ and~$C = c(b-a) + r^{-1} \kappa.$ 
\end{lemma}

Another important family of scale-independent inequality measures is the so-called \emph{generalized entropy family}, cf.~\cite{cowellentropy}: Given a parameter~$c \in \R$, an inequality measure is obtained via (if the involved expressions are well defined)
\begin{equation}\label{eqn:entropy}
\mathsf{E}_{c}(F) = 
\begin{cases}
	\frac{1}{c(c-1)} \int \left[ \left( x/\mu(F) \right)^{c} - 1 \right]dF(x) & \text{ if } c \notin \{0, 1\} \\
	\int \left( x/\mu(F) \right) \log\left( x/\mu(F) \right) dF(x) &
	\text{ if } c = 1 \\
	\int \log\left( \mu(F)/x \right) dF(x) &
	\text{ if } c = 0.
\end{cases}
\end{equation}
The inequality measure corresponding to~$c = 1$ is known as Theil's entropy index (cf.~also \cite{theil}), and the measure corresponding to~$c = 0$ is the mean logarithmic deviation (cf.~\cite{lambert}, p.112). A formal result providing conditions under which a generalized entropy measure satisfies Assumption~\ref{as:MAIN} is presented next. The regularity conditions we need to impose depend on~$c$. Note in particular that support assumptions implicit in~$\mathscr{D}$ are somewhat weaker for~$c \in (0, 1)$.
\begin{lemma}\label{lem:entropy}
Let~$0 \leq a < b$ be real numbers, and let~$c \in \R$.
\begin{enumerate}
	\item If~$c \in (0, 1)$, then, for every~$\delta \in (a,b)$, the functional~$\mathsf{T} = \mathsf{E}_{c}$ (\emph{defined} as~$-1/(c(c-1))$ for the cdf corresponding to point mass~$1$ at~$0$) satisfies Assumption~\ref{as:MAIN} with~$\mathscr{D} = \mathscr{D}(\delta)$ (cf.~Equation~\eqref{eqn:DDelta}) and~$C =|c(c-1)|^{-1}\left[ \delta^{-c}(b^{c} - a^{c}) + \delta^{-1} (b-a) \right]$.
	\item If~$c \notin [0, 1]$ and~$a > 0$, then the functional~$\mathsf{T} = \mathsf{E}_{c}$ satisfies Assumption~\ref{as:MAIN} with~$\mathscr{D} = D_{cdf}([a,b])$ and~$$C = |c(c-1)|^{-1}[\max(a^{-c}, b^{-c})|b^c - a^c| + |c| \max\left((a/b)^{2c-1}, (b/a)^{2c-1}\right)a^{-1}(b-a)].$$
	\item If~$c \in \{0, 1\}$ and~$a > 0$, then the functional~$\mathsf{T} = \mathsf{E}_{c}$ satisfies Assumption~\ref{as:MAIN} with~$\mathscr{D} = D_{cdf}([a,b])$ and~$C = (b-a)/a + \log(b/a)$ if~$c = 0$, and with~$C = \int_{[a/b,b/a]}|1+\log(x)|dx + \frac{b(b-a)}{a^2}  \left\{\log(b/a) +1\right\}$ if~$c = 1$.
\end{enumerate}
\end{lemma}

We continue with a family of relative inequality indices introduced by \cite{atkinson1970}. This family depends on an ``inequality aversion'' parameter~$\varepsilon \in (0, 1) \cup (1, \infty)$. For a fixed~$\varepsilon$ in that range, the index obtained equals (if the involved quantities are well defined)
\begin{equation}
\mathsf{A}_{\varepsilon}(F) = 1- \frac{1}{\mu(F)}\left[\int x^{1-\varepsilon} dF(x)\right]^{1/(1-\varepsilon)}.
\end{equation}
It is well known (cf., e.g., \cite{lambert} p.112) that~$\mathsf{A}_\varepsilon$ can be written as 
\begin{equation}\label{eqn:entropAtk}
\mathsf{A}_{\varepsilon}(F) = 1-[\varepsilon(\varepsilon - 1) \mathsf{E}_{1-\varepsilon}(F) + 1]^{1/(1-\varepsilon)}.
\end{equation}
Together with Lemma~\ref{lem:entropy}, this relation can be used to obtain the following result:
\begin{lemma}\label{lem:atkinson}
Let~$0 \leq a < b$ be real numbers, let~$\varepsilon \in (0, 1) \cup (1, \infty)$ and set~$c(\varepsilon) = 1-\varepsilon$.
\begin{enumerate}
	\item If~$\varepsilon \in (0, 1)$, then, for every~$\delta \in (a,b)$, the functional~$\mathsf{T} = \mathsf{A}_{\varepsilon}$ (\emph{defined} as~$1$ for the cdf corresponding to point mass~$1$ at~$0$) satisfies Assumption~\ref{as:MAIN} with~$\mathscr{D} = \mathscr{D}(\delta)$ (cf.~Equation~\eqref{eqn:DDelta}) and~$C =c(\varepsilon)^{-1}\left[ \delta^{-c(\varepsilon)}(b^{c(\varepsilon)} - a^{c(\varepsilon)}) + \delta^{-1} (b-a) \right]$.
	\item If~$\varepsilon \in (1, \infty)$ and~$a > 0$, then the functional~$\mathsf{T} = \mathsf{A}_{\varepsilon}$ satisfies Assumption~\ref{as:MAIN} with~$\mathscr{D} = D_{cdf}([a,b])$ and $$C = 
	(\eps -1)^{-1}(b/a)^{\varepsilon} [b^{-c(\varepsilon)}(a^{c(\varepsilon)} - b^{c(\varepsilon)}) + |c(\eps)| (a/b)^{2c(\varepsilon)-1}a^{-1}(b-a)].$$
\end{enumerate}
\end{lemma}
As the last example in this section, we proceed to an important family of absolute inequality indices, the \emph{Kolm-indices} (\cite{kolm1}, cf.~also the discussion in Section~1.8.1 of \cite{chakravarty2009}). Given a parameter~$\kappa > 0$ the corresponding index is defined as
\begin{equation}
\mathsf{K}_{\kappa}(F) = \kappa^{-1} \log\left(\int e^{\kappa[\mu(F) - x]} dF(x) \right).
\end{equation}
The following lemma verifies Assumption~\ref{as:MAIN} for this class of inequality indices.
\begin{lemma}\label{lem:kolm}
Let~$a < b$ and let~$\kappa > 0$. Then the functional~$\mathsf{T} = \mathsf{K}_{\kappa}$ satisfies Assumption~\ref{as:MAIN} with~$\mathscr{D} = D_{cdf}([a,b])$ and~$C = e^{\kappa (b-a)}[b-a]+\kappa^{-1} e^{\kappa b}[e^{-\kappa a}-e^{-\kappa b}]$. 
\end{lemma}

\subsection{Welfare measures}\label{sec:welfaremeasures}

The structurally most elementary welfare measures are of the form 
\begin{equation}\label{eqn:welfadd}
F \mapsto \int u(x) dF(x),
\end{equation}
for a utility function~$u$. Functionals as in Equation~\eqref{eqn:welfadd} are accessible to our theory, but are not our main focus, as they fall into the standard multi-armed bandit framework, because a mean is targeted.

There are many important welfare measures that are not of the simple form~\eqref{eqn:welfadd}, but can be obtained as a function of the mean functional and an inequality measure.\footnote{Historically, the theoretical foundation of inequality measures was based on a social welfare function, e.g.,~\cite{dalton1920measurement} and~\cite{atkinson1970}. That is, contrary to our presentation, which started with a discussion of inequality measures, inequality measures were derived from a given social welfare functions. For our presentation, however, it is convenient to base the welfare functions on inequality measures. Articles that start with an inequality measure and derive a corresponding welfare measures from it are \cite{blackorby1978measures}, \cite{blackorby1980theoretical} or \cite{dagum1990relationship}.} Many such measures are related to a relative inequality measure~$F \mapsto \mathsf{I}_{\text{rel}}(F)$, say, via the transformation
\begin{equation}\label{eqn:welfinrel}
\mathsf{W}(F) = \mu(F) (1 - \mathsf{I}_{\text{rel}}(F));
\end{equation}
or are related to an absolute inequality measure~$F \mapsto \mathsf{I}_{\text{abs}}(F)$, say, via the transformation
\begin{equation}\label{eqn:welfinabs}
\mathsf{W}(F) = \mu(F)  - \mathsf{I}_{\text{abs}}(F);
\end{equation}
we refer to \cite{blackorby1978measures}, \cite{blackorby1980theoretical} and \cite{dagum1990relationship} for theoretical background on this relationship between inequality and welfare measures. 

While many important welfare measures are of this form, we do not argue that any relative or absolute inequality measure implies a \emph{reasonable} welfare measure through one of the above two relations. In particular, to arrive at a reasonable welfare measure, one may want to impose additional restrictions on the inequality measure, e.g., one may want to assume that the relative inequality measure in Equation~\eqref{eqn:welfinrel} satisfies~$0 \leq \mathsf{I}_{\text{rel}} \leq 1$, and that the absolute inequality measure in Equation~\eqref{eqn:welfinabs} satisfies~$0 \leq \mathsf{I}_{\text{abs}} \leq \mu$. These conditions are satisfied by many inequality measures; otherwise, they can often be achieved by re-normalization in case the inequality measures are nonnegative and bounded from above, cf.~also the discussion in \cite{chakravarty2009} p.30. Apart from a boundedness condition concerning the relative inequality measures, such restrictions are not needed in our proof verifying Assumption~\ref{as:MAIN} for welfare measures obtained through~\eqref{eqn:welfinrel} and~\eqref{eqn:welfinabs}, and are therefore not incorporated into the lemma given below.

\begin{example}
The Gini-welfare measure from Equation~\eqref{eqn:GINIWELFAREM} is obtained upon choosing~$\mathsf{I}_{\text{abs}} = \mathsf{G}_{\text{abs}}$ (cf.~Equation~\eqref{eqn:Ginirel}) in Equation~\eqref{eqn:welfinabs}. That~$0\leq \mathsf{G}_{\mathrm{abs}} \leq \mu$ is well known; an argument may be found in the proof of Lemma~\ref{lem:giniindex}.
\end{example} 

The following result allows one to use the results from the preceding section in establishing Assumption~\ref{as:MAIN} for welfare measures derived via~\eqref{eqn:welfinrel} and~\eqref{eqn:welfinabs} from an inequality measure. 
\begin{lemma}\label{lem:derivedwelfare}
Let~$a < b$ be real numbers. Then, the following holds:
\begin{enumerate}
	\item Let the relative inequality measure~$\mathsf{I}_{\text{rel}}$ satisfy Assumption~\ref{as:MAIN} with~$\mathscr{D}_{\text{rel}}$ and~$C$. Suppose further that~$|1-\mathsf{I}_{\text{rel}}| \leq \gamma < \infty$ holds. Then the welfare measure~$\mathsf{W}$ derived via Equation~\eqref{eqn:welfinrel} satisfies Assumption~\ref{as:MAIN} with~$\mathscr{D} = \mathscr{D}_{\text{rel}}$ and constant~$\gamma (b-a) + \max(|a|,|b|)C$.  
	\item Let the absolute inequality measure~$\mathsf{I}_{\text{abs}}$ satisfy Assumption~\ref{as:MAIN} with~$\mathscr{D}_{\text{abs}}$ and~$C$. Then the welfare measure~$\mathsf{W}$ derived via Equation~\eqref{eqn:welfinabs} satisfies Assumption~\ref{as:MAIN} with $\mathscr{D} = \mathscr{D}_{\text{abs}}$ and with constant~$(b-a) + C$. 
\end{enumerate}
\end{lemma}
Note that if an absolute inequality measure~$\mathsf{I}_{\mathrm{abs}}$ and a relative inequality measure~$\mathsf{I}_{\mathrm{rel}}$ are related via~$\mathsf{I}_{\mathrm{abs}}(F) = \mu(F) \mathsf{I}_{\mathrm{rel}}(F)$ for every~$F$ (\cite{blackorby1980theoretical} then call~$\mathsf{I}_{\mathrm{abs}}$ a ``compromise index," cf.~their Section 5), then the welfare measures obtained via Equation~\eqref{eqn:welfinrel} and Equation~\eqref{eqn:welfinabs}, respectively, coincide. One can then verify Assumption~\ref{as:MAIN} via Part 1 or Part 2 in Lemma~\ref{lem:derivedwelfare}. We note that in such a situation Part~2 of the lemma will typically imply weaker restrictions. 

Together with the results in the preceding section, Lemma~\ref{lem:derivedwelfare} verifies Assumption~\ref{as:MAIN} for many specific welfare measures.
For example, Lemma~\ref{lem:giniindex} can be used to show that the Gini-welfare measure satisfies Assumption~\ref{as:MAIN} with~$a<b$ real numbers,~$\mathscr{D} = D_{cdf}([a, b])$, and constant~$C = 2(b-a)$. Similarly, Lemma~\ref{lem:abslinin} can be used to verify Assumption~\ref{as:MAIN} for all welfare measures corresponding to linear inequality measures. The latter class of welfare measures was recently considered in a different context by \cite{kit2017}.
%This relationship allows us to leverage the results obtained in the preceding section. The first measure we consider is the Gini social welfare function (obtained from the Gini inequality index)
%%
%\begin{equation}
%\mathsf{G}^w(F) = \mu(F) - \mathsf{G}_{\text{abs}}(F);
%\end{equation}
%%
%the second type of such social welfare functions are obtained from the family of Atkinson inequality indices
%%
%\begin{equation}
%A_{\varepsilon}^w(F) = \mu(F) - \left[\int x^{1-\varepsilon} dF(x) \right]^{1/(1-\varepsilon)}.
%\end{equation}
%%
%Conditions under which these two important measures satisfy Assumption~\ref{as:MAIN} are summarized in the following two lemmata.
%
%\begin{lemma}\label{lem:GiniSWF}
%Let~$a < b$ be real numbers and let~$\mathscr{D} = D_{cdf}([a,b])$. For this choice of~$a,b$ and~$\mathscr{D}$ the functional~$T = \mathsf{G}^w$ satisfies Assumption~\ref{as:MAIN} with~$C = 2(b-a)$.
%\end{lemma}
%
%\begin{lemma}\label{lem:AtkiSWF}
%Let~$a < b$ be real numbers and let~$\mathscr{D} = D_{cdf}([a,b])$. For this choice of~$a,b$ and~$\mathscr{D}$ the functional~$T = A_{\varepsilon}^w$ satisfies Assumption~\ref{as:MAIN} with~$C = |b-b^{\varepsilon}a^{1-\varepsilon}|/|1-\varepsilon| + (b-a)$ in case~$\varepsilon \in (0, 1)$, and with the same constant~$C$ for~$\varepsilon > 1$ in case~$a > 0$ additionally holds. 
%\end{lemma}

\subsection{Poverty measures}\label{sec:povertymeasures}

Poverty indices are typically based on a \emph{poverty line}, i.e., a threshold~$\mathsf{z}$ below which an, e.g., income is classified as ``poor.'' There are two basic approaches to defining~$\mathsf{z}$: The absolute approach considers~$\mathsf{z}$ as fixed (i.e., independent of the underlying income distribution~$F$), whereas the relative approach views~$\mathsf{z} = \mathsf{z}(F)$ as a functional of the ``income distribution''~$F$. In the relative approach, the poverty line adapts to growth or decline of the economy. To make this formal and to give an example, the following poverty line functional combines both approaches (cf.~\cite{kakwani} and~\cite{lambert}, p.139) in taking a convex combination of a fixed amount~$z_0$ and a centrality measure of the underlying income distribution:
\begin{equation}\label{eqn:povline}
\mathsf{z}_{\mathsf{m}, z_0, \delta}(F) = z_0 + \delta (\mathsf{m}(F) - z_0)
\end{equation}
where~$z_0 > 0$,~$0 \leq \delta \leq 1$, and~$\mathsf{m}$ is a location functional that either coincides with the mean functional~$\mu$, or the median functional~$q_{1/2}$. Note in particular that~$\mathsf{z}_{\mathsf{m}, z_0, 0} = z_0$ and~$\mathsf{z}_{\mathsf{m}, z_0, 1} = \mathsf{m}$, i.e., this definition nests both an absolute and a relative approach. Lemma~\ref{lem:povline} in Appendix~\ref{ap:verif} summarizes conditions under which the poverty line functionals in the family~\eqref{eqn:povline} satisfy Assumption~\ref{as:MAIN}. 

The first poverty measure we shall consider is the so-called \emph{headcount ratio}, which is the proportion in a population~$F$ that, according to a given poverty line~$\mathsf{z}$, qualifies as poor:
\begin{equation}
\mathsf{H}_\mathsf{z}(F) = F(\mathsf{z}(F)).
\end{equation}
For the sake of generality, the following lemma establishes conditions under which the headcount ratio satisfies Assumption~\ref{as:MAIN} under high-level conditions concerning the poverty line functional~$\mathsf{z}$. Specific constants and domains for the concrete family of poverty lines defined in Equation~\eqref{eqn:povline} can immediately be obtained with Lemma~\ref{lem:povline} in Appendix~\ref{ap:verif}. An analogous remark applies to the poverty measures introduced further below, and will not be restated.
\begin{lemma}\label{lem:headcount}
Let~$a<b$ be real numbers, and let~$\mathsf{z}: D_{cdf}([a,b]) \to \R$ denote a poverty line functional that satisfies Assumption~\ref{as:MAIN} with~$\mathscr{D}_{\mathsf{z}}$ and constant~$C_{\mathsf{z}}$, say. Let~$s > 0$. Then,~$\mathsf{T} = \mathsf{H}_{\mathsf{z}}$ satisfies Assumption~\ref{as:MAIN} with~$\mathscr{D} = \mathscr{D}_{\mathsf{z}} \cap \mathscr{D}^s([a,b])$ and~$C = C_\mathsf{z} s + 1$.
\end{lemma}

Certain disadvantages of the headcount ratio motivated \cite{sen1976} to introduce a different
family of poverty measures using an axiomatic approach.
We shall now discuss this family in the generalized form of \cite{kakwani1980}. Given a poverty line~$\mathsf{z}$ and a ``sensitivity parameter''~$\kappa \geq 1$, say, each element of this family of poverty indices is written as
\begin{equation}\label{eqn:sen}
\mathsf{P}_{SK}(F; \mathsf{z}, \kappa) =
(\kappa + 1) \int_{[0, \mathsf{z}(F)]} \left[1-\frac{x}{\mathsf{z}(F)} \right] \left[1 - \frac{F(x)}{F(\mathsf{z}(F))}\right]^{\kappa} dF(x),
\end{equation}
with the convention that~$0/0 := 0$. A result discussing conditions under which~$\mathsf{P}_{SK}(F; \mathsf{z}, \kappa)$ satisfies Assumption~\ref{as:MAIN}, and which is again established under high-level assumptions on the poverty line~$\mathsf{z}$, is provided next. Note that in case~$F$ is supported on~$[0, \infty)$, the poverty line in Equation~\eqref{eqn:povline} is greater or equal to~$(1-\delta)z_0$, which is positive unless~$\delta = 1$. 
\begin{lemma}\label{lem:sen}
Let~$a = 0 < b$,~$\kappa \geq 1$, and let~$\mathsf{z}: D_{cdf}([a,b]) \to \R$ denote a poverty line functional that satisfies Assumption~\ref{as:MAIN} with~$\mathscr{D}_{\mathsf{z}}$ and constant~$C_\mathsf{z}$, say. Suppose  further that~$\mathsf{z} \geq z_* > 0$ holds for some real number~$z_*$. Let~$s > 0$. Then~$\mathsf{T} = \mathsf{P}_{SK}(\cdot; \mathsf{z}, \kappa)$ satisfies Assumption~\ref{as:MAIN} with~$\mathscr{D} = \mathscr{D}_{\mathsf{z}} \cap \mathscr{D}^s([a,b])$ and~$C = (\kappa + 1)\{1 + (b z_*^{-2} + 2\kappa s + s) C_{\mathsf{z}} + 4 \kappa\}.$
\end{lemma}
Finally, we consider a family, each element of which can be written as 
\begin{equation}\label{eqn:pfgt}
\mathsf{P}_{FGT}(F; \mathsf{z}, \Lambda) = \int_{[0, \mathsf{z}(F)]} \Lambda\left(1-[x/\mathsf{z}(F)]\right) dF(x),
\end{equation}
where~$\Lambda: [0, 1] \to [0,1]$ is non-decreasing and Lipschitz continuous. This class contains (at least after monotonic transformations), e.g., the measures of \cite{foster84} or \cite{chakravarty1983} as special cases (cf.~\cite{lambert} Chapter~6.3, and also the more recent review in~\cite{Foster2010}). The following result provides conditions under which~$\mathsf{P}_{FGT}$ satisfies Assumption~\ref{as:MAIN}. Again the result is established under high-level assumptions on the poverty line~$\mathsf{z}$.
\begin{lemma}\label{lem:pfgt}
Let~$a = 0 < b$ and let~$\mathsf{z}: D_{cdf}([a,b]) \to \R$ denote a poverty line functional that satisfies Assumption~\ref{as:MAIN} with~$\mathscr{D}_{\mathsf{z}}$ and constant~$C_\mathsf{z}$, say. Suppose  further that~$\mathsf{z} \geq z_* > 0$ holds for some real number~$z_*$, and that~$\Lambda: [0, 1] \to \R$ is non-decreasing, Lipschitz continuous with constant~$C_{\Lambda}$, and satisfies~$\Lambda(0) = 0$. Then,~$\mathsf{T} = \mathsf{P}_{FGT}(\cdot; \mathsf{z}, \Lambda)$ satisfies Assumption~\ref{as:MAIN} with~$\mathscr{D} = \mathscr{D}_{\mathsf{z}}$ and~$C = 
b z_*^{-2} C_{\Lambda} C_{\mathsf{z}} + \Lambda(1)$. 
\end{lemma}
As a direct application of Lemma~\ref{lem:pfgt}, we note that given a poverty line~$\mathsf{z}$ the poverty measure of \cite{foster84} is obtained upon setting~$\Lambda(x) = x^{\alpha}$ for some~$\alpha \geq 0$ in Equation~\eqref{eqn:pfgt}. The conditions in the preceding lemma are satisfied for~$\alpha \geq 1$ (in which case~$C_{\Lambda} = \alpha$). The preceding lemma does not cover the case where~$\alpha = 0$. However, note that the functional corresponding to~$\Lambda(x) = x^{\alpha}$ with~$\alpha = 0$ coincides with the headcount ratio, which is already covered via Lemma~\ref{lem:headcount}. 

\section{Proofs of results in  Section~\ref{app:FUNC}}\label{ap:verif}

\begin{proof}[Proof of Lemma~\ref{lem:schutzindex}:]
Given~$F,G \in D_{cdf}([a,b])$ it holds that~$|\mathsf{S}_{\text{abs}}(F) - \mathsf{S}_{\text{abs}}(G)|$ is not greater than~$1/2$-times
\begin{equation*}
\int_{[a,b]} \left||x-\mu(F)| - |x-\mu(G)|\right|dF(x)  + \left| \int_{[a,b]} |x-\mu(G)| dF(x) - \int_{[a,b]} |x-\mu(G)| dG(x) \right|.
\end{equation*}
Using the reverse triangle inequality, the first integral in the previous display can be bounded from above by~$|\mu(F) - \mu(G)| \leq (b-a) \|F-G\|_{\infty}$ (cf.~Example~\ref{ex:mean} for the inequality). Using Lemma~\ref{lem:U}, the remaining expression to the right in the previous display is seen not to be greater than~$(b-a)\|F-G\|_{\infty}$. Hence, the first statement follows (noting that~$\mathsf{S}_{\text{abs}}$ is obviously well defined on all of~$D_{cdf}([a,b])$).

Concerning the second claim, we first observe that for every~$F \in D_{cdf}([a,b])$ it holds that
\begin{equation}\label{eqn:simpleex}
\frac{1}{2}\int_{[a,b]}|x-\mu(F)|dF(x) = \int_{[a,\mu(F)]}(\mu(F)-x)dF(x).
\end{equation}
Next, let~$s>0$,~$\delta \in (a,b)$,~$F \in \mathscr{D}(s, \delta)$ and~$G \in D_{cdf}([a,b])$. We consider two cases, and start with the case where~$\mu(G) = 0$ (implying that~$a = 0$ and that~$G$ is the cdf corresponding to point mass at~$0$). Then, by convention,~$\mathsf{S}_{\text{rel}}(G) = 0$, and it follows from Equation~\eqref{eqn:simpleex} (recalling that~$\mu(F) \geq \delta > 0$) that
\begin{equation}
|\mathsf{S}_{\text{rel}}(F) - \mathsf{S}_{\text{rel}}(G)| \leq \int_{[a,\mu(F)]}|1-x/\mu(F)|dF(x) \leq F(\mu(F)).
\end{equation}
Since~$F$ is continuous~$0 = F(0) = F(\mu(G))$ holds. It follows that~$F(\mu(F)) = |F(\mu(F))-F(\mu(G))|$. Using the mean-value theorem of \cite{minassian} and Example~\ref{ex:mean} we conclude that~$|F(\mu(F))-F(\mu(G))| \leq s (b-a)\|F-G\|_{\infty}$.

Next, we turn to the case where~$\mu(G) > 0$. First, we note that by~\eqref{eqn:simpleex}
\begin{equation*}
|\mathsf{S}_{\text{rel}}(F) - \mathsf{S}_{\text{rel}}(G)| \leq |F(\mu(F)) - G(\mu(G))| + \left| \int_{[a, \mu(F)]} \frac{x}{\mu(F)} dF(x) - \int_{[a, \mu(G)]} \frac{x}{\mu(G)} dG(x)\right|.
\end{equation*}
Consider the first term in absolute values in the previous display: By the triangle inequality:
\begin{equation*}
|F(\mu(F)) - G(\mu(G))| \leq |F(\mu(F)) - F(\mu(G))| + \|F - G\|_{\infty}.
\end{equation*}
From the mean-value theorem for right-differentiable functions as in~\cite{minassian}, and the definition of~$\mathscr{C}^s([a,b])$, we obtain~$|F(\mu(F)) - F(\mu(G))| \leq s |\mu(F) - \mu(G)| \leq s (b-a)\|F-G\|_{\infty}$, the second inequality following from Example~\ref{ex:mean}. Now, it remains to show that
\begin{equation}\label{eqn:claimSchutz}
\left| \int_{[a, \mu(F)]} \frac{x}{\mu(F)} dF(x) - \int_{[a, \mu(G)]} \frac{x}{\mu(G)} dG(x)\right| \leq ((s+\delta^{-1})(b-a) + 4) \|F-G\|_{\infty}.
\end{equation}
To this end, denote~$m := \min(\mu(F), \mu(G))$,~$M := \max(\mu(F), \mu(G))$, let~$\tilde{F}$ denote a cdf in~$\{F, G\}$ which realizes the latter maximum, and rewrite the difference of integrals inside the absolute value to the left in the preceding display as
\begin{equation*}
\int_{[a, m]} \frac{x}{\mu(F)} dF(x) - \int_{[a, m]} \frac{x}{\mu(F)} dG(x) \pm \int_{(m,M]} \frac{x}{\mu(\tilde{F})} d\tilde{F}(x) + \int_{[a, m]} \left[\frac{x}{\mu(F)} - \frac{x}{\mu(G)}\right]dG(x),
\end{equation*}
where~``$\pm$'' is to be interpreted as~``$+$'' in case~$\tilde{F} = F$ and as~``$-$'' in case~$\tilde{F} = G$. Next, denote the difference of the first two integrals in the previous display by~$A$, the third integral by~$B$ and the fourth by~$D$, respectively. First, Lemma~\ref{lem:U} (applied with~$k = 1$,~$c = a$,~$d = m$ and~$\varphi(x) = x/\mu(F)$) implies (working with the upper bounds~$|M^*|\leq 1$ and~$C \leq 1$ in Lemma~\ref{lem:U} for the special case under consideration) that~$|A|\leq2 \|F-G\|_{\infty}$. Second, note that the integrand in~$B$ is smaller than~$1$, hence
\begin{equation}
|B| \leq \tilde{F}(M) - \tilde{F}(m) \leq F(M) - F(m) + 2\|F-G\|_{\infty} \leq s|\mu(F) - \mu(G)| + 2\|F-G\|_{\infty}
\end{equation}
where we used~$\|\tilde{F} - F\| \leq \|F-G\|_{\infty}$ for the second inequality, and the mean-value theorem of \cite{minassian} for the third. To obtain an upper bound for~$|B|$ we now use Example~\ref{ex:mean} to see that the right hand side in the previous display is not greater than~$[s(b-a) + 2]\|F-G\|_{\infty}$. Concerning~$|D|$ note that (cf.~Example~\ref{ex:mean})
\begin{equation*}
|D| \leq \int_{[a, \mu(G)]} \left|\frac{\mu(G)}{\mu(F)} - 1\right| dG(x) \leq \left|\frac{\mu(G)}{\mu(F)} - 1\right| \leq \delta^{-1} |\mu(G) - \mu(F)| \leq \delta^{-1} (b-a) \|F-G\|_{\infty}.
\end{equation*}
Summarizing,~$$|A| + |B| + |D| \leq ((s+\delta^{-1})(b-a) + 4) \|F-G\|_{\infty},$$ which proves the statement in Equation~\eqref{eqn:claimSchutz}.
\end{proof}

\begin{proof}[Proof of Lemma~\ref{lem:giniindex}:]
The first statement follows from Example~\ref{ex:ginivar}. To prove the statement concerning~$\mathsf{G}_{\text{rel}}$, we first note that~$\mathsf{G}_{\text{rel}}$ is well defined on~$D_{cdf}([a,b])$ (note that~$\mu(F) \leq 0$ implies that~$a = 0$ and that~$\mu_F$ is point mass at~$0$, implying that~$\mathsf{G}_{\text{rel}}(F) = 0$). Let~$F, G \in D_{cdf}([a,b])$, and assume~$\mu(F) \geq \delta$, where~$\delta \in (a,b)$. Consider first the case where~$\mu(G) = 0$. Then,~$\mathsf{G}_{\text{rel}}(G) = 0$ and
\begin{equation}
|\mathsf{G}_{\text{rel}}(F) - \mathsf{G}_{\text{rel}}(G)| = \mathsf{G}_{\text{rel}}(F) \leq \delta^{-1} [\mu(F) - \mu(G)]  \leq \delta^{-1} (b-a) \|F-G\|_{\infty},
\end{equation}
where we used that~$\mathsf{G}_{\mathrm{abs}}(F) \leq \mu(F)$ (just note that~$|x_1-x_2|=(x_1+x_2)-2\min(x_1,x_2)$), and Example~\ref{ex:mean}.

If, on the other hand,~$\mu(G) > 0$ (recall that~$a \geq 0$), we abbreviate~$\varphi(x_1, x_2) = |x_1 - x_2|$, and write
\begin{equation}
|\mathsf{G}_{\text{rel}}(F) - \mathsf{G}_{\text{rel}}(G)| \leq (A + B)/2,
\end{equation}
where
\begin{equation}
A := \delta^{-1} \left|\int_{[a,b]}\int_{[a,b]} \varphi(x_1, x_2) dF(x_1)dF(x_2) - \int_{[a,b]}\int_{[a,b]} \varphi(x_1, x_2) dG(x_1)dG(x_2) \right|,
\end{equation}
which, by Example~\ref{ex:ginivar} is not greater than~$\delta^{-1} 2(b-a)\|F-G\|_{\infty}$, 
and
\begin{equation}
B := \int_{[a,b]}\int_{[a,b]} |(\mu(F)^{-1} - \mu(G)^{-1})\varphi(x_1, x_2)| dG(x_1) dG(x_2) \leq 2 \left|[\mu(G)/\mu(F)] - 1\right|,
\end{equation}
where we used~$\mathsf{G}_{\mathrm{abs}}(G) \leq \mu(G)$. Note that~$\left|[\mu(G)/\mu(F)] - 1\right| \leq \delta^{-1} (b-a)\|F-G\|_{\infty}$ (cf.~the end of the proof of Lemma~\ref{lem:schutzindex}). Hence, in case~$\mu(G) \neq 0$, we obtain that 
\begin{equation}
|\mathsf{G}_{\text{rel}}(F) - \mathsf{G}_{\text{rel}}(G)| \leq 2 \delta^{-1} (b-a)\|F-G\|_{\infty}.
\end{equation}
Together with the first case, this proves the result.

%OLD ARGUMENT:
%We start with the proof concerning~$\mathsf{G}_{\text{rel}}$ (under the additional assumption that~$a>0$):
%We use an argument based on Lemma~\ref{lem:comp} similar to the one in the proof of Lemma~\ref{lem:schutzindex}, but we now set~$T_2(H) = \int_{[a,b]} \int_{[a,b]} |x_1 - x_2| dH(x_1)dH(x_2)$. Note that~$T_2(D_{cdf}([a,b]))$ is a subset of~$[0, (b-a)]$. Example~\ref{ex:ginivar} shows that~$T_2$ satisfies Assumption~\ref{as:MAIN} with~$a,b$ and~$\mathscr{D}$ as in the statement of the present lemma, and with constant~$2(b-a)$. Finally, observe that~$G: [a,b] \times [0, (b-a)] \to \R$ as defined in the proof of Lemma~\ref{lem:schutzindex} is Lipschitz continuous (the domain being equipped with~$\|.\|_1$) with constant~$C = \max(\cdot5 a^{-1}, b-a) \max(\cdot5a^{-2}, 1)$ (use Lemma~\ref{lem:auxLip} with~$c = .5$,~$d = a$,~$e = b$,~$p = 0$ and~$q = (b-a)$). We conclude using Lemma~\ref{lem:comp}.
\end{proof}

\begin{proof}[Proof of Lemma~\ref{lem:linin}:]
The functional in Equation~\eqref{def:linineq} is well defined on~$D_{cdf}([a,b])$, because of Lemma~\ref{lem:L}, and since~$a > 0$ is assumed. Next, we apply Lemma~\ref{lem:suffQ} together with Lemma~\ref{lem:lorenz} to obtain that for every~$u \in [0, 1]$ the functional~$F \mapsto L(F,u)$ satisfies Assumption~\ref{as:MAIN} with~$a,b$ and~$\mathscr{D}$ (as in the statement of the present lemma) and with constant~$a^{-1} ( r^{-1} + (b-a)a^{-1}b)$. The statement immediately follows.
\end{proof}

\begin{proof}[Proof of Lemma~\ref{lem:abslinin}:]
Arguing similarly as in the proof of Lemma~\ref{lem:linin}, the triangle inequality, together with Example~\ref{ex:mean} and Lemma~\ref{lem:lorenz} (which is applicable due to Lemma~\ref{lem:suffQ}) immediately yield the claimed result.
\end{proof}

\begin{proof}[Proof of Lemma~\ref{lem:entropy}:]
We start with the first statement. The functional~$\mathsf{T}$ is obviously everywhere defined on~$D_{cdf}([a,b])$ (in case~$\mu(F) = 0$ it follows that~$a = 0$ and that~$F$ corresponds to point mass~$1$ at~$0$ in which case~$\mathsf{T}(F) = -1/(c(c-1))$, by definition). Next, let~$\delta \in (a,b)$, let~$F \in \mathscr{D}(\delta)$ and let~$G \in D_{cdf}([a,b])$. We consider first the case where~$\mu(G) = 0$ (implying that~$\mathsf{T}(G) = -1/(c(c-1))$ and~$a= 0$). Then
\begin{equation}
|\mathsf{T}(F) - \mathsf{T}(G)| \leq \frac{1}{c|c-1|\delta^{c}} \left|\int_{[a,b]} x^{c} dF(x) - \int_{[a,b]} x^{c} dG(x)\right| \leq \frac{b^{c} - a^{c}  }{c|c-1| \delta^{c}}\|F-G\|_{\infty},
\end{equation}
where we used Example~\ref{ex:pmean} (recall that~$a = 0$) for the last inequality. Next, consider the case where~$\mu(G) > 0$. We note that
\begin{equation}\label{eqn:ntentropy}
\left|\int_{[a,b]} (x/\mu(F))^{c} dF(x) - \int_{[a,b]} (x/\mu(G))^{c} dG(x) \right|
\end{equation}
can be upper bounded by~$A + B$ with 
\begin{equation}
A := \left|\int_{[a,b]} (x/\mu(F))^{c} dF(x) - \int_{[a,b]} (x/\mu(F))^{c} dG(x)\right| \leq \frac{b^{c} - a^{c}}{\delta^{c}}\|F-G\|_{\infty}
\end{equation}
the inequality following from Lemma~\ref{lem:U}, and
\begin{equation}
B := |(1/\mu(F))^{c} - (1/\mu(G))^{c} |\int_{[a,b]} x^{c} dG(x) \leq |(\mu(G)/\mu(F))^{c} - 1|,
\end{equation}
the inequality following from Jensen's inequality (recalling that~$c \in (0, 1)$). It remains to observe that the simple inequality~$|z^c - 1| \leq |z-1|$ for~$z > 0$ implies
\begin{equation}
|(\mu(G)/\mu(F))^{c} - 1| \leq |\mu(G)/\mu(F) - 1| \leq \delta^{-1} (b-a) \|F-G\|_{\infty},
\end{equation}
where the second inequality follows from Example~\ref{ex:mean} together with~$\mu(F) \geq \delta$. Hence, in case~$\mu(G) > 0$ we see that~$$|\mathsf{T}(F) - \mathsf{T}(G)|\leq (c|c-1|)^{-1}\left[ \frac{b^{c} - a^{c}}{\delta^{c}} + \delta^{-1} (b-a) \right]\|F-G\|_{\infty},$$
which proves the first claim.

We now prove the second claim. Since~$a > 0$ holds in this case,~$\mu(G)$ and~$\mu(F)$ cannot be smaller than~$a$. Hence the functional is well defined on all of~$D_{cdf}([a,b])$. Furthermore, the expression in Equation~\eqref{eqn:ntentropy} is not greater than~$A + B$, where~$A$ and~$B$ have been defined above. By Lemma~\ref{lem:U} it holds that~$A$ is not greater than~$\max(a^{-c}, b^{-c})|b^c - a^c| \|F-G\|_{\infty}$. Furthermore,~$B$ is not greater than 
\begin{equation}
\begin{aligned}
\max\left((a/b)^c, (b/a)^c\right) |(\mu(F)/\mu(G))^c-1 | & ~\leq |c| \max\left((a/b)^{2c-1}, (b/a)^{2c-1}\right) |\mu(F)/\mu(G)-1| \\
& ~\leq |c| \max\left((a/b)^{2c-1}, (b/a)^{2c-1}\right)a^{-1}(b-a)\|F-G\|_{\infty},  
\end{aligned}
\end{equation}
the first inequality following from~$|z^c - 1| \leq |c|\max((a/b)^{c-1}, (b/a)^{c-1})|z-1|$ for~$z \in [a/b, b/a]$ (noting that this interval contains~$1$ and recalling that~$c \notin [0, 1]$), and the second inequality following from Example~\ref{ex:mean}. 

We now turn to the last case where~$c \in \{0, 1\}$ (and~$a > 0$ guaranteeing that the functional is then well defined on all of~$D_{cdf}([a,b])$). Let~$F, G \in D_{cdf}([a,b])$, and, without loss of generality, assume~$\mu(F) \leq \mu(G)$. We consider first the case where~$c = 0$. The statement follows after noting that~$|\mathsf{T}(F) - \mathsf{T}(G)|$ is not greater than~$C + D$ with
\begin{equation}
C := \left|\int_{[a,b]} \log(x) dF(x) -\int_{[a,b]} \log(x)dG(x)\right| \leq \log(b/a) \|F-G\|_{\infty},
\end{equation}
the inequality following from Lemma~\ref{lem:U}, and (using Example~\ref{ex:mean})
\begin{equation}\label{eqn:ineqlogmumu}
D := \log(\mu(G)/\mu(F)) \leq \log(1 + \frac{b-a}{a}\|F-G\|_{\infty}) \leq  \frac{b-a}{a}\|F-G\|_{\infty}.
\end{equation}
In case~$c = 1$, set~$f(x) := (x/\mu(F))\log(x/\mu(F))$ and~$g(x) := (x/\mu(G))\log(x/\mu(G))$. Write
\begin{equation}
|\mathsf{T}(F) - \mathsf{T}(G)| \leq \left|\int_{[a,b]}f(x)dF(x) - \int_{[a,b]}f(x)dG(x)\right| + \int_{[a,b]}|f(x) - g(x)|dG(x).
\end{equation}
From Lemma~\ref{lem:U} it follows that the first absolute value in the upper bound is not greater than~$\|F-G\|_{\infty}$ times the total variation of~$f$ on~$[a,b]$, the latter being bounded from above by~$\int_{[a/b,b/a]}|1+\log(x)|dx$. Finally, noting that for every~$x \in [a,b]$ we have
\begin{equation}
\begin{aligned}
|f(x) - g(x)| &\leq |x|\left\{|\mu^{-1}(F) - \mu^{-1}(G)||\log(x/\mu(F))| + \mu^{-1}(G) \left| \log(\mu(F)/\mu(G)) \right|\right\} \\
&\leq\frac{b(b-a)}{a^2}  \left\{\log(b/a) +1\right\} \|F-G\|_{\infty},
\end{aligned}
\end{equation}
where (in addition to~$a > 0$) we used Example~\ref{ex:mean} and~\eqref{eqn:ineqlogmumu}. The final claim follows. 
\end{proof}

\begin{proof}[Proof of Lemma~\ref{lem:atkinson}:]
We start with Part 1: Let~$\delta \in (a,b)$. From the first part of Lemma~\ref{lem:entropy} we obtain that in case~$\varepsilon \in (0, 1)$ the functional~$\eps (\eps -1)\mathsf{E}_{c(\varepsilon)}+1$ satisfies Assumption~\ref{as:MAIN} with~$\mathscr{D} = \mathscr{D}(\delta)$ and constant~$[ \delta^{-c(\varepsilon)}(b^{c(\varepsilon)} - a^{c(\varepsilon)}) + \delta^{-1} (b-a) ]$. It remains to observe that the function
\begin{equation}\label{eqn:auxilfun}
z \mapsto 1-z^{1/c(\varepsilon)}
\end{equation}
is Lipschitz continuous on~$[0, 1]$ with constant~$c(\varepsilon)^{-1}$. The claim then follows from Lemma~\ref{lem:comp} (with~$m = 1$), and the representation in Equation~\eqref{eqn:entropAtk} together with the observation that~$0 \leq \varepsilon (\eps -1) \mathsf{E}_{c(\varepsilon)}(F) +1  \leq 1$ holds for every~$F \in D_{cdf}([a,b])$ as a consequence of Jensen's inequality. 

For Part 2 we argue similarly as in Part 1. From the second part of Lemma~\ref{lem:entropy} we obtain that in case~$\varepsilon \in (1, \infty)$ the functional~$\varepsilon (\eps -1) \mathsf{E}_{c(\varepsilon)} + 1$ satisfies Assumption~\ref{as:MAIN} with~$\mathscr{D} = D_{cdf}([a,b])$ and constant (note that~$2c(\varepsilon)-1 < 0$) equal to
$$[b^{-c(\varepsilon)}(a^{c(\varepsilon)} - b^{c(\varepsilon)}) + |c(\eps)| (a/b)^{2c(\varepsilon)-1}a^{-1}(b-a)].$$
The function in Equation~\eqref{eqn:auxilfun} is Lipschitz continuous on~$[(b/a)^{c(\varepsilon)}, (a/b)^{c(\varepsilon)}]$ with constant~$(\eps -1)^{-1}(b/a)^{\varepsilon}$. From Equation~\eqref{eqn:entropAtk}, and because~$(b/a)^{c(\varepsilon)} \leq \varepsilon (\eps-1) \mathsf{E}_{1-\varepsilon}(F) +1  \leq (a/b)^{c(\varepsilon)}~$ trivially holds for every~$F \in D_{cdf}([a,b])$ ($a > 0$ and~$c(\varepsilon) < 0$), the claim follows from Lemma~\ref{lem:comp} (with~$m = 1$).
\end{proof}

\begin{proof}[Proof of Lemma~\ref{lem:kolm}:]
Clearly~$\mathsf{T}$ is well defined on all of~$D_{cdf}([a,b])$. Let~$F \in D_{cdf}([a,b])$. Then, by Jensen's inequality:
\begin{equation}
\int_{[a,b]} e^{\kappa[\mu(F) - x]}dF(x) \geq 1.
\end{equation}
Since~$x \mapsto \log(x)$ restricted to~$[1, \infty)$ is Lipschitz continuous with constant~$1$, we obtain for any~$G \in D_{cdf}([a,b])$ that $\kappa|\mathsf{T}(F) - \mathsf{T}(G)|$ is bounded from above by
\begin{align*}
&\left|
\int_{[a,b]} e^{\kappa[\mu(F) - x]} dF(x) - \int_{[a,b]} e^{\kappa[\mu(G) - x]} dG(x)
\right|\\
=~&
\left|
(e^{\kappa\mu(F)}-e^{\kappa\mu(G)})\int_{[a,b]} e^{-\kappa x} dF(x) + e^{\kappa\mu(G)}\del[2]{\int_{[a,b]} e^{-\kappa x} dF(x)-\int_{[a,b]} e^{-\kappa x} dG(x)}
\right|\\
\leq~&
\envert[1]{e^{\kappa \mu(F)}-e^{\kappa \mu(G)}}e^{-\kappa a}
+e^{\kappa b}\envert[3]{\int_{[a,b]} e^{-\kappa x} dF(x)-\int_{[a,b]} e^{-\kappa x} dG(x)}\\
\leq~&
\del[2]{\kappa e^{\kappa (b-a)}[b-a]+e^{\kappa b}[e^{-\kappa a}-e^{-\kappa b}]}\enVert[0]{F-G}_\infty,
\end{align*}
where we used Example~\ref{ex:mean} and Lemma~\ref{lem:U} in bounding in each of the summands on the left hand side of the last inequality (as well as the mean-value theorem for the first summand).
\end{proof}

\begin{proof}[Proof of Lemma~\ref{lem:derivedwelfare}:]
Obviously, the welfare function~$\mathsf{W}$ is well defined on~$D_{cdf}([a,b])$ in both parts of the lemma. The first statement follows from the assumptions and Example~\ref{ex:mean}, noting that~$x_1 x_2 - y_1 y_2 = (x_1 - y_1)x_2 - y_1(y_2 - x_2)$ holds for real numbers~$x_i, y_i$,~$i = 1, 2$. The second statement follows directly from the assumptions and Example~\ref{ex:mean}. 
\end{proof}

\begin{lemma}\label{lem:povline}
Let~$a < b$ be real numbers,~$z_0>0$ and~$0 \leq \delta \leq 1$. Then, the following holds:
\begin{enumerate}
\item If~$\delta = 0$, then~$\mathsf{z}_{\mathsf{m},z_0,\delta}$ satisfies Assumption~\ref{as:MAIN} with~$\mathscr{D} = D_{cdf}([a,b])$, and any~$C > 0$.
\item If~$\delta > 0$ and~$\mathsf{m} \equiv \mu(\cdot)$, then~$\mathsf{z}_{\mathsf{m},z_0,\delta}$ satisfies Assumption~\ref{as:MAIN} with~$\mathscr{D} = D_{cdf}([a,b])$ and~$C = \delta(b-a)$.
\item If~$\delta > 0$ and~$\mathsf{m} \equiv q_{1/2}(\cdot)$, then, for every~$r > 0$, the poverty line~$\mathsf{z}_{\mathsf{m},z_0,\delta}$ satisfies Assumption~\ref{as:MAIN} with~$\mathscr{D} = \mathscr{C}_r([a,b])$, and~$C = r^{-1} \delta$.
\end{enumerate}
\end{lemma}

\begin{proof}[Proof of Lemma~\ref{lem:povline}:]
By definition~$\mathsf{z}_{\mathsf{m}, z_0, \delta}(F) = z_0 + \delta (\mathsf{m}(F) - z_0)$. The first statement is trivial; the second follows directly from~Example~\ref{ex:mean}; and the third follows from Lemma~\ref{lem:suffQ} and Example~\ref{ex:median}.
\end{proof}

\begin{proof}[Proof of Lemma~\ref{lem:headcount}:]
Since~$\mathsf{z}$ satisfies Assumption~\ref{as:MAIN} the functional~$\mathsf{z}$ is well defined on all of~$D_{cdf}([a,b])$. Thus~$\mathsf{H}_\mathsf{z}$ is well defined on~$D_{cdf}([a,b])$ as well. Finally, given~$F \in \mathscr{D}$ and~$G \in D_{cdf}([a,b])$, note that by definition and the triangle inequality: 
\begin{equation}
|\mathsf{H}_\mathsf{z}(F) - \mathsf{H}_\mathsf{z}(G)| \leq |F(\mathsf{z}(F)) - F(\mathsf{z}(G))| + \|F - G\|_{\infty} \leq (C_\mathsf{z} s + 1)\|F - G\|_{\infty},
\end{equation}
where we used that~$\mathsf{z}$ satisfies~Assumption~\ref{as:MAIN} together with a mean-value theorem as in \cite{minassian} for the last inequality.
\end{proof}

\begin{proof}[Proof of Lemma~\ref{lem:sen}:]
Obviously,~$\mathsf{P}_{SK}(\cdot; \mathsf{z}, \kappa)$ is well defined on~$D_{cdf}([a,b])$ because~$\mathsf{z} \geq z_* > 0$ holds by assumption, and due to our convention that~$0/0 := 0$ (noting also that~$F(x) = 0$ for every~$x\in [0, \mathsf{z}(F)]$ in case~$F(\mathsf{z}(F)) = 0$). Next, fix~$F \in \mathscr{D}$ and~$G \in D_{cdf}([a,b])$.  Define for all~$x \in \R$~$$f(x) := \max(1-[x/\mathsf{z}(F)], 0) |1-[F(x)/F(\mathsf{z}(F))]|^{\kappa},$$
and analogously
$$g(x) := \max(1-[x/\mathsf{z}(G)], 0) |1-[G(x)/G(\mathsf{z}(G))]|^{\kappa}.$$

Define~$m := \min(\mathsf{z}(F), \mathsf{z}(G))$ and~$M := \max(\mathsf{z}(F), \mathsf{z}(G))$, and the following partition of~$[a,M]$ (using our convention~$0/0 := 0$):
\begin{equation*}
A := \left\{x \in [a,m] : \frac{F(x)}{F(\mathsf{z}(F))} > \frac{G(x)}{G(\mathsf{z}(G)) } \right\}, \quad 
B: = [a,m] \backslash A, \quad \text{ and } D:= (m,M],
\end{equation*}
where~$D = \emptyset$ in case~$m = M$. Next, write
\begin{equation*}
\frac{\mathsf{P}_{SK}(F; \mathsf{z}, \kappa) - \mathsf{P}_{SK}(G; \mathsf{z}, \kappa)}{\kappa + 1} = \int_{[a,M]} [f(x) - g(x)] dF(x) + \left[\int_{[a,b]} g(x) dF(x) - \int_{[a,b]} g(x) dG(x)\right],
\end{equation*}
noting that~$f$ and~$g$ vanish for~$x > M$; and denote the right-hand side by~$S_1 + S_2$,~$S_2$ denoting the term in brackets to the far right. Since~$g([a,b]) \subseteq [0, 1]$ and because~$g$ is right-continuous ($G$ is a cdf) and non-increasing, it hence follows from Lemma~\ref{lem:monoU} that~$|S_2| \leq \|F-G\|_{\infty}$. 

Concerning~$S_1$, note that for every~$x \in [a,m]$ it holds that~$|f(x) - g(x)|$ is not greater than the sum of 
\begin{equation}
\begin{aligned}\label{eqn:up1}
|\max([1-\frac{x}{\mathsf{z}(F)} ], 0) - \max([1-\frac{x}{\mathsf{z}(G)} ], 0)| & \leq x |\mathsf{z}(F)^{-1} - \mathsf{z}(G)^{-1}| \\
&\leq x z_*^{-2} |\mathsf{z}(F) - \mathsf{z}(G)| \leq b z_*^{-2} C_{\mathsf{z}} \|F-G\|_{\infty},
\end{aligned}
\end{equation}
(where we used that~$\mathsf{z} \geq z_*$ to obtain the second inequality, and that~$\mathsf{z}$ satisfies Assumption~\ref{as:MAIN} with constant~$C_{\mathsf{z}}$ to obtain the third) and
\begin{equation}\label{eqn:up2}
\big||1 - [F(x)/F(\mathsf{z}(F))]|^{\kappa} - |1 - [G(x)/G(\mathsf{z}(G))]|^{\kappa}\big| \leq \kappa |[F(x)/F(\mathsf{z}(F))] - [G(x)/G(\mathsf{z}(G))]|
\end{equation}
(where we used~$\kappa \geq 1$, the mean-value theorem, and the reverse triangle inequality to obtain the upper bound). For~$x \in D$, it holds that~$|f(x) - g(x)| \leq 1$.
It hence follows that~$|S_1|$ is bounded from above by the sum of~$b z_*^{-2} C_{\mathbf{z}} \|F-G\|_{\infty}$,~$\kappa$ times 
\begin{equation}\label{eqn:ABCintegral}
\begin{aligned}
&
\int_A [F(x)/F(\mathsf{z}(F))] - [G(x)/G(\mathsf{z}(G))] dF(x) + \int_B [G(x)/G(\mathsf{z}(G))] - [F(x)/F(\mathsf{z}(F))] dF(x)
\end{aligned}
\end{equation}
and~$\int_D dF(x)$.
%The integral in the second line is easily seen to be not greater than (bound the integrand by~$1$)~$|F(\mathsf{z}(F)) - F(\mathsf{z}(G))| \leq sC_{\mathsf{z}}\|F-G\|_{\infty}$, where we used the mean-value theorem of \cite{minassian} and the assumption that~$\mathsf{z}$ satisfies Assumption~\ref{as:MAIN} with constant~$C_{\mathsf{z}}$. 
From the mean-value theorem in \cite{minassian}, and~$\mathsf{z}$ satisfying Assumption~\ref{as:MAIN}, we conclude
\begin{equation}\label{eqn:hbound}
\int_DdF(x) = F(M)-F(m) \leq s C_{\mathsf{z}}\|F-G\|_{\infty}.
\end{equation}
Before bounding the integrals in Equation~\eqref{eqn:ABCintegral}, we recall that Lemma~\ref{lem:headcount} shows that 
\begin{equation}\label{eqn:recallthat}
G(\mathsf{z}(G)) - (C_{\mathsf{z}}s + 1) \|F-G\|_{\infty} \leq F(\mathsf{z}(F)) \leq G(\mathsf{z}(G)) + (C_{\mathsf{z}}s + 1) \|F-G\|_{\infty}.
\end{equation}
To bound the integrals in Equation~\eqref{eqn:ABCintegral}, we now consider different cases:

Consider first the case where~$F(\mathsf{z}(F)) = 0$: Then, the convention~$0/0 := 0$ implies~$A = \emptyset$ and~$B = [a,m]$. Furthermore, the integral over~$B$ in Equation~\eqref{eqn:ABCintegral} vanishes in this case, because~$m \leq \mathsf{z}(F)$ implies~$F(m) = 0$ (and the integrand is non-negative). Hence, the expression in Equation~\eqref{eqn:ABCintegral} is~$0$. 

Next, consider the case where~$G(\mathsf{z}(G)) = 0~$ and~$F(\mathsf{z}(F)) > 0$. It follows from our convention that then~$A = \{x \in [a,m]: F(x)/F(\mathsf{z}(F)) > 0\}$, and that the integral over~$B$ in~\eqref{eqn:ABCintegral}  vanishes. The integral over~$A$ is not greater than~$$F(m) = F(m) - G(\mathsf{z}(G)) \leq F(\mathsf{z}(F)) - G(\mathsf{z}(G)) \leq (C_{\mathsf{z}}s + 1) \|F-G\|_{\infty},$$
where we used Equation~\eqref{eqn:recallthat} to obtain the last inequality. We thus see that in this case the expression in Equation~\eqref{eqn:ABCintegral} does not exceed~$(C_{\mathsf{z}}s + 1)\|F-G\|_{\infty}$.

Finally, consider the case where~$G(\mathsf{z}(G))$ and~$F(\mathsf{z}(F))$ are both positive. Then, we can write the integral over~$A$ in Equation~\eqref{eqn:ABCintegral} as
\begin{align*}
&F(\mathsf{z}(F))^{-1} \int_A F(x) - F(\mathsf{z}(F))\frac{G(x)}{G(\mathsf{z}(G))} dF(x) \\
\leq ~~& 
F(\mathsf{z}(F))^{-1} \int_A [F(x) - G(x)] + (C_{\mathsf{z}}s + 1) \|F-G\|_{\infty}
\frac{G(x)}{G(\mathsf{z}(G))} dF(x) \\
\leq ~~& 
F(\mathsf{z}(F))^{-1} \int_A dF(x) (C_{\mathsf{z}}s+2) \|F-G\|_{\infty}
\leq (C_{\mathsf{z}}s + 2) \|F-G\|_{\infty}, 
\end{align*}
where we used Equation~\eqref{eqn:recallthat} to obtain the first inequality. Similarly, the integral over~$B$ in Equation~\eqref{eqn:ABCintegral} can be shown not to be greater than~$(C_{\mathsf{z}}s + 2) \|F-G\|_{\infty}$. Summarizing, in this last case the expression in Equation~\eqref{eqn:ABCintegral} does not exceed~$[2C_{\mathsf{z}}s + 4]\|F-G\|_{\infty}$. In particular, this bound is bigger than the two bounds in the other two cases. Hence, we conclude that the expression in Equation~\eqref{eqn:ABCintegral} is not greater than~$[2C_{\mathsf{z}}s + 4]\|F-G\|_{\infty}$.

It follows that~$|S_1|$ is bounded from above by 
\begin{equation}
[(b z_*^{-2} + 2\kappa s + s) C_{\mathsf{z}} + 4 \kappa ]\|F-G\|_{\infty}.
\end{equation}
Recalling~$|S_2| \leq \|F-G\|_{\infty}$, it follows that
\begin{equation*}
\frac{\mathsf{P}_{SK}(F; \mathsf{z}, \kappa) - \mathsf{P}_{SK}(G; \mathsf{z}, \kappa)}{\kappa + 1} \leq |S_1| + |S_2| \leq [1 + (b z_*^{-2} + 2\kappa s + s) C_{\mathsf{z}} + 4 \kappa]\|F-G\|_{\infty}.
\end{equation*}
\end{proof}

\begin{proof}[Proof of Lemma~\ref{lem:pfgt}:]
Obviously,~$\mathsf{P}_{FGT}(\cdot; \mathsf{z}, \Lambda)$ is well defined on~$D_{cdf}([a,b])$ because~$\mathsf{z} \geq z_* > 0$ is assumed. Next, fix~$F \in \mathscr{D}$ and~$G \in D_{cdf}([a,b])$. Since~$\Lambda(0) = 0$, we can write
\begin{equation}
\mathsf{P}_{FGT}(F; \mathsf{z}, \Lambda) = \int_{[a,b]} \Lambda( \max( 1 - [x/\mathsf{z}(F)],0 ) ) dF(x).
\end{equation}
Abbreviating~$f(x) := \Lambda( \max( 1 - [x/\mathsf{z}(F)],0 ) )$ and~$g(x) := \Lambda( \max( 1 - [x/\mathsf{z}(G)],0 ) )$, we obtain
\begin{equation*}
\mathsf{P}_{FGT}(F; \mathsf{z}, \Lambda)  - \mathsf{P}_{FGT}(G; \mathsf{z}, \Lambda) = \int_{[a,b]} [f(x)-g(x)]dF(x) + \left[ \int_{[a,b]}g(x)dF(x) - \int_{[a,b]}g(x)dG(x) \right].
\end{equation*}
Denote the first integral on the right by~$A$, and the term in brackets to the far right by~$B$. Because~$g:[a,b] \to [0, \Lambda(1)]$ is continuous and non-increasing, Lemma~\ref{lem:monoU} implies~$|B| \leq \Lambda(1) \|F-G\|_{\infty}$. Concerning~$A$, we use the Lipschitz-continuity of~$\Lambda$, and the inequality~$|\max(1-z_1, 0) - \max(1-z_2, 0)| \leq |z_1 - z_2|$ for nonnegative $z_1, z_2$, to bound
\begin{equation}\label{eq:Aup}
|A| \leq  b C_{\Lambda} |[1/\mathsf{z}(F)] - [1/\mathsf{z}(G)] |\leq  b z_*^{-2} C_{\Lambda} |\mathsf{z}(F) - \mathsf{z}(G)| \leq  b z_*^{-2} C_{\Lambda} C_{\mathsf{z}}\|F-G\|_{\infty},
\end{equation}
where we used the Lipschitz-continuity of the map~$x \mapsto x^{-1}$ on~$[z_*, \infty)$ (with constant~$z_*^{-2}$), and the assumption that~$\mathsf{z}$ satisfies Assumption~\ref{as:MAIN} for obtaining the second inequality. Together with the upper bound on~$|B|$ we obtain the claimed statement.
\end{proof}

\section{General results for establishing Assumption~\ref{as:MAIN}} \label{sec:LCgeneral}

In this appendix, we summarize in a self-contained way a body of techniques that turns out to be useful for establishing Assumption~\ref{as:MAIN} for empirically relevant functionals~$\mathsf{T}$. Once  Assumption~\ref{as:MAIN} is verified for a given functional~$\mathsf{T}$ the Dvoretzky-Kiefer-Wolfowitz-Massart inequality delivers a concentration inequality for~$\mathsf{T}$ of the type
\begin{equation}\label{eqn:appconcineq}
\P(|\mathsf{T}(\hat{F}_n) - \mathsf{T}(F)| > \eps) \leq 2e^{-2n\eps^2/C^2} \quad \text{ for every } \eps > 0;
\end{equation}
(here $\hat{F}_n$ denotes the empirical cdf of an i.i.d.~sample of size~$n$ from the cdf~$F \in \mathscr{D}$), a fact which we heavily use after an optional skipping argument, e.g., in the proofs concerning the finite-sample upper bounds on the F-UCB policy. As already mentioned at the end of Section~\ref{subs:exshort}, due to its simplicity and generality, such a concentration inequality could also be of independent interest for, e.g., constructing uniformly valid confidence intervals in finite samples. 

Applications of the results in the present section to specific functionals were discussed in detail in Appendix~\ref{app:FUNC}. They include inequality measures (cf.~Appendix~\ref{sec:inequalitymeasures}), welfare measures (cf.~Appendix~\ref{sec:welfaremeasures}), and poverty measures (cf.~Appendix~\ref{sec:povertymeasures}). 

The techniques we describe are based on decomposability-properties of the functional, its specific structural (e.g., linearity) properties, and on properties of quantiles and quantile functions, or related quantities such as Lorenz curves. We emphasize that \emph{the results in the present section are elementary, but are difficult to pinpoint in the literature in the form needed.} We start with a short section concerning notation.

\subsection{Notation}\label{sec:notation}

We denote by~$D(\R)$ the Banach space of real-valued bounded c\`adl\`ag functions equipped with the supremum norm~$\|G \|_{\infty} = \sup\{|G(x)|: x\in \R\}$. The closed convex subset of~$D(\R)$ consisting of all cumulative distribution functions (cdfs) shall be denoted by~$D_{cdf}(\R)$. Furthermore, given two real numbers~$a < b$, we define the subset~$D_{cdf}((a,b])$ of~$D_{cdf}(\R)$ as follows:~$F \in D_{cdf}((a,b])$ if and only if~$F \in D_{cdf}(\R)$,~$F(a) = 0$ and~$F(b) = 1$. We also recall the definition of~$D_{cdf}([a,b])$ from Section~\ref{sec:setup}:~$F \in D_{cdf}([a,b])$ if and only if~$F \in D_{cdf}(\R)$,~$F(a-) = 0$ and~$F(b) = 1$. Here~$F(a-)$ denotes the left-sided limit of~$F$ at~$a$. Recall also from the beginning of the appendix of this article that given a cdf~$F$, we denote by~$\mu_F$ the (uniquely defined) probability measure on the Borel sets of~$\R$ that satisfies
\begin{equation*}
\mu_F((-\infty, x]) = F(x) \quad \text{ for every } x \in \R;
\end{equation*}
as usual, we denote the integral of a~$\mu_F$-integrable Borel measurable function~$f: \R \to \R$ by~$\int_{\R} f(x) dF(x) := \int_{\R} f(x) d\mu_F(x)$.

In the following subsections we shall repeatedly encounter \emph{functionals}~$\mathsf{T}$ with a domain~$\mathscr{T} \subseteq D_{cdf}(\R)$, say, and co-domain~$\R$, which are Lipschitz continuous ($\mathscr{T}$ being equipped with the metric induced by the supremum norm on~$D(\R)$): Recall that a functional~$\mathsf{T}: \mathscr{T} \to \R$ is called \emph{Lipschitz continuous} if there exists a nonnegative real number~$C$ such that for every~$F$ and every~$G \in \mathscr{T}$ it holds that
\begin{equation}
|\mathsf{T}(F) - \mathsf{T}(G)| \leq C \|F - G\|_{\infty}.
\end{equation}
In this case, we call~$C$ \emph{a} Lipschitz constant of~$\mathsf{T}$. When we say that a functional~$\mathsf{T}: \mathscr{T} \to \R$ is Lipschitz continuous with constant~$C$, we do not imply that this is the smallest such constant. Recall from Remark~\ref{rem:asym} that if a functional~$\mathsf{T}$ is Lipschitz continuous on~$\mathscr{T} = D_{cdf}([a,b])$ for real numbers~$a<b$, then~$\mathsf{T}$ satisfies Assumption~\ref{as:MAIN} with~$\mathscr{D} = D_{cdf}([a,b])$.

\subsection{Decomposability}

Oftentimes a given functional can be decomposed into a function of several ``simpler'' functionals. It is a straightforward but useful fact that if a functional can be written as a composition of a number of functionals that satisfy Assumption~\ref{as:MAIN} with a Lipschitz continuous function on a suitable intermediating metric space, this composition satisfies Assumption~\ref{as:MAIN} as well. A corresponding result is as follows. 

\begin{lemma}\label{lem:comp}
Let~$a < b$ be real numbers, and let~$\emptyset \neq \mathscr{D} \subseteq D_{cdf}([a,b])$. Let~$m \in \N \cup \{\infty\}$. For every~$i \in \{1, \hdots, m\} \cap \N$ let~$\mathsf{T}_i: D_{cdf}([a,b])\to \R$ satisfy Assumption~\ref{as:MAIN} with~$\mathscr{D}$ and with constant~$C_i$. Denote by~$\bar{C}$ the vector with~$i$-th coordinate~$C_i$, and by~$\overline{\mathsf{T}}$ the vector with~$i$-th coordinate~$\mathsf{T}_i$. Set~$\mathfrak{I} := \{\overline{\mathsf{T}}(F): F \in D_{cdf}([a,b])\} \subseteq \R^m$. Suppose that for~$p \in [1,\infty]$ it holds that~$\|\bar{C}\|_p = (\sum_i |C_i|^p)^{1/p}< \infty$ (where ~$\|\bar{C}\|_{p} := \sup_{i} |C_i|$ in case~$p = \infty$). Then,~$(x, y)\mapsto \|x - y\|_p$ defines a metric on~$\mathfrak{I}$. If the function~$G: \mathfrak{I} \to \R$ is Lipschitz continuous with constant~$C$ (with respect to the just-mentioned metric), then~$\mathsf{T} = G \circ \overline{\mathsf{T}}$ satisfies Assumption~\ref{as:MAIN} with~$\mathscr{D}$ and constant~$C \|\bar{C}\|_p$.
\end{lemma}
\begin{proof}
We first show that~$(x, y)\mapsto \|x - y\|_p$ defines a metric on~$\mathfrak{I} \subseteq \R^m$. To this end, we only verify that~$\|x - y\|_p < \infty$ for every~$x, y \in \mathfrak{I}$; all remaining properties of a metric are trivially satisfied. For every~$x, y \in \mathfrak{I}$ there exist~$F, G \in D_{cdf}([a,b])$ such that~$x_i= \mathsf{T}_i(F)$  and~$y_i= \mathsf{T}_i(H)$ for every~$i$. Fix an arbitrary element~$F^*\in \mathscr{D}$. It then follows from Assumption~\ref{as:MAIN} that~$|x_i-y_i| \leq |\mathsf{T}_i(F) - \mathsf{T}_i(F^*)| + |\mathsf{T}_i(F^*) - \mathsf{T}_i(H)| \leq 2C_i$. Hence,~$\|\bar{C}\|_p < \infty$ implies~$\|x-y\|_p < \infty$. Having established the first claim in the lemma, we move on to the final claim. Let~$F \in \mathscr{D}$ and $H \in D_{cdf}([a,b])$. We have~$|\mathsf{T}(F) - \mathsf{T}(H)| = |G(\overline{\mathsf{T}}(F)) - G(\overline{\mathsf{T}}(H))| \leq C\|\overline{\mathsf{T}}(F) - \overline{\mathsf{T}}(H) \|_p$, the inequality following from Lipschitz continuity of~$G: \mathfrak{I} \to \R$. From the definition of~$\|\cdot\|_p$ and Assumption~\ref{as:MAIN} it immediately follows that~$\|\overline{\mathsf{T}}(F) - \overline{\mathsf{T}}(H) \|_p \leq \|\bar{C}\|_p \|F-H\|_{\infty}$, which proves the lemma.
\end{proof}

\subsection{U-functionals}
%
%The first result establishes a Lipschitz property under a condition slightly weaker than ``linearity and boundedness'' of the functional~$T$:
%%
%\begin{lemma}\label{lem:lin}
%Let~$T: [D^* - D^*] \to \R$ for some nonempty~$D^*\subseteq D_{cdf}(\R)$, and assume that~$T(F) - T(G) = T(F-G)$ and~$\|F-G\|^{-1}_{\infty} T(F-G) = T(\|F-G\|^{-1}_{\infty}[F-G]) \leq C$ for every pair~$F \neq G$ in~$D^*$ for a given constant~$C$. Then~$T$ is Lipschitz continuous with constant~$C$. 	
%\end{lemma}
%%
%\begin{proof}
%For cdfs~$F \neq G$ in~$D_{cdf}([a,b])$ we have~$$|T(F) - T(G)| = |T(F-G)| = \|F-G\|_{\infty} |T(\|F-G\|_{\infty}^{-1}[F-G])| \leq \|F-G\|_{\infty} C.$$
%\end{proof}
%%
%\begin{example}[Evaluating a cdf~at a point]\label{ex:pointeval}
%Let~$x$ be a real number and set~$T(H) = H(x)$ for~$H \in D(\R)$. Clearly, for~$F \neq G$ in~$D_{cdf}(\R)$ it holds that~$T(F) - T(G) = T(F-G)$ and~$\|F-G\|_{\infty}^{-1} T(F-G) = T(\|F-G\|^{-1}_{\infty}[F-G]) \leq 1$. Lemma~\ref{lem:lin} thus shows that the evaluation functional is Lipschitz continuous on~$D_{cdf}(\R)$ with constant~$1$.
%\end{example}

We here consider ``U-functionals'' (the corresponding sample plug-in variants being traditionally referred to as U-statistics, hence the name). The following result covers examples such as moments and certain concentration measures or dependence measures, cf.~Chapter~5 in~\cite{serfling}, and see also the subsequent discussion for examples.
\begin{lemma}\label{lem:U}
Let~$a < b$ be real numbers and let~$\varphi: [a,b]^k \to \R$ for some~$k \in \N$. Suppose that~$\varphi$ is bounded, and is symmetric in the sense that~$\varphi(x_1, \hdots, x_k) = \varphi(x_{\pi_1}, \hdots, x_{\pi_k})$ for every permutation~$x_{\pi_1}, \hdots, x_{\pi_k}$ of~$x_1, \hdots, x_k$. Let~$a \leq c < d \leq b$. Suppose that for every~$x_2^*, \hdots, x^*_k \in [c,d]^{k-1}$ the function~$x \mapsto \varphi(x, x_2^*, \hdots, x^*_k)$ defined on~$[c,d]$ is continuous and has total variation not greater than~$C \in \R$. For~$F \in D_{cdf}([a,b])$ define 
\begin{equation}\label{eqn:SINT}
\mathsf{m}_{\varphi; c,d}(F) := \int_{[c,d]} \hdots \int_{[c,d]} \varphi(x_1, \hdots, x_k) dF(x_1) \hdots dF(x_k),
\end{equation}
which we abbreviate as~$\mathsf{m}_{\varphi}(\cdot)$ in case~$c = a$ and~$d = b$. Then,~$\mathsf{m}_{\varphi; c,d}$ is Lipschitz continuous on~$D_{cdf}([a,b])$ with constant~$k C^*$, where
\begin{equation}
C^* = 
\begin{cases}
C & \text{ if } a = c, b = d \\
C + m^* & \text{ if } b = d \\
C + M^* & \text{ if } a = c \\
C + m^* + M^* & \text{ else}, \\
\end{cases}
\end{equation}
and where
\begin{align}
m^* &:= \sup \{ |\varphi(c, x_2^*, \hdots, x^*_k)| : x_2^*, \hdots, x^*_k \in [c,d]^{k-1}\} \\
M^* &:= \sup \{ |\varphi(d, x_2^*, \hdots, x^*_k)| : x_2^*, \hdots, x^*_k \in [c,d]^{k-1}\}.
\end{align}
\end{lemma}
\begin{proof}
Note first that~$\mathsf{m}_{\varphi; c,d}(F)$ is well defined (i.e.,~$\varphi$ is integrable w.r.t.~the $k$-fold product measure~$\bigotimes_{i = 1}^k \mu_{F}$) on~$D_{cdf}([a,b])$ because~$\varphi$ is bounded and measurable (see, e.g.,~\cite{BURKE200329}). Next, we reduce the statement to the case~$k = 1$: Let~$F, G \in D_{cdf}([a,b])$, let~$\mu$ be a probability measure that dominates~$\mu_F$ and~$\mu_G$, and let~$f$ and~$g$ denote~$\mu$-densities of~$\mu_F$ and~$\mu_G$, respectively. Then,
\begin{equation}\label{eqn:expfun}
\mathsf{m}_{\varphi; c,d}(F) = \int_{[c,d]} \hdots \int_{[c,d]} \varphi(x_1, \hdots, x_k) \prod_{j = 1}^k f(x_j) d\mu(x_1) \hdots d\mu(x_k),
\end{equation}
and an analogous expression (replacing the density~$f$ by the density~$g$) corresponds to~$\mathsf{m}_{\varphi; c,d}(G)$. Recall also that for arbitrary real numbers~$a_j, b_j$ for~$j = 1, \hdots, k$ we may write (e.g., \cite{witting2} Hilfssatz~5.67(a))
\begin{equation}\label{eqn:recallprod}
\prod_{j = 1}^k a_j - \prod_{j = 1}^k b_j = \sum_{j = 1}^k \left[\left(\prod_{i = 1}^{j-1}a_i \right) (a_j - b_j) \prod_{i = j+1}^k b_i \right],
\end{equation}
where empty products are to be interpreted as~$1$. Equipped with~\eqref{eqn:recallprod}, using Equation~\eqref{eqn:expfun}, and Fubini's theorem, we write~$\mathsf{m}_{\varphi; c,d}(F) - \mathsf{m}_{\varphi; c,d}(G)$ as
\begin{small}
\begin{equation*}
\sum_{j = 1}^k \int_{[c,d]} \hdots \int_{[c,d]} \varphi(x_1, \hdots, x_k) [f(x_j) - g(x_j)] d \mu(x_j) dF(x_1) \hdots dF(x_{j-1})dG(x_{j+1}) \hdots dG(x_{k}).
\end{equation*}
\end{small}
Using the triangle inequality to upper bound~$|\mathsf{m}_{\varphi; c,d}(F) - \mathsf{m}_{\varphi; c,d}(G)|$, an application of the symmetry condition shows that it suffices to verify that for~$x_2^*, \hdots, x_k^*$ in~$[c,d]^{k-1}$ arbitrary
\begin{equation}\label{eqn:boundU}
\left|\int_{[c,d]} \varphi(x, x_2^* \hdots, x_k^*) d F(x) - \int_{[c,d]} \varphi(x, x_2^* \hdots, x_k^*) d G(x) \right| \leq 
C^* \|F-G\|_{\infty}.
\end{equation}
Let~$f^*: \R \to \R$ be a continuous function of bounded variation (possibly depending on~$x_2^*, \hdots, x_k^*$) such that~$f^*(x) = \varphi(x, x_2^* \hdots, x_k^*)$ holds for every~$x \in [c,d]$, and such that~$f^*(x) \to 0$ as~$x \to -\infty$.
Integration-by-parts (as in, e.g., Exercise~34.b on~p.108 in \cite{folland}) gives 
\begin{equation}
\int_{[c,d]} \varphi(x, x_2^* \hdots, x_k^*) dF(x) =  \int_{[c,d]} f^*(x) dF(x) = f^*(d) F(d) - f^*(c-)F(c-) - \int_{[c,d]} F(x) df^*(x),
\end{equation}
an analogous statement holding for~$F$ replaced by~$G$. Hence, the quantity to the left in the inequality in~\eqref{eqn:boundU} is seen to be not greater than
\begin{equation}
|f^*(d)||F(d) - G(d)| + |f^*(c)||F(c-) - G(c-)| + \left|\int_{[c,d]} F(x) - G(x) df^*(x)\right|.
\end{equation}
Noting that~$|f^*(d)| \leq M^*$, that~$|f^*(c)| \leq m^*$, that~$|F(d) - G(d)| = 0$ if~$d = b$, that~$|F(c-) - G(c-)| = 0$ if~$a = c$, and furthermore noting that~$|F(d) - G(d)|\leq \|F-G\|_{\infty}$ and~$|F(c-) - G(c-)| \leq \|F-G\|_{\infty}$ always hold,~\eqref{eqn:boundU} follows from~$\left|\int_{[c,d]} F(x) - G(x) df^*(x)\right|  \leq \|F-G\|_{\infty}C$, a consequence of the total variation of~$f^*$ on~$[c,d]$ being not greater than~$C$.
\end{proof}

\begin{example}[Mean]\label{ex:mean}
Let~$a < b$ be real numbers. Let~$k = 1$ and set~$\varphi(x) = x$, i.e., we consider the mean functional~$F \mapsto \mu(F)$, say, defined via
\begin{equation}
F \mapsto \int_{[a,b]} x dF(x).
\end{equation}	
Note that~$\varphi$ is bounded on~$[a,b]$, is trivially symmetric, and~$\varphi$ satisfies the continuity condition in Lemma~\ref{lem:U}. Furthermore, the total variation of~$\varphi$ is~$(b-a)$. As a consequence of Lemma~\ref{lem:U} the functional~$m_{\varphi}$ is thus Lipschitz continuous on~$D_{cdf}([a, b])$ with constant~$(b-a)$.
\end{example}

\begin{example}[Moments]\label{ex:pmean}
For simplicity, let~$a = 0$ and~$b > 0$. Let~$k = 1$ and set~$\varphi(x) = x^p$ for some~$p > 0$, i.e., we consider the~$p$-mean functional
\begin{equation}
F \mapsto \int_{[0,b]} x^p dF(x).
\end{equation}
Note that~$\varphi$ is bounded on~$[a,b]$, is trivially symmetric, and~$\varphi$ satisfies the continuity condition in Lemma~\ref{lem:U}. Furthermore, by monotonicity, the total variation of~$\varphi$ is~$b^p$. As a consequence of Lemma~\ref{lem:U} the functional~$m_{\varphi}$ is thus Lipschitz continuous on~$D_{cdf}([0, b])$ with constant~$b^p$.
\end{example}

\begin{example}[Variance]\label{ex:variance}
Let~$a < b$ be real numbers. Let~$k = 2$ and set~$\varphi(x_1, x_2) = 0.5(x_1 - x_2)^2$, i.e., we consider the variance
\begin{equation}
F \mapsto 0.5 \int_{[a, b]} \int_{[a, b]} (x_1 - x_2)^2 dF(x_1) dF(x_2) = \int_{[a, b]}\left[x_1 - \int_{[a,b]} x_2 dF(x_2)\right]^2dF(x_1).
\end{equation}
Note that~$\varphi$ is bounded on~$[a,b]^2$, is symmetric, and~$\varphi$ satisfies the continuity condition in Lemma~\ref{lem:U}. For every~$x_2 \in [a,b]$ the total variation of~$x \mapsto 0.5(x - x_2)^2$ is~$\int_{[a,b]}|x-x_2|dx \leq (a-b)^2/2$. It follows from Lemma~\ref{lem:U} that the variance functional is Lipschitz continuous with constant~$(a-b)^2$.
\end{example}

\begin{example}[Gini-mean difference]\label{ex:ginivar}
Let~$a < b$ be real numbers, and let~$\varphi(x_1, x_2) = |x_1 - x_2|$. This corresponds to the functional
\begin{equation}
F \mapsto \int_{[a,b]} \int_{[a,b]} |x_1 - x_2| dF(x_1) dF(x_2),
\end{equation}
which constitutes the numerator of the Gini-index defined in Equation~\eqref{eqn:Ginirel} (and equals twice the absolute Gini index~$\mathsf{G}_{\mathrm{abs}}$ defined in Equation~\eqref{eqn:Giniabs}), and is sometimes called the Gini-mean difference or absolute mean difference. Clearly,~$\varphi$ is bounded on~$[a,b]^2$, symmetric, and satisfies the continuity condition in Lemma~\ref{lem:U}. Furthermore, for every~$x_2 \in [a,b]$ the total variation of~$x \mapsto |x-x_2|$ equals~$(b-a)$. It follows from Lemma~\ref{lem:U} that~$m_{\varphi}$ is Lipschitz continuous on~$D_{cdf}([a,b])$ with constant~$2(b-a)$.
\end{example}

The following lemma is sometimes useful, because it avoids the continuity condition of the integrand in Lemma~\ref{lem:U} by working with a right-continuity and monotonicity condition.
\begin{lemma}\label{lem:monoU}
Let~$a<b$ be real numbers and let~$\varphi: [a,b] \to \R$ be right-continuous, and be non-decreasing or non-increasing. Then, the functional 
\begin{equation}
F \mapsto \int_{[a,b]} \varphi(x) dF(x)
\end{equation}
is Lipschitz continuous on~$D_{cdf}([a,b])$ with constant~$|\varphi(b) - \varphi(a)|$.
\end{lemma}

\begin{proof}
Note first that the functional under consideration is well defined on~$D_{cdf}([a,b])$; and that we only need to consider the case where~$\varphi$ is non-decreasing. To this end let~$F, G \in D_{cdf}([a,b])$ and note that, by the transformation theorem, we have
\begin{equation}
\int_{[a,b]} \varphi(x) dF(x) - \int_{[a,b]} \varphi(x) dG(x) = \int_{[\varphi(a),\varphi(b)]} x dF_{\varphi}(x) - \int_{[\varphi(a),\varphi(b)]} x dG_{\varphi}(x),
\end{equation}
where~$F_{\varphi} \in D_{cdf}([\varphi(a),\varphi(b)])$ denotes the cdf~corresponding to the image measure~$\mu_F \circ \varphi$, and~$G_{\varphi}  \in D_{cdf}([\varphi(a),\varphi(b)])$ is defined analogously. An application of Example~\ref{ex:mean} thus shows that 
\begin{equation}
\left|\int_{[a,b]} \varphi(x) dF(x) - \int_{[a,b]} \varphi(x) dG(x) \right| \leq [\varphi(b)-\varphi(a)] \|F_{\varphi} - G_{\varphi}\|_{\infty}.
\end{equation}
It remains to observe that by Lemma~\ref{lem:KSmono} we have~$\|F_{\varphi} - G_{\varphi}\|_{\infty} \leq \|F-G\|_{\infty}$.
\end{proof}

\begin{lemma}\label{lem:KSmono}
Let~$F$ and~$G$ be cdfs, and let~$\varphi: \R \to \R$ be right-continuous, and be non-decreasing. Then~$\|F_{\varphi} - G_{\varphi}\|_{\infty} \leq \|F-G\|_{\infty}$, where~$F_{\varphi}$ denotes the cdf corresponding to the image measure~$\mu_F \circ \varphi$, and~$G_{\varphi}$ is defined analogously.
\end{lemma}

\begin{proof}
First of all, note that~$\|F_{\varphi} - G_{\varphi}\|_{\infty} = \sup_{z \in C(F, G)} |F_{\varphi}(z) - G_{\varphi}(z)|$, where~$C(F, G) \subseteq \R$ is defined as the (dense) subset of points at which both~$F_{\varphi}$ and~$G_{\varphi}$ are continuous. Next, define~$\varphi^-(x) := \inf\{y \in \R: \varphi(y) \geq x\}$, i.e., a generalized inverse of~$\varphi$. Part~(5) of Proposition~1 in \cite{embrechts} shows that for every~$z \in \R$ we have
\begin{equation}\label{eqn:embrechts5}
A(z) := \{x \in \R: \varphi(x) < z\} = \{x \in \R: x < \varphi^-(z)\}.
\end{equation}
Using this expression for~$A(z)$, we can for every~$z \in C(F, G)$ rewrite~$|F_{\varphi}(z) - G_{\varphi}(z)|$ as
%fe
\begin{equation*}
\begin{aligned}
|\mu_{F_{\varphi}}((-\infty, z)) - \mu_{G_{\varphi}}((-\infty, z))| &= |\mu_F(A(z)) - \mu_G(A(z))| \\
&= |\mu_F(\{x \in \R: x < \varphi^-(z)\}) - \mu_G(\{x \in \R: x < \varphi^-(z)\})|.
\end{aligned}
\end{equation*}
On the one hand, the expression to the far right in the previous display equals~$0 \leq \|F-G\|_{\infty}$ in case~$\varphi^-(z) \in \{ -\infty, +\infty\}$. On the other hand, if~$\varphi^-(z) \in \R$, the same expression is seen to equal~$|F(\varphi^-(z)-) - G(\varphi^-(z)-)| \leq \|F-G\|_{\infty}$. Since this argument goes through for every~$z \in C(F, G)$, we are done.
\end{proof}

\subsection{Quantiles, quantile functions, L-functionals, Lorenz curve, and truncation}

In the present subsection we provide some results concerning quantile-based functionals. For~$\alpha \in [0, 1]$ we define the~$\alpha$-quantile of a cdf~$F$ as usual via~$q_{\alpha}(F) = 
\inf\{x \in \R: F(x) \geq \alpha \}$. Note that for~$\alpha = 0$ we have~$q_{\alpha}(F) = -\infty$, and that (by monotonicity) the quantile function~$\alpha \mapsto q_{\alpha}(F)$ is~$\mathcal{B}([0, 1])-\mathcal{B}(\bar{\R})$ measurable. The first result is as follows:

\begin{lemma}\label{lem:quantiles}
Let~$\alpha \in (0, 1]$ and let~$F \in D_{cdf}([a,b])$ for real numbers~$a < b$. Suppose~$F(q_{\alpha}(F)) = \alpha$ and that there exists a positive real number~$r$ such that 
\begin{equation}\label{eqn:growth}
\begin{aligned}
&F(q_{\alpha}(F) - x) - \alpha  \leq
-rx  \quad &&\text{ if } \quad x>0 \text{ and }  q_{\alpha}(F) - x > a, \\
&F(q_{\alpha}(F) + x) - \alpha \geq 
rx  \quad &&\text{ if } \quad x>0 \text{ and } q_{\alpha}(F) + x < b.
\end{aligned}
\end{equation}
Then, for every~$G \in D_{cdf}([a,b])$ it holds that~$|q_{\alpha}(F)-q_{\alpha}(G)| \leq r^{-1} \|F-G\|_{\infty}$. Consequently, denoting by~$\mathscr{D}$ the set of all cdfs that satisfy the conditions imposed on~$F$ above, it follows that~$q_{\alpha}$ satisfies Assumption~\ref{as:MAIN} with~$\mathscr{D}$ and constant~$C = r^{-1}$.
\end{lemma}

\begin{proof}
To prove the first statement, we may impose the additional assumption that the inequalities to the left in Equation~\eqref{eqn:growth} hold strictly for all~$x$ in the considered ranges (to see this, just observe that Equation~\eqref{eqn:growth} implies the just mentioned strict version for all~$0 < r_* < r$, which can then be used to take care of situations where the additional assumption is not satisfied). 

Given this additional assumption, let~$G$ be an element of~$D_{cdf}([a,b])$. The claimed inequality is trivial if~$F = G$. Thus, we assume that~$F \neq G$. Note that~$F(x) = 0 < \alpha$ for every~$x < a$, and~$F(x) = 1 \geq \alpha$ for every~$x \geq b$ implies~$q_{\alpha}(F) \in [a,b]$; and that, by the same reasoning,~$q_{\alpha}(G) \in [a,b]$. 

We first show that~$q_{\alpha}(G) \geq q_{\alpha}(F) -  r^{-1} \|G-F\|_{\infty}$: On the one hand, if~$q_{\alpha}(F) -  r^{-1} \|G-F\|_{\infty} \leq a$, then~$q_{\alpha}(G) \geq q_{\alpha}(F) -  r^{-1} \|G-F\|_{\infty}$ trivially holds. If, on the other hand,~$q_{\alpha}(F) -  r^{-1} \|G-F\|_{\infty} > a$, then, from the (strict) inequality in the first line of~\eqref{eqn:growth} with~$x =  r^{-1} \|G-F\|_{\infty}$, one obtains~$\alpha > F(q_{\alpha}(F) - r^{-1} \|G-F\|_{\infty}) + \|G-F\|_{\infty}$, thus~$\alpha > G(q_{\alpha}(F) - r^{-1} \|G-F\|_{\infty})$ and hence, again,~$q_{\alpha}(G) \geq q_{\alpha}(F) -  r^{-1} \|G-F\|_{\infty}$. 

We next show that~$q_{\alpha}(G) \leq q_{\alpha}(F) + r^{-1}\|G-F\|_{\infty}$: On the one hand, if~$q_{\alpha}(F) +  r^{-1} \|G-F\|_{\infty} \geq b$, then~$q_{\alpha}(G) \leq q_{\alpha}(F)+ r^{-1} \|G-F\|_{\infty}$ trivially holds. If, on the other hand,~$q_{\alpha}(F) +  r^{-1} \|G-F\|_{\infty} < b$, then the second line in~\eqref{eqn:growth} with~$x =  r^{-1} \|G-F\|_{\infty}$ shows that~$F(q_{\alpha}(F) + r^{-1}\|G-F\|_{\infty}) - \|G-F\|_{\infty} \geq \alpha$, thus~$G(q_{\alpha}(F) + r^{-1}\|G-F\|_{\infty}) \geq \alpha$, and hence, again,~$q_{\alpha}(G) \leq q_{\alpha}(F) + r^{-1}\|G-F\|_{\infty}$. Summarizing yields~$|q_{\alpha}(F)-q_{\alpha}(G)| \leq r^{-1} \|F-G\|_{\infty}$. The last statement is trivial.
\end{proof}

\begin{example}[Median]\label{ex:median}
The median of a distribution~$F$ is defined as its~$\alpha = 1/2$ quantile~$q_{1/2}(F)$. Let~$a < b$ and~$r > 0$ be real numbers, and denote by~$\mathscr{D}$ the set of cdfs~$F$ such that~$F(q_{1/2}(F)) = 1/2$, and such that Equation~\eqref{eqn:growth} is satisfied for~$\alpha = 1/2$ (Lemma~\ref{lem:suffQ} provides a sufficient condition for~$F \in \mathscr{D}$). Then, the functional~$F \mapsto q_{1/2}(F)$ satisfies Assumption~\ref{as:MAIN} with~$a,b$ and~$\mathscr{D}$ with constant~$C = r^{-1}$.
\end{example}

The second result is auxiliary, and concerns not a single quantile, but the whole quantile \emph{function}~$F \mapsto q_{.}(F)$ over closed subintervals of~$(0,1]$. It follows immediately from Lemma~\ref{lem:quantiles}.

\begin{lemma}\label{lem:Qfunction}
Let~$F \in D_{cdf}([a,b])$ for real numbers~$a < b$, and let~$\alpha_* < \alpha^*$ for~$\alpha_*$ and~$\alpha^*$ in~$(0, 1]$. Suppose~$F(q_{\alpha}(F)) = \alpha$ holds for every~$\alpha \in [\alpha_*, \alpha^*]$, and that there exists a positive real number~$r$ so that  Equation~\eqref{eqn:growth} is satisfied for every~$\alpha \in [\alpha_*, \alpha^*]$. Then, for every~$G \in D_{cdf}([a,b])$ it holds that~$$\sup_{\alpha \in [\alpha_*, \alpha^*]}|q_{\alpha}(F)-q_{\alpha}(G)| \leq r^{-1} \|F-G\|_{\infty}.$$
\end{lemma}

A simple sufficient condition for the assumption on~$F$ in Lemma~\ref{lem:Qfunction} (and hence also for the assumption on~$F$ in Lemma~\ref{lem:quantiles}) is that~$F$ admits a density that is bounded from below (on the support of~$F$):

\begin{lemma}\label{lem:suffQ}
Let~$a < b$ be real numbers and let~$F \in D_{cdf}([a,b])$. Suppose~$F$ is continuous, and is right-sided differentiable on~$(a,b)$ with right-sided derivative~$F^+$, which furthermore satisfies~$F^+(x) \geq r$ for every~$x \in (a,b)$ for some~$r > 0$. Then,~$F(q_{\alpha}(F)) = \alpha$ and Equation~\eqref{eqn:growth} holds for every~$\alpha \in (0, 1]$.
\end{lemma}
\begin{proof}
The condition~$F^+(x) \geq r$ for every~$x \in (a,b)$ for an~$r > 0$ implies that~$F$ is strictly increasing on~$[a,b]$, which (together with continuity of~$F$) implies~$F(q_{\alpha}(F)) = \alpha$ for every~$\alpha \in (0, 1]$. The second claim follows from the mean-value theorem for right-differentiable functions in \cite{minassian} (noting that~$q_{\alpha}(F) \in [a,b]$ for every~$\alpha \in (0, 1]$, cf.~the proof of Lemma~\ref{lem:quantiles}).
\end{proof}

The next result, which essentially follows from the previous one, concerns population versions of generalized L-statistics introduced by \cite{serflinggen} (cf.~his Section 2), i.e., L-functionals. 

\begin{lemma}\label{lem:L}
Let~$\nu$ be a measure on the Borel sets of~$[0, 1]$, and let~$J: [0, 1] \to \R$ be such that~$\int_{[0,1]} |J(\alpha)| d\nu(\alpha) = c < \infty$. Assume further that~$\nu(\{0\}) = 0$. Let~$a < b$ be real numbers and define on~$D_{cdf}([a,b])$ the functional
\begin{equation}\label{eqn:defLstat}
\mathsf{T}(F) = \int_{[0, 1]} q_{\alpha}(F) J(\alpha) d\nu(\alpha).
\end{equation}
Let~$F \in D_{cdf}([a,b])$ satisfy~$F(q_{\alpha}(F)) = \alpha$ for every~$\alpha \in (0, 1]$, and suppose there is a positive real number~$r$ such that  Equation~\eqref{eqn:growth} holds for every~$\alpha \in (0, 1]$. Then, for every~$G \in D_{cdf}([a,b])$, it holds that~$$|\mathsf{T}(F) - \mathsf{T}(G)| \leq \frac{c}{r}\|F - G\|_{\infty}.$$
Consequently, denoting by~$\mathscr{D}$ the set of all cdfs that satisfy the conditions imposed on~$F$ above, it follows that~$\mathsf{T}$ defined in Equation~\eqref{eqn:defLstat} satisfies Assumption~\ref{as:MAIN} with~$\mathscr{D}$ and constant~$C = c/r$.
\end{lemma}

\begin{proof}
That~$\int_{[0, 1]} q_{\alpha}(F) J(\alpha) d\nu(\alpha)$ exists for every~$F \in D_{cdf}([a,b])$ follows from~$\nu(\{0\}) = 0$, from~$q_{\alpha}(F) \in [a,b]$ for every~$\alpha \in (0, 1]$ (cf.~the proof of Lemma~\ref{lem:quantiles}), and from the integrability condition on~$J$. Next, for~$F$ and~$G$ as in the statement of the lemma, note that
\begin{equation}\label{eqn:inttarg}
|\mathsf{T}(F) - \mathsf{T}(G)| \leq \int_{(0,1]} |q_{\alpha}(F) - q_{\alpha}(G)||J(\alpha)|d\nu(\alpha).
\end{equation}
Note that the function~$\alpha \mapsto |q_{\alpha}(F) - q_{\alpha}(G)|$ is bounded on~$(0,1]$. By the monotone convergence theorem, for~$\varepsilon \searrow 0$ the integral~$\int_{[\varepsilon,1]} |q_{\alpha}(F) - q_{\alpha}(G)||J(\alpha)|d\nu(\alpha)$
converges to the integral in~\eqref{eqn:inttarg}. But~$\int_{[\varepsilon,1]} |q_{\alpha}(F) - q_{\alpha}(G)||J(\alpha)|d\nu(\alpha) \leq r^{-1}c \|F-G\|_{\infty}$ by Lemma~\ref{lem:Qfunction}. The last statement in the lemma is trivial.
\end{proof}

One particularly important application of Lemma~\ref{lem:L} concerns the so-called Lorenz curve associated with a cdf~$F$ (cf.~\cite{gastwirth}). 

\begin{lemma}\label{lem:lorenz}
Let~$a < b$ be real numbers and define on~$D_{cdf}([a,b])$ the family of functionals indexed by~$u \in [0,1]$ and defined by
\begin{equation}
Q(F, u) := \int_{[0, u]} q_{\alpha}(F)d\alpha;
\end{equation}
furthermore, if~$a > 0$, define the family of functionals indexed by~$u \in [0,1]$ via
\begin{equation}\label{def:Lorenz}
L(F, u) := \mu(F)^{-1} \int_{[0, u]} q_{\alpha}(F)d\alpha 
\end{equation}
Let~$F \in D_{cdf}([a,b])$ satisfy~$F(q_{\alpha}(F)) = \alpha$ for every~$\alpha \in (0, 1]$, and suppose there is a positive real number~$r$ such that  Equation~\eqref{eqn:growth} holds for every~$\alpha \in (0, 1]$. Then, for every~$G \in D_{cdf}([a,b])$ it holds that~
\begin{equation}\label{eqn:firstLo}
|Q(F, u) - Q(G,u)| \leq r^{-1} u \|F-G\|_{\infty} \leq r^{-1} \|F-G\|_{\infty}.
\end{equation}
Consequently, denoting by~$\mathscr{D}$ the set of all cdfs that satisfy the conditions imposed on~$F$ above, it follows that~$\mathsf{T}(\cdot) = Q(\cdot,u)$ satisfies Assumption~\ref{as:MAIN} with~$\mathscr{D}$ and constant~$C = r^{-1} u$. Furthermore, if~$a > 0$, then~$$|L(F, u) - L(G,u)| \leq a^{-1} ( r^{-1}  + (b-a)a^{-1}b ) u \|F-G\|_{\infty},$$
and it follows that~$\mathsf{T}(\cdot) = L(\cdot,u)$ satisfies Assumption~\ref{as:MAIN} with~$\mathscr{D}$ and constant~$C = a^{-1} ( r^{-1}  + (b-a)a^{-1}b ) u$.
\end{lemma}

\begin{proof}
For the claim in Equation~\eqref{eqn:firstLo} we just apply Lemma~\ref{lem:L} with~$\nu$ equal to  Lebesgue measure,~$J = \mathds{1}_{[0, u]}$, which satisfies the integrability condition with~$c = u \leq 1$. For the second claim, note that~$L(\cdot, u)$ is well defined on~$D_{cdf}([a,b])$ because~$a > 0$. Next, observe that for~$F$ and~$G$ as in the statement of the lemma we can bound~$|L(F, u) - L(G, u)|$ from above by
\begin{equation}\label{eqn:uplor}
\mu(F)^{-1} \left\{\left|Q(F, u) - Q(G, u) \right|  + 
|1-\mu(F)/\mu(G)| \int_{[0,u]} q_{\alpha}(G)d\alpha \right\}.
\end{equation}
Since~$\mu(G)$ and~$\mu(F)$ are not smaller than~$a$, since~$q_{\alpha}(G) \leq b$ for~$\alpha \in (0, u]$, and because we already know that~$$\left|Q(F, u) - Q(G, u) \right| \leq r^{-1} u\|F-G\|_{\infty},$$ it remains to observe that by Example~\ref{ex:mean}
\begin{equation}
|1-(\mu(F)/\mu(G))|\leq  (b-a)\|F-G\|_{\infty}/\mu(G) \leq (b-a)a^{-1} \|F-G\|_{\infty}
\end{equation}
to conclude that the expression in~\eqref{eqn:uplor} is not greater than~$a^{-1} \left\{ r^{-1}  + (b-a) a^{-1}b  \right\} u\|F-G\|_{\infty}$.
\end{proof}

The final result in this section concerns trimmed generalized-mean functionals. We consider one-sidedly trimmed functionals, the trimming affecting the lower or upper tail. Two-sided trimming can be dealt with similarly. We abstain from spelling out the details.

\begin{lemma}\label{lem:trimmedV1}
Let~$a < b$ be real numbers, let~$\varphi: \R \to \R$, let~$\varphi$ restricted to~$[a,b]$ be continuous, let the total variation of~$\varphi$ on~$[a,b]$ be not greater than~$C$, and let~$|\varphi(x)| \leq u$ hold for all~$x \in [a,b]$. Furthermore, let~$\alpha \in (0, 1)$. For~$F \in D_{cdf}([a,b])$ define
\begin{equation}\label{eqn:functinttrimmed}
\mathsf{m}^{t-}_{\varphi; \alpha}(F) := \int_{[a, q_{\alpha}(F)]} \varphi(x) dF(x) \quad \text{ and } \quad \mathsf{m}^{t+}_{\varphi; \alpha}(F) := \int_{[ q_{\alpha}(F),b]} \varphi(x) dF(x).
\end{equation}
Let~$F \in D_{cdf}([a,b])$, assume that~$F$ is continuous, and right-sided differentiable on~$(a,b)$, with right-sided derivative~$F^+$ satisfying~$r \leq F^+(x) \leq \kappa$ for every~$x \in (a,b)$, and for positive real numbers~$\kappa$ and~$r$. Then, for every~$G \in D_{cdf}([a,b])$ it holds that 
\begin{equation*}
|\mathsf{m}^{t-}_{\varphi; \alpha}(F) - \mathsf{m}^{t-}_{\varphi; \alpha}(G)| \leq [C + u(1+\kappa r^{-1})]\|F-G\|_{\infty},
\end{equation*}
and 
\begin{equation}
|\mathsf{m}^{t+}_{\varphi; \alpha}(F) - \mathsf{m}^{t+}_{\varphi; \alpha}(G)| \leq [C + u(1+\kappa r^{-1})]\|F-G\|_{\infty},	
\end{equation}
Consequently, denoting by~$\mathscr{D}$ the set of all cdfs that satisfy the conditions imposed on~$F$ above, it follows that~$m^{t-}_{\varphi; \alpha}$ and~$m^{t+}_{\varphi; \alpha}$ satisfy Assumption~\ref{as:MAIN} with~$\mathscr{D}$ and constant~$C + u(1+\kappa r^{-1})$.
\end{lemma}

\begin{proof}
We only provide an argument for the first claimed inequality, the second is obtained analogously. Furthermore, throughout the proof we write~$\mathsf{m}_{\varphi; \alpha}^{t}$ instead of~$\mathsf{m}_{\varphi; \alpha}^{t-}$. First, note that the functional~$\mathsf{m}^t_{\varphi; \alpha}(F)$ is indeed well defined for every~$F \in D_{cdf}([a,b])$. This follows from~$q_{\alpha}(F) \in [a,b]$ (cf.~the proof of Lemma~\ref{lem:quantiles}), and since~$\varphi$ is bounded on~$[a,b]$. Next, let~$F$ be as in the statement of the lemma and satisfy the conditions imposed. Let~$G \in D_{cdf}([a,b])$, implying that~$q_{\alpha}(G) \in [a,b]$. By the triangle inequality,~$|\mathsf{m}^t_{\varphi; \alpha}(F) - \mathsf{m}^t_{\varphi; \alpha}(G)| \leq A + B$, where (using the notation introduced in Equation~\eqref{eqn:SINT})
\begin{equation*}
A:= \left| \mathsf{m}_{\varphi; a, q_{\alpha}(G)} (F) - \mathsf{m}_{\varphi; a, q_{\alpha}(G)} (G)\right| \leq (C + u)\|F-G\|_{\infty},
\end{equation*}
the upper bound following from Lemma~\ref{lem:U}, and 
\begin{equation}\label{eqn:lastup}
B:=  \int  g(x) |\varphi(x)| dF(x) \leq u \int g(x) dF(x),
\end{equation}
where~$g(x) = \left|\mathds{1}_{[a, q_{\alpha}(F)]}(x) - \mathds{1}_{[a, q_{\alpha}(G)]}(x) \right|$. By continuity of~$F$: 
\begin{equation}
\int g(x) dF(x) \leq |F(q_{\alpha}(G)) - F(q_{\alpha}(F))|.
\end{equation}
which, by the assumed behavior of the right-derivative of~$F$ and a mean-value theorem for right-differentiable functions (for example the one by~\cite{minassian}), is not greater than
\begin{equation}
\kappa |q_{\alpha}(G)- q_{\alpha}(F)| \leq \kappa r^{-1} \|F-G\|_{\infty}
\end{equation}
the last inequality following from Lemma~\ref{lem:quantiles} together with Lemma~\ref{lem:suffQ}. This proves the claim. The last statement is trivial.
\end{proof}

\end{document}